\PassOptionsToPackage{svgnames}{xcolor}
\documentclass[12pt]{article}

\usepackage{amssymb, amsthm, amsmath, amsfonts, eurosym, geometry, ulem, graphicx, caption, color, xcolor, setspace, sectsty, comment, footmisc, caption, pdflscape, subfigure, array}

\usepackage[hidelinks]{hyperref}

\usepackage{multirow}
\usepackage{tabularx}
\usepackage{arydshln}

\usepackage{mathtools}
\usepackage[mathscr]{euscript}

\usepackage[svgnames]{xcolor}

\usepackage{soulutf8}

\usepackage[disable]{todonotes}

\usepackage[utf8]{inputenc}

\usepackage{multirow} 

\usepackage[style=authoryear, giveninits=true, date=year, uniquename=false, uniquelist=false, maxcitenames=2, maxbibnames=30, backend=biber]{biblatex}

\AtEveryBibitem{%
  \clearfield{note} 
}

\renewbibmacro{in:}{}
\DeclareNameAlias{default}{last-first}
\AtEveryBibitem{%
  \ifentrytype{article}{
    \clearfield{issn}%
    \clearfield{url}%
    \clearfield{urlyear}%
    \clearfield{urldate}%
    \clearfield{doi}%
  }{}
}

\usepackage{contour}
\usepackage{ulem}

\contourlength{0.8pt}
\newcommand{\myuline}[1]{%
  \uline{\phantom{#1}}%
  \llap{\contour{white}{#1}}%
}

\newtheorem{theorem}{Theorem}
\newtheorem{corollary}{Corollary}
\newtheorem{definition}[theorem]{Definition}
\newtheorem{proposition}{Proposition}
\newtheorem{lemma}{Lemma}

\newcolumntype{L}[1]{>{\raggedright\let\newline\\arraybackslash\hspace{0pt}}m{#1}}
\newcolumntype{C}[1]{>{\centering\let\newline\\arraybackslash\hspace{0pt}}m{#1}}
\newcolumntype{R}[1]{>{\raggedleft\let\newline\\arraybackslash\hspace{0pt}}m{#1}}

\usepackage{ragged2e}
\usepackage{array, booktabs}
\usepackage[referable]{threeparttablex}
\usepackage[column=O]{cellspace}
\usepackage{siunitx}

\usepackage{setspace}

\begin{document}

\begin{titlepage}
\title{Optimal Carbon Prices in an Unequal World: \\
The Role of Regional Welfare Weights} 
 \author{Simon F. Lang\thanks{
ETH Zurich, Department of Management, Technology and Economics (simon.lang@mtec.ethz.ch). \\
Earlier versions of this paper circulated as part of a longer draft that included a companion analysis on international climate finance.
I am grateful to Kenneth Gillingham, William Nordhaus, Matthew Kotchen, David Anthoff, and Narasimha Rao for invaluable feedback and support. I also thank Eli Fenichel, François Salanié, Francis Dennig, Frikk Nesje, Moritz Drupp, Antony Millner, Itai Sher, Maya Eden, Paul Kelleher, Mirco Dinelli, Andrew Vogt, Matthew Gordon, Ethan Addicott, Matthew Ashenfarb, Andie Creel, Kenneth Jung, Jillian Stallman, Lisa Rennels, Enrico Salonia, Eric Donald, Costanza Tomaselli, 
and participants of the SSE Economics of Inequality and the Environment Workshop, the EEA Congress, the CEPR Climate Change and the Environment Symposium, the EAERE Annual Conference, the AERE Summer Conference, the Toulouse Conference on the Economics of Energy and Climate, SURED, the Swiss Workshop on Environmental, Resource and Energy Economics, the Oxford Global Priorities Fellowship, the Columbia IPWSD Workshop, the IAMC Annual Meeting, the Penn State Science and Values in Climate Risk Management webinar, as well as seminar audiences at Yale University, ETH Zurich, the University of Gothenburg, Aarhus University, the Nova School of Business and Economics, and Middlebury College, for valuable comments and discussions.}}

\date{May 6, 2026}
\maketitle
\vspace{-10pt}
\begin{abstract}
\noindent
How should nations price carbon? This paper examines how the treatment of global inequality, captured by regional welfare weights, affects optimal carbon prices. I develop theory to identify the conditions under which accounting for differences in marginal utilities of consumption across countries leads to more stringent global climate policy in the absence of international transfers. 
I further establish a connection between the optimal uniform carbon prices implied by different welfare weights and heterogeneous regional preferences over climate policy stringency.
In calibrated simulations, I find that accounting for global inequality reduces optimal global emissions relative to an inequality-insensitive benchmark. This holds both when carbon prices are regionally differentiated, with emissions 21\% lower, and when they are constrained to be globally uniform, with the uniform carbon price 15\% higher. \\

\vspace{-15pt}
\singlespacing

\vspace{-10pt}
\vspace{7pt}
\normalsize

\end{abstract}

\setcounter{page}{0}
\thispagestyle{empty}
\end{titlepage}
\pagebreak \newpage


\onehalfspacing

\setlength{\abovedisplayskip}{8pt}  
\setlength{\belowdisplayskip}{8pt}  

\section{Introduction} \label{sec:introduction}

The distributional effects of climate change and climate policies are at the heart of international climate change negotiations. Central to these debates are inequalities in climate impacts, responsibilities for emissions, and capabilities to mitigate and adapt---aspects that are all interlinked with global wealth inequality \parencite{chancel_climate_2023}. These inequalities are recognized in international climate agreements, as exemplified by the principle of “common but differentiated responsibilities and respective capabilities” of the United Nations Framework Convention on Climate Change \parencite{unfccc_united_1992}. The Paris Agreement further emphasizes that developed countries should take the lead in reducing emissions,
stressing equity considerations \parencite{unfccc_report_2015}. 

Despite the importance of inequalities 
in international climate policy, the standard approach to estimating optimal carbon prices concentrates solely on efficiency and effectively disregards global inequality 
\parencite{nordhaus_economic_2010, yang_magnitude_2006}. It does so by maximizing a social welfare function (SWF) with Negishi welfare weights, which are higher for wealthier countries, offsetting differences in marginal utilities of consumption across countries. In contrast, an alternative approach focuses on maximizing global welfare, subject to constraints on international transfers \parencite{budolfson_utilitarian_2021}. A common version of this approach maximizes the utilitarian SWF, which assigns equal weight to the welfare of all individuals. Crucially, the utilitarian SWF accounts for global inequality in that it considers differences in marginal utilities of consumption across wealthier and poorer countries.
These differences in welfare weights may be particularly important, as the costs and benefits of emission reductions are unevenly distributed, with poorer countries often disproportionately impacted by climate change \parencite{burke_global_2015, carleton_valuing_2022, kalkuhl_impact_2020, mejean_climate_2024}. 
Given that optimal carbon prices are well-known to be highly sensitive to the utility discount rate \parencite{stern_economics_2007, nordhaus_review_2007}---a form of temporal welfare weighting---it may be surprising that comparatively little research has explored the role of regional welfare weights.

In this paper, I examine how the stance on global inequality---reflected in the choice of regional welfare weights---affects optimal carbon prices. 
I study this question theoretically, to uncover insights into the underlying economic forces, and through numerical simulations with the integrated assessment model RICE \parencite{nordhaus_economic_2010}, to evaluate the quantitative implications for climate policy. 

Specifically, I compare optimal carbon prices under the Negishi-weighted SWF to those under the utilitarian SWF, in the absence of international transfers\footnote{
In a companion paper (in progress), I examine how the availability of international transfers affects optimal carbon prices.
}. 
I begin by imposing the same two constraints on the utilitarian optimization that are implicit in the Negishi solution: no international transfers and uniform carbon prices. This constrains the utilitarian problem to an identical policy instrument---a globally uniform carbon price in each period---enabling a direct comparison with the Negishi solution. Next, I remove the uniform carbon price constraint, allowing for differentiated carbon prices, which can improve utilitarian welfare by shifting some of the abatement cost burden from poorer to wealthier countries. I refer to carbon prices in the utilitarian solution as welfare-maximizing carbon prices to highlight that they maximize the (unweighted or equally-weighted) sum of individuals’ utilities.

Using a theoretical model, I show that optimal carbon prices and aggregate abatement may be higher or lower in the utilitarian solutions than in the Negishi solution and that this depends on the distribution of the marginal climate damages and the abatement cost burden across countries. 
Specifically, the utilitarian climate policy is more stringent if poorer nations have comparatively high marginal climate damages and a relatively steep marginal abatement cost function, resulting in smaller changes in abatement when carbon prices are altered.
In a dynamic extension of the model, I show that regional differences in population growth and economic growth critically influence how regional welfare weights affect optimal carbon prices by impacting the relative importance that regions place on the future, when most climate impacts occur, versus the present, when abatement efforts take place.
Moreover, I establish a novel and intuitive connection between the uniform carbon prices under utilitarian and Negishi weights and nations' preferred uniform carbon prices that maximize national welfare, a notion that was established by \textcite{weitzman_can_2014} and \textcite{kotchen_which_2018}: The utilitarian uniform carbon price exceeds the Negishi-weighted carbon price if and only if poorer nations prefer higher uniform carbon prices than wealthier nations.
These uniform prices can be expressed as weighted averages of nations' preferred prices; under additional structure, the weights take a simple form---endowment weights under Negishi and approximately population weights under utilitarianism.
At a conceptual level, I introduce the concept of \textit{welfare-cost-effectiveness}, which refers to emission reductions that are achieved at the lowest possible welfare (utility) cost, contrasting it with the prevalent concept of cost-effectiveness, which focuses on minimizing monetary costs. I demonstrate that regionally differentiated utilitarian carbon prices are welfare-cost-effective, yielding higher carbon prices in wealthier countries, while Negishi-weighted carbon prices are cost-effective.

To address the theoretical ambiguity regarding whether accounting for inequality raises or lowers optimal carbon prices, I employ numerical simulations with RICE to explore the direction and magnitude of this effect.
I find that the welfare-maximizing solutions yield more stringent climate policies than the Negishi solution. 
Simply put, accounting for inequality means stronger climate action.
Specifically, the utilitarian uniform carbon price in 2025 is around 15\% higher than the Negishi-weighted carbon price, under default discounting parameters. 
The utilitarian solution that allows for differentiated carbon prices features higher carbon prices in rich countries and lower carbon prices in poor countries; globally, cumulative emissions are 21\% lower compared to the Negishi solution.

I leverage the theoretical insights to uncover the key drivers of the numerical results. 
The utilitarian differentiated carbon price solution results in reduced global emissions because the additional abatement in developed countries outweighs the reduced abatement in the poorest regions.
Higher uniform carbon prices in the utilitarian solution, compared to the Negishi solution, are primarily driven by the poorest region, Africa, which is most impacted by climate change for two main reasons: it has the highest marginal damages and the fastest population growth.
Using the theoretical link to regions’ preferred uniform carbon prices to strengthen the intuition, I find that Africa also favors the highest uniform carbon price---more than twice the preferred price of the US and the Negishi-weighted carbon price. 
Thus, the main intuition for lower carbon prices in the Negishi solution is as follows: by assigning lower weight to the welfare of poorer regions, Negishi weights effectively also downweight the region most impacted by climate change, leading to less stringent climate policy.


This paper makes several contributions to the literature on optimal carbon prices with heterogeneous regions.
First, it provides a set of novel theoretical results on how optimal carbon prices depend on regional welfare weights. 
To my knowledge, it is the first paper to derive the conditions under which accounting for global inequality increases the optimal global climate policy stringency in the absence of transfers.
In doing so, I identify factors that have previously been underappreciated in this context: regional heterogeneities in the convexity of the abatement cost function, population growth, and economic growth.
These results build on influential papers by \textcite{chichilnisky_who_1994} and \textcite{eyckmans_efficiency_1993}, which show that globally uniform carbon prices are optimal if and only if distributional issues are ignored (through the choice of welfare weights) or unrestricted lump-sum transfers can be made between countries. 
I also offer a new perspective to the literature on regions' preferred uniform carbon prices \parencite{weitzman_can_2014, kotchen_which_2018}, 
casting the choice of welfare weights as a problem of aggregating heterogeneous climate policy preferences. 
Instead of focusing on voting mechanisms, I concentrate on an aggregation of preferences rooted in welfare-economic theory.
However, I show that under certain assumptions, there is an intuitive connection between the two: the Negishi solution behaves like ``one dollar, one vote'', whereas the utilitarian solution behaves (approximately) like ``one person, one vote''.
Another related strand of literature explored aspects of efficiency and equity in emission permit markets \parencite{chichilnisky_environmental_2000, shiell_equity_2003, sandmo_global_2007, borissov_optimal_2022}. In contrast, this paper focuses on a setting without international permit markets.
Other related papers examined the importance of accounting for inequalities at a fine-grained level \parencite{dennig_inequality_2015, schumacher_aggregation_2018}, and how optimal carbon taxes, 
under different welfare weights, 
depend on distortionary fiscal policy \parencite{barrage_optimal_2020, dautume_should_2016, douenne_optimal_2023} and inequality within and between countries \parencite{kornek_social_2021}.
However, unlike the present paper, these studies do not theoretically explore how utilitarian and Negishi weights shape the optimal stringency of climate policy.

Second, this paper adds to a body of work that numerically investigates the role of regional welfare weights in IAMs. The study most closely related to this research is \textcite{anthoff_optimal_2011}, which also compares the Negishi solution to a utilitarian solution, although using a different integrated assessment model\footnote{
\textcite{budolfson_optimal_2020} also compare utilitarian and Negishi solutions. However, they do not technically use Negishi weights but a model in which all individuals consume the global average consumption.
}. 
The present paper expands upon this study in multiple ways. First, it offers a deeper understanding of the key drivers behind the results by linking the new theoretical insights to the numerical findings and regional characteristics. 
Furthermore, I extend the analysis by examining the distributional implications, the utilitarian uniform carbon price, and regions' preferred uniform carbon prices. This offers additional insights into heterogeneous climate policy preferences and further strengthens the intuition behind the numerical results.
A related but distinct literature estimates the equity-weighted social cost of carbon (SCC), a measure of the marginal damages of carbon emissions that places more weight on the costs and benefits in poor countries. 
The key difference between this literature and the present paper is that the equity-weighted SCC typically estimates the marginal damages along \textit{non-optimal} emissions pathways \parencite{azar_discounting_1996, anthoff_equity_2009, adler_priority_2017, anthoff_inequality_2018, prest_equity_2024}, while this paper investigates how regional welfare weights affect optimal carbon prices.


The remainder of this paper is organized as follows. Section \ref{sec:Conceptual Background} provides conceptual background on the different optimization approaches. In Section \ref{sec:Theory}, a theoretical model is introduced and key analytical results are derived. Section \ref{sec:Simulations} describes modifications to the RICE model and presents simulation results. Section \ref{sec:conclusion} concludes.

\section{Conceptual Background} \label{sec:Conceptual Background}

\todo[inline]{Old text. Completely rewrite, restructure, and shorten the conceptual background section. Maybe put some supplementary information in the appendix.}

This section provides conceptual background on different optimization approaches in order to lay the foundation for my own analysis. In Section \ref{ssec: Positive versus normative optimizations}, the difference between positive and normative approaches is introduced. Section \ref{ssec: The positive approach: Background on Negishi weights} discusses the positive approach.
Finally, Section \ref{ssec: The normative approach: Welfare-economic conceptualization} provides a welfare-economic conceptualization of normative optimizations.
 
\subsection{Positive and normative optimizations}\label{ssec: Positive versus normative optimizations}
The purpose of optimization is a main source of debate in climate economics, and two approaches are sometimes distinguished: positive and normative optimizations \parencite{kelleher_reflections_2019}. \textcite[p.~1081]{nordhaus_integrated_2013} provides an instructive discussion of these two approaches, noting that ``the use of optimization can be interpreted in two ways: they can be seen both, from a positive point of view, as a means of simulating the behavior of a system of competitive markets and, from a normative point of view, as a possible approach to comparing the impact of alternative paths or policies on economic welfare''. 
In brief, the positive approach seeks to identify the competitive equilibrium, while the normative approach aims at maximizing social welfare. 
Which approach is taken depends on the welfare weights in the SWF\footnote{
The positively and normatively determined welfare weights coincide for a specific normative stance, but in general they are different.
}. 

The issue of discounting, which determines the intertemporal weighting of consumption and welfare, has received much attention in the debate on positive versus normative optimization approaches \parencite{arrow_how_2013, azar_discounting_1996, beckerman_ethics_2007, dasgupta_discounting_2008, dietz_why_2008, gollier_debate_2012, nordhaus_review_2007}. Under the positive approach, the discount rate is determined based on market observations. In contrast, the normative approach relies on ethical reasoning to determine the discount rate.

However, the difference between positive and normative optimization approaches extends to the interregional weighting of welfare. The typical positive approach relies on Negishi welfare weights, which are higher for rich individuals, to identify the competitive equilibrium. In contrast, under the normative approach, uniform welfare weights, which are also called utilitarian welfare weights, are most commonly used, weighting everybody’s welfare within a time period equally\footnote{
Note that other normatively-founded SWFs have been used in the climate economics literature, including the prioritarian SWF \parencite{adler_priority_2017} and variants of the Rawlsian SWF \parencite{roemer_ethics_2011, llavador_intergenerational_2010,llavador_dynamic_2011}.
}. This paper focuses on interregional welfare weights and how they influence optimal carbon prices.

While the distinction between positive and normative optimizations is useful to clarify the different purposes of optimization, this distinction is not always clear-cut in climate economics.
In particular, it has been questioned whether it is possible to interpret the modeling choices that are typically made under the positive optimization approach as purely positive (see \textcite{chawla_normative_2023}, for a discussion)\footnote{
An additional confusion sometimes arises when ``positive'' optimization results appear to be used to suggest how policies \textit{should} be designed \parencite{kelleher_reflections_2019}. 
While possible in principle, a normative justification of positively calibrated welfare weights would be required to draw normative conclusions from a positive analysis.
}.
Keeping this caveat in mind, I use the labels ``positive'' and ``normative'' to highlight the conceptual difference underlying these optimization approaches: whether the optimization seeks to simulate markets or to maximize social welfare.

\todo[inline]{Maybe add my Table on positive and normative welfare weights. Actually don't think I need it.}

\subsection{The positive approach: Background on Negishi weights}\label{ssec: The positive approach: Background on Negishi weights}

Negishi welfare weights are commonly used in regionally disaggregated integrated assessment models of climate change. Popular IAMs that use Negishi weights include RICE \parencite{nordhaus_regional_1996}, which this paper focuses on, MERGE \parencite{manne_merge_2005},  REMIND \parencite{leimbach_mitigation_2010} and WITCH \parencite{bosetti_incentives_2012}. 
This section outlines the rationale for and critiques of using Negishi weights in IAMs. It finishes with a welfare economics perspective on the positive optimization approach.

\subsubsection{Rationale for using Negishi weights in IAMs}\label{sssec: Rationale for using Negishi weights in IAMs}

The theoretical basis for the use of Negishi weights is a theorem by \textcite{negishi_welfare_1960}. Negishi proved that a competitive equilibrium can be found by maximizing a social welfare function in which the welfare of each agent is appropriately weighted such that each agent’s budget constraint is satisfied at the equilibrium \parencite{nordhaus_regional_1996}. 
The Negishi-weighted SWF is given by a weighted sum of agents’ utilities, where the weights are inversely proportional to the marginal utility of consumption.
For identical and concave utility functions, which are commonly assumed, this implies higher welfare weights for wealthy individuals, with a low marginal utility of consumption, than for poorer individuals.
The appeal of the Negishi-weighted SWF is that it provides a computationally convenient method to identify the competitive equilibrium, which is Pareto efficient if the conditions of the first fundamental theorem of welfare economics are satisfied.

Besides this theoretical foundation, there were two additional motivations for the use of Negishi weights in IAMs: (1) to prevent transfers across regions, which were considered politically infeasible or unrealistic \parencite{nordhaus_regional_1996}, and (2) to obtain a uniform carbon price in all regions, ensuring that global emissions are reduced at the lowest possible cost. 
Indeed, to achieve these two objectives in every period of the RICE model, \textcite{nordhaus_regional_1996} made refinements to what they call the ``pure Negishi solution'' that relies on time-invariant welfare weights.
Specifically, \textcite{nordhaus_regional_1996} adjust the Negishi weights in each time period such that the weighted marginal utility of consumption is equalized in each period \parencite{stanton_negishi_2011}.
This approach yields time-variant Negishi weights and accomplishes the goal of equalizing the carbon price across regions in every period. 
Moreover, these weights ensure that no cross-regional transfers take place, since such transfers do not increase the objective value of the Negishi-weighted SWF.
Notably, without Negishi weights, social welfare could be increased by redistributing capital or consumption from rich to poor regions in models that maximize the unweighted sum of agents’ utilities, as long as utility is an increasing concave function of consumption, which is commonly assumed.
Hence, the constraints of equalized carbon prices and no transfers are effectively incorporated in the time-variant Negishi weights used in RICE.

\subsubsection{Critiques of using Negishi weights in IAMs}\label{sssec: Critiques of using Negishi weights in IAMs}
While Negishi weights are commonly used in IAMs, the use of such weights has been criticized on both ethical and theoretical grounds \parencite{anthoff_differentiated_2021, dennig_note_2019, stanton_negishi_2011, stanton_inside_2009}. This section provides a summary of main critiques.

From an ethical perspective, a main critique is that Negishi weights assign greater weight to the welfare of people in rich countries than in poor countries. This is the case because Negishi weights are inversely proportional to the marginal utility of consumption and the utility function is typically assumed to be concave. 
Models with Negishi weights are thus ``acting as if human welfare is more valuable in the richer parts of the world'' \parencite[176]{stanton_inside_2009}. 
Moreover, because Negishi weights equalize the weighted marginal utility of consumption, aspects of interregional equity are effectively ignored and global inequality is neglected \parencite{stanton_negishi_2011, stanton_inside_2009}. As a result, it is irrelevant whether poor or rich countries are affected by climate change and climate policies \parencite{dennig_inequality_2015}.




Moreover, \textcite{stanton_negishi_2011} notes that models with Negishi weights have an inherent conceptual inconsistency: the diminishing marginal utility of consumption is embraced intertemporally, but suppressed interregionally. 
Consequently, transfers from richer to poorer individuals are desired in an intertemporal context but rejected in an interregional context.

Another criticism from a theoretical perspective is provided by \textcite{dennig_note_2019} and \textcite{anthoff_differentiated_2021}. In a simple analytical model, these authors show that the time-variant Negishi weights, used for example in the RICE model \parencite{nordhaus_regional_1996}, distort the time-preferences of agents and result in different saving rates than those implied by the underlying preference parameters. Furthermore, they note that the time-invariant weights proposed by \textcite{negishi_welfare_1960} do not have this problem because they only consist of one weight per agent, and thus only affect the distribution between agents, but leave the intertemporal choices of each agent unaffected. 

A final criticism of Negishi weights concerns the manner in which Negishi weights are often introduced---if discussed at all---which is frequently rather technical with no or little transparent discussion of the ethical implications \parencite{abbott_following_2014, stanton_negishi_2011}. 


\subsubsection{Welfare economics perspective on the positive approach}\label{sssec: Welfare economics perspective on the positive approach}


This section provides a discussion of the positive optimization approach from the perspective of welfare economics. 
From the first fundamental theorem of welfare economics, it is known that, under certain conditions,
the competitive equilibrium is Pareto efficient \parencite{sen_moral_1985}. 
The maximization of a Negishi-weighted SWF in IAMs seeks to identify the competitive equilibrium with a Pareto efficient level of abatement\footnote{
However, \textcite{anthoff_differentiated_2021} show that the time-variant Negishi weights used in IAMs do not, in fact, yield a Pareto efficient solution. This is because of a time-preference altering effect of time-variant Negishi weights. In this section, I focus on a static setting in which this issue does not arise. 
}. 
I refer to this solution as the ``Negishi solution''. 
The Negishi solution is one particular point---among infinitely many points---on the Pareto frontier in a first-best setting in which only resource and technology constraints are present (assuming that the conditions of the first fundamental theorem of welfare economics hold).
Notably, it is the only Pareto efficient allocation in a first-best setting that does not require transfers \parencite{shiell_equity_2003}. 
In the absence of abatement, the competitive equilibrium is not efficient due to the climate externality\footnote{
This is also the case if abatement is inefficiently low, as it is the case in the Nash equilibrium.
}. 
The Pareto efficiency of the Negishi solution, and the inefficiency of no abatement, is illustrated in Figure~\ref{F1}a, which shows 
the Pareto frontier for a simple example with two regions: a rich Global North, and a comparatively poor Global South.

While the Negishi solution is Pareto efficient, it cannot generally be interpreted as maximizing social welfare in a normative sense.
The Negishi-weighted SWF is not intended to reflect ethical theories of social welfare but is instead calibrated to support a Pareto-efficient allocation without requiring transfers.
In contrast, normative analyses rely on social welfare functions grounded in explicit normative theories of social welfare. The most commonly used such theory in public economics is utilitarianism \parencite{fleurbaey_optimal_2018}, which 
assigns equal weight to the welfare of all individuals.
Importantly, the Negishi solution does not maximize aggregate welfare if the welfare of all people is weighted equally. 
Maximizing a utilitarian SWF maximizes the (equally-weighted) sum of the welfare of all individuals. 
This is illustrated in Figure~\ref{F1}b, which shows the social indifference curves of the utilitarian and Negishi-weighted SWFs, and the points that maximize these SWFs\footnote{
To choose among different points on the Pareto frontier (or, more generally, any vector of utilities), interpersonal utility comparisons are often made. However, the admissibility of such comparisons is a longstanding point of contention in welfare economics \parencite{robbins_essay_1935, harsanyi_cardinal_1955, sen_interpersonal_1970, stiglitz_pareto_1987}.    
Indeed, contemporary welfare economics is split into two branches: one dismisses interpersonal utility comparisons, while the other branch relies on such comparisons and uses SWFs to determine socially preferable outcomes \parencite{fleurbaey_utilitarianism_2008}.
This paper belongs to the second branch.
}.

\begin{figure}[!htb]
    \centering
    \includegraphics[width=0.8\linewidth]{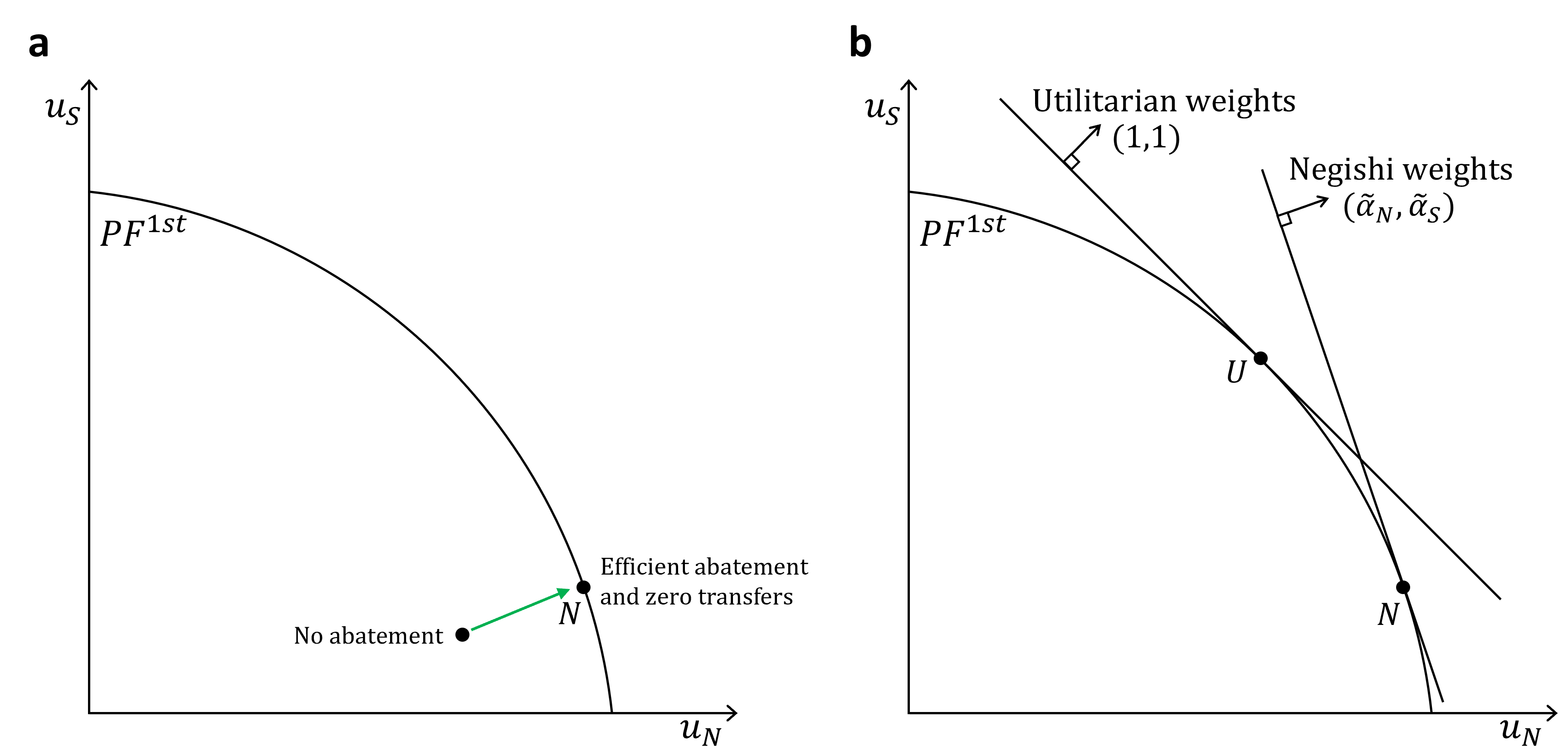}
\caption{Illustrative two-region example of the welfare outcomes under the Negishi and utilitarian solutions.}
 \vspace{1mm}  
    \begin{minipage}{1\textwidth}
        \small \textit{Notes}:
        Panel (a) shows that the Negishi solution ($N$) is Pareto efficient. Panel (b) shows an illustrative comparison of the Negishi ($N$) and utilitarian ($U$) solutions. The utilities of representative agents in the Global North and Global South are denoted by $u_N$ and $u_S$, respectively. $PF^{1st}$ is the Pareto frontier in a first-best setting. The welfare weights vectors are the gradient vectors of the SWFs, which are perpendicular to the linear social indifference curves. The Negishi weights are denoted by $\tilde{\alpha}_i$.
    \end{minipage}
\label{F1}
\end{figure}
 
Given that the Negishi solution does not maximize aggregate (unweighted) welfare, how may the use of Negishi weights in IAMs be justified? 
There are at least two possible lines of argument.
First, it may be argued that the Negishi solution has no normative but only a positive interpretation; that it is merely a procedure to identify the competitive equilibrium with Pareto efficient abatement and zero transfers. 
For example, \textcite[1111]{nordhaus_integrated_2013} notes that ``if the distribution of endowments across individuals, nations, or time is ethically unacceptable, then the ``maximization'' is purely algorithmic and has no compelling normative properties''.
Moreover, \textcite[20]{nordhaus_dice_2013} clarify: ``We do not view the solution as one in which a world central planner is allocating resources in an optimal fashion''.

A second line of argument used to support employing Negishi weights relies on the second fundamental theorem of welfare economics, which states that any point on the Pareto frontier can be supported as a competitive equilibrium if unrestricted lump-sum transfers can be made. 
This is sometimes used to argue that the issues of equity and efficiency can be separated.
However, this typically is not the case for climate policy; the Pareto efficient abatement level generally depends on the distribution of wealth.
This is because the marginal willingness to pay for abatement generally varies with income \parencite{shiell_equity_2003}. 
Therefore, the Negishi solution only identifies a Pareto efficient abatement level if no transfers occur.
Moreover, the practical relevance of the second welfare theorem has been questioned. 
For instance, \textcite[12]{sen_moral_1985} notes that ``if there is an absence of---or reluctance to use---a political mechanism that would actually redistribute resource-ownership and endowments appropriately, then the practical relevance of the converse theorem [the second welfare theorem] is severely limited''.

To summarize, the abatement in the Negishi solution generally differs from the abatement that maximizes global utilitarian welfare (hereafter, simply ``global welfare''),
regardless of whether unrestricted transfers are feasible.





\subsection{The normative approach: Welfare-economic conceptualization}\label{ssec: The normative approach: Welfare-economic conceptualization}


This section provides a conceptualization of the normative optimization approach, grounded in welfare economics. In doing so, the objective of this section is to clarify the fundamental distinction between positive and normative optimization approaches in climate economics.
In Section \ref{sssec: Rationale for using Negishi weights in IAMs}, I have argued that constraints are implicitly incorporated in the welfare weights under the positive approach. Here, I emphasize that this marks a key difference from the normative approach, where constraints and welfare weights are determined separately.

I propose to conceptualize the normative optimization approach as consisting of two steps. First, the social welfare function is defined based on ethical principles. Second, potential constraints are specified which affect the feasible set of allocations.
The first step---the specification of the SWF based on ethical principles---is common in normative analyses.
Such SWFs have a long tradition in public economics, and particularly in the optimal income taxation literature \parencite{mirrlees_exploration_1971,piketty_optimal_2013,fleurbaey_optimal_2018}. 
They are referred to as Bergson-Samuelson SWFs \parencite{bergson_reformulation_1938, samuelson_foundations_1947}
and produce an ethical ordering of societal outcomes. Common Bergson-Samuelson SWFs include the utilitarian, prioritarian and Rawlsian SWFs \parencite{mas-colell_microeconomic_1995}. 
In this paper, I focus on the utilitarian SWF,
the most widely used normatively grounded Bergson-Samuelson SWF in climate economics.

\todo[inline]{Axiomatic foundation of SWFs}


The second step is to carefully consider and explicitly account for real-world constraints in the optimization. This step is often less thoroughly addressed in the existing literature. It is of course challenging to determine and formalize plausible real-world constraints, especially in stylized IAMs. It therefore seems valuable to explore a plausible range of constraints.
Conceptually, such constraints affect the feasible set of allocations, which, in turn, determines the utility possibility set (UPS), which was introduced by \textcite{samuelson_foundations_1947}. 
Ultimately, we are interested in the Pareto frontier, which is defined as the upper frontier of the UPS\footnote{
Economists sometimes use the term efficiency to simply mean outcomes that maximize the total monetary sum (for short, ``maximizing dollars''). In a first-best setting in which unrestricted lump-sum transfers are feasible, maximizing dollars is necessary and sufficient for Pareto efficiency. Importantly, however, in a second-best setting in which unrestricted lump-sum transfers are infeasible, maximizing dollars is not necessary for Pareto efficiency (nonetheless, maximizing dollars is, of course, one Pareto efficient outcome on the Pareto frontier among infinitely many other points on the Pareto frontier that do not maximize dollars). Throughout this paper, I use the standard definition of Pareto efficiency that no one can be made better off without making someone else worse off, given the constraints of the problem.
}. 
Finally, the social optimum is the point on the Pareto frontier that maximizes the SWF.

Depending on the constraints imposed on the optimization, a conceptual distinction between first-best and second-best settings is frequently made \parencite{mas-colell_microeconomic_1995}.
Typically, a first-best setting is considered to be a setting in which only resource and technology constraints are present, but otherwise the social planner has access to any policy instrument, including unrestricted lump-sum transfers. 
In contrast, the notion of second-best settings is used when additional constraints are present.

It is instructive to illustrate how the normative optimization approach works in the context of this paper. This is shown in Figure~\ref{fig:Illustration of the normative optimization approach} for optimization problems considered in this paper. In the first step, the utilitarian SWF is specified (which has linear social indifference curves with slope -1). In the second step, potential constraints are specified. 
Of particular relevance in the context of international climate policy are constraints on international transfers and whether carbon prices are constrained to be uniform across countries.

In the first-best setting, there are no constraints apart from the usual resource and technology constraints. 
In particular, unrestricted lump-sum transfers can be made. 
In this setting, the social planner uses cost-effective and efficient uniform carbon prices to internalize the climate externality and lump-sum transfers to address distributional issues. 
With identical and concave utility functions, large transfers are made to equalize per capita consumption across regions \parencite{dennig_inequality_2015}, eliminating inequality. This results in the highest utilitarian welfare; the outermost social indifference curve, $W_{1st}$, is achieved. 

However, as discussed above, such a first-best setting with large international transfers may be politically infeasible. 
As \textcite[43]{shiell_equity_2003} puts it, ``Unrestricted lump-sum transfers are a useful construct which scarcely exist outside the confines of economic theory''. 
As discussed in Section \ref{sssec: Rationale for using Negishi weights in IAMs}, the political infeasibility of large transfers was one of the reasons that motivated the use of Negishi weights under the positive optimization approach. In contrast, under the normative optimization approach, political transfer constraints affect the feasible set of allocations while welfare weights remain unchanged.

\begin{figure}[!htb]
    \centering
    \includegraphics[width=0.45\linewidth]{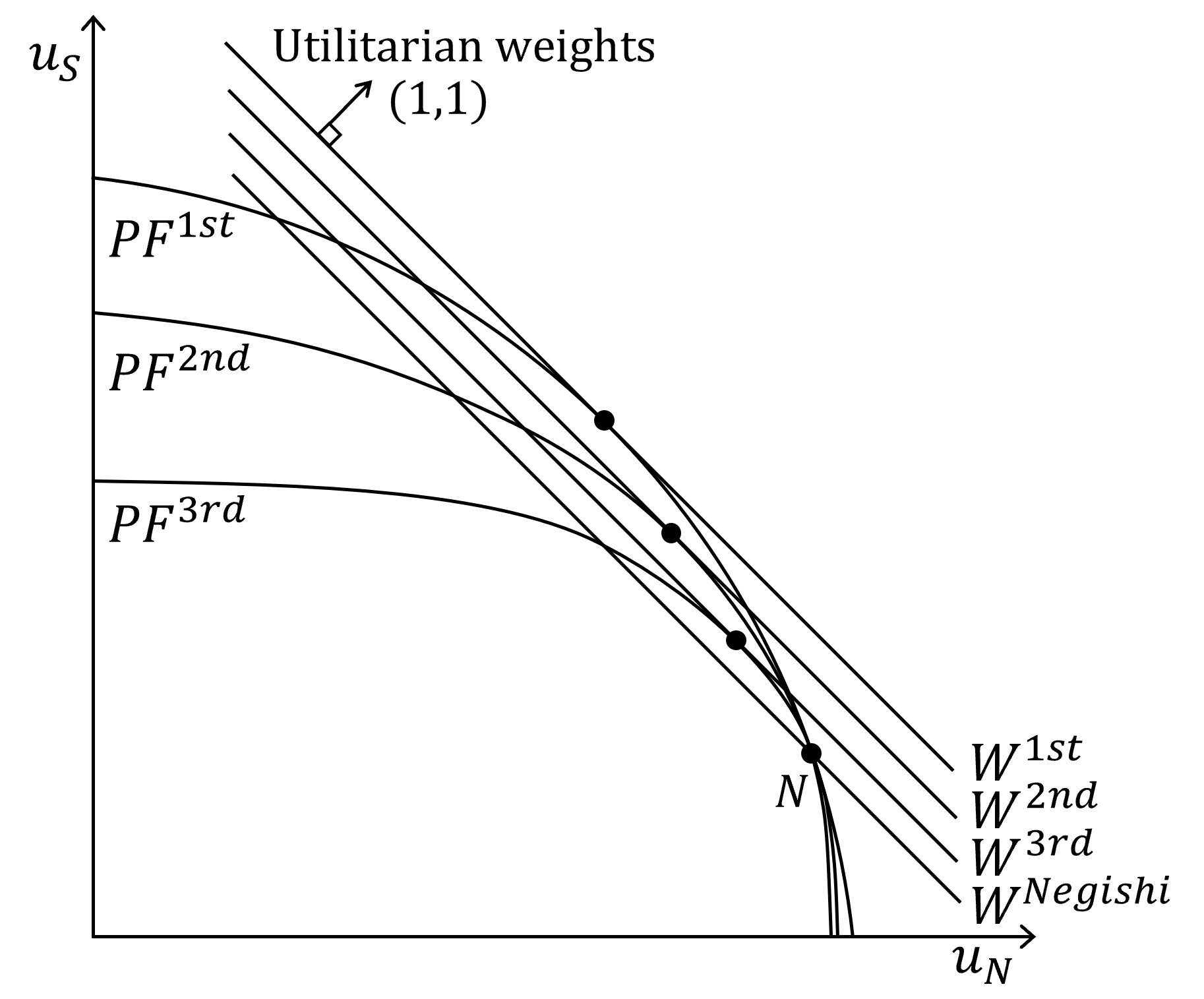}
    \caption{Illustration of the utilitarian social welfare outcomes in first-, second-, and third-best settings, and the Negishi solution.}
    \vspace{1mm}  
    \begin{minipage}{1\textwidth}
        \small \textit{Notes}:
        The figure shows the utilitarian optima in first-, second-, and third-best settings and the utilitarian welfare level of the Negishi solution.
        The utilities of representative agents in the Global North and Global South are denoted by $u_N$ and $u_S$, respectively. 
        $PF^x$ is the Pareto frontier in the $x^{th}$-best setting. $W^x$ is the utilitarian social indifference curve that corresponds to the social optimum in the $x^{th}$-best setting or the Negishi solution.
        The utilitarian weights vector is the gradient vector of the utilitarian SWF.
    \end{minipage}
    \label{fig:Illustration of the normative optimization approach}
\end{figure}

Let us now consider such a second-best setting in which international lump-sum transfers are infeasible\footnote{
I intentionally focus on the case of no transfers here to keep the discussion simple. In reality, however, some transfers are feasible (e.g., international aid or climate finance). A companion paper (in progress) examines how international climate finance affects optimal carbon prices.
}. The lack of this policy option reduces the feasible set of allocations, the UPS gets smaller, and the Pareto frontier moves inward (except for one point on the frontier, which corresponds to the Negishi solution, which does not require transfers). Consequently, the social optimum lies on a lower social indifference curve, $W_{2nd}$.
In the absence of the option to address inequality with lump-sum transfers, the utilitarian social planner accounts for global inequality in the climate policy design. Specifically, differentiated carbon prices that are higher in rich regions and lower in poor regions are used to reduce the welfare cost of abating emissions \parencite{chichilnisky_who_1994} (also see Section \ref{sssec:Optimal carbon prices under different welfare weights}). It should be noted that a potential problem with differentiated carbon prices is carbon leakage---an increase in emissions in countries with laxer climate policies as a result of stricter climate policies elsewhere. However, additional policies such as carbon border adjustments and binding emission targets can avert the issue of carbon leakage. For a more detailed discussion, see \textcite{budolfson_utilitarian_2021} and Appendix \ref{Assec: Discussion of DCPO}. 

Finally, consider a third-best setting in which the policy instruments the social planner can use are restricted even further to a globally uniform carbon price (in addition to a constraint of no transfers). I would argue that this is not a plausible constraint in reality, as evidenced by widely different empirical carbon prices across countries \parencite{world_bank_state_2023}.
Nevertheless, it provides a useful comparison to the solution under the positive optimization approach, as it constrains the utilitarian problem to an identical policy instrument---a globally uniform carbon price and no transfers. 
Yet, an important difference remains.
The utilitarian uniform carbon price accounts for inequality in the carbon price level, while the positive optimization approach ignores inequality altogether through the specification of the Negishi weights, which equalize the weighted marginal utility across regions.
Consequently, the utilitarian uniform carbon price solution is weakly better, from the perspective of utilitarian welfare, than the Negishi solution (compare social indifference curves $W_{3rd}$ and $W_{Negishi}$). 
This is simply because the utilitarian uniform carbon price is, by construction, the uniform carbon price that maximizes utilitarian welfare in a setting in which transfers are infeasible. 

It is worth highlighting how the different solutions respond to global inequality. The spectrum ranges from completely solving inequality through lump-sum transfers in the first-best utilitarian setting to ignoring inequality altogether in the Negishi solution. While the social optima in the second- and third-best settings do not solve inequality through transfers, they account for inequality to different degrees in the carbon pricing policy. In the second-best setting, inequality is accounted for in the \textit{level} and \textit{differentiation} of carbon prices across regions. In contrast, in the third-best setting, inequality is only accounted for in the \textit{level} of the carbon price.

The extent to which inequality is ultimately accounted for in international climate policy is decided by policymakers and international negotiations. However, international agreements indicate that there is a political consensus to account for inequality to some extent. This is evidenced, for example, by the UNFCCC principle of ``common but differentiated responsibilities and respective capabilities, in the light of different national circumstance'' and a general recognition that developed countries have an obligation to reduce their emissions faster and support developing countries in their transitions toward low-carbon economies, which is also reflected in the respective nationally determined contributions (NDCs) under the Paris Agreement \parencite{unfccc_report_2015, climate_watch_explore_2022}. More broadly, the Paris Agreement underscores the necessity of incorporating the principle of equity and the goal of poverty eradication into climate policy, indicating that countries have agreed to account for inequality in international climate policy \parencite{unfccc_report_2015}. Hence, policymakers may be interested in socially optimal climate policies that take inequality into account.
The present study seeks to identify such policies and contrasts them with the conventional, positive approach that neglects inequality.

\section{Theory} \label{sec:Theory}

This section provides a theoretical analysis of how regional welfare weights affect optimal carbon prices.

\subsection{Model setup} \label{ssec:Model setup}

The model setup builds on \textcite{chichilnisky_who_1994} and \textcite{dennig_note_2017}.
I intend to construct the simplest model possible to generate key insights and to provide theoretical underpinnings for the main drivers of the simulation results in Section \ref{sec:Simulations}.

There are two regions, indexed by $i \in \mathcal{I} = \{N,S\}$, and a single period; for intuition, interpret $N$ and $S$ as the Global North and Global South, respectively. 
Section \ref{sssec: Extension: Dynamic model} considers a two-period extension, and the appendix generalizes main results to $n$ regions.
The population of each region is exogenous and denoted by $L_{i}$.
Uppercase letters are used for aggregate variables at the region level, while lowercase letters are used for per capita variables and, in some cases, per endowment variables.

Abatement costs, $C_{i}(A_{i})$, are a function of the abatement $A_{i} \geq 0$ in region $i$. The abatement cost function differs by region and is assumed to be smooth, strictly increasing, $\frac{dC_{i}}{dA_{i}} > 0$, and strictly convex, $\frac{d^2C_{i}}{dA_{i}^2} > 0$. Moreover, to keep the exposition simple, I assume that $\frac{d^2C_{i}}{dA_{i}^2}$ is constant but region-specific; that is, $\frac{d^3C_{i}}{dA_{i}^3} = 0$ for all $A_i$.
\todo[inline]{Confirm that it's true that some of the main results do not require the assumption of constant second derivatives of the abatement cost functions}
This is the case for the commonly assumed quadratic abatement cost function.
The aggregate global abatement is given by $A \equiv \sum_i A_{i}$. 
Region-specific climate damages, $D_{i}(A)$, are a function of the global abatement. 
The damage function is assumed to be smooth, strictly decreasing, $\frac{dD_{i}}{dA} < 0$, and strictly convex in abatement, $\frac{d^2D_{i}}{dA^2} > 0$, reflecting the idea of convex damages as a function of emissions.
Regional consumption, $X_{i}$, is given by the exogenous endowment, $W_{i}$, net of abatement costs and climate damages:
$X_{i} = W_{i} - C_{i}(A_{i}) - D_{i}(A)$.
There is a representative agent in each region, who derives utility, $u(x_i)$, from per capita consumption, $x_i = X_i/L_i$. The utility function is assumed to be identical for all individuals, strictly increasing, strictly concave, and smooth. Thus, $\frac{du}{dx_i}>0$ and $\frac{d^2u}{dx_i^2}<0$.

I assume throughout that the Global North is richer than the Global South, both in terms of per capita endowment and consumption. Thus, we have $w_{N} > w_{S}$  and $x_{N} > x_{S}$. The implicit assumption is that the difference in endowment per capita between the Global North and the Global South is sufficiently large such that individuals in the Global North remain richer after abatement costs and climate damages are subtracted. From the concavity of the utility function, it follows that $u'(x_{N}) < u'(x_{S})$.

While I derive the theoretical results for general functional forms, it is useful to put more structure on the abatement cost and damage functions to closely link theory and simulation results. To do this, I use simplified versions of the functions employed in the RICE model\footnote{
See Appendix \ref{Assec: Abatement cost and damage functions of the RICE model} for the abatement cost and damage functions of the RICE model.
}, capturing their key characteristics. I define these ``simplified RICE functions'' as follows:
\begin{gather}
    D_{i}(A) = L_{i} w_{i} d_{i}(A),\\
    C_{i}(A_i) = L_{i} w_{i} c_{i}\left(\frac{A_{i}}{L_{i} w_{i}}\right),
\end{gather}
where $w_i$ is the endowment per capita and $d_i$ and $c_i$ denote the damages and abatement costs per aggregate regional endowment, respectively. Note that $c_i$ is a function of abatement relative to the size of the economy, reflecting that larger economies have more abatement opportunities of a certain type and cost.

\todo[inline]{Simplified RICE functions. Mention that I assume that per capita endowment and population are exogenous (maybe cite Kevin Kuruc).}

\subsubsection{Optimization problems} \label{sssec:Optimization problems}


I consider two general optimization problems, reflecting the optimizations that are most commonly performed in the literature on optimal carbon prices (e.g., in \textcite{nordhaus_regional_1996}, \textcite{dennig_inequality_2015}, \textcite{budolfson_utilitarian_2021}). 
The first allows for (but does not require) differentiated carbon prices and the second requires uniform carbon prices. 
The social planner's objective is to choose the carbon prices that maximize the SWF, with welfare weights $\alpha_i \geq 0$, subject to regional budget constraints, reflecting a constraint of no interregional transfers. 


Formally, the \textit{differentiated carbon price optimization problem} is
\begin{gather}
    \max_{X_i, A_i} \sum_i L_i \alpha_i u \left(\frac{X_i}{L_i}\right) \\
    \begin{aligned}
    \text{subject to: }
    \label{E: constraints diff1period}
    &X_i = W_i - C_i(A_{i}) - D_i(A), \quad \forall i. \\
    \end{aligned}
\end{gather}

The \textit{uniform carbon price optimization problem} is identical except that an additional constraint of uniform marginal abatement costs is imposed\footnote{
Note that I am using prime notation for derivatives: $C'_{i}(A_{i}) \equiv \frac{dC_{i}(A_{i})}{dA_{i}}$,  $D'_{i}(A) \equiv \frac{dD_{i}(A)}{dA}$, $u'(x_{i}) \equiv \frac{du(x_{i})}{dx_{i}}$.
},
\begin{gather}
    C'_N(A_N) = C'_S(A_S).
\end{gather}

\subsection{Optimal carbon prices under different welfare weights} \label{sssec:Optimal carbon prices under different welfare weights}

Solving the optimization problems yields expressions for the optimal marginal abatement costs. Optimal carbon prices, $\tau_{i}$, are equal to the optimal marginal abatement costs, $C'^*_{i}$, because regions are assumed to optimally respond to a carbon price by abating until their marginal abatement cost equals the carbon price; that is,   $C'_{i}(A_i^*(\tau_i)) = \tau_i$. 

I focus on the optimal carbon prices under the welfare weights that are most commonly used in climate economics---Negishi weights and utilitarian weights. Optimal carbon prices under arbitrary welfare weights are shown in Appendix \ref{Assec: Optimal carbon prices for arbitrary welfare weights}. Derivations are provided in Appendix \ref{Assec: Derivation of optimal carbon prices}.

\subsubsection{The Negishi solution}
I begin with the Negishi solution. Negishi weights, $  \tilde{\alpha}_i$, are inversely proportional to a region's marginal utility of consumption at the optimal solution that was obtained with the Negishi weights\footnote{
The Negishi weights that satisfy this are obtained by iteratively updating the weights until convergence. 
}; that is, $\tilde{\alpha}_i = 1/u'(\tilde{x}_i)$. I use ``tilde'' to indicate the Negishi solution. 
Since we assume that consumption is higher in the North and the utility function is concave, it follows that the Negishi weight is greater for the North than the South: $\alpha_N > \alpha_S$.

Solving the differentiated carbon price optimization problem with Negishi weights yields the Negishi solution. For reference, I record the optimal carbon prices in definitions.

\begin{definition}
    The \textbf{Negishi-weighted carbon price} is implicitly defined by
    \begin{gather}
    \label{Def: Negishi-weighted carbon price}
        \tilde{\tau} = {C}'_{i}(\tilde{A}_i) = - \sum_i {D}'_i(\tilde{A}) .
    \end{gather}
\end{definition}
The Negishi-weighted carbon price is simply equal to the sum of marginal benefits of abatement (i.e., the reduced marginal damages) in monetary terms. This condition is effectively the Samuelson condition for the optimal provision of public goods \parencite{samuelson_pure_1954}.
We have thus obtained the knife-edge result that the Negishi-weighted carbon price is uniform even though we allowed for differentiated carbon prices. Uniform carbon prices arise from the specification of the Negishi weights, which equalize weighted marginal utilities across regions. 
Notably, this also renders no transfers between regions optimal. 

It is insightful to also characterize the optimality conditions in terms of the derivatives with respect to carbon prices. 
Rewriting Equation~\eqref{Def: Negishi-weighted carbon price},  we can see that the Negishi-weighted carbon price equalizes the sum of the marginal abatement costs and benefits from marginally increasing the carbon price (see Appendix \ref{Asssec: Negishi solution} for a derivation):
\begin{gather}
\label{E: Optimality condition, Negishi}
        \sum_i
        \frac{dC_i(\tilde{A}_i(\tilde{\tau}))}
        {d\tilde{\tau}}
        = 
        - \sum_i \frac{d{D}_i(\tilde{A}(\tilde{\tau}))}{d\tilde{\tau}} .
    \end{gather}
\subsubsection{The utilitarian solution with uniform carbon prices}
Next, I turn to the optimal carbon prices under the utilitarian SWF. 
Utilitarian welfare weights are uniform across regions. Without loss of generality, I set them equal to unity: $\alpha_i^{U} = 1$. To highlight that the maximization of the utilitarian SWF  maximizes the (equally-weighted/unweighted) sum of utilities, I refer to the utilitarian solutions as \textit{welfare-maximizing} solutions.

First, I solve the uniform carbon price optimization problem to determine the uniform carbon price that maximizes global welfare.
\begin{definition}
    The \textbf{utilitarian uniform carbon price} is implicitly defined by
    \begin{gather}\label{Def: utilitarian uniform carbon price}
        \check{\tau} =
        {C}'_{i}(\check{A}_i) = 
        - \sum_i {u}'(\check{x}_i) {D}'_i(\check{A})
        \frac{{C}''_S + {C}''_N}
        {{u}'(\check{x}_N) {C}''_S + {u}'(\check{x}_S) {C}''_N } .
    \end{gather}
\end{definition}
The utilitarian uniform carbon price is a function of the sum of the avoided marginal damages in welfare terms rather than monetary terms, which is the case for the Negishi-weighted carbon price. 
Moreover, it depends on a second factor which contains the second derivatives of the abatement cost functions, which govern the abatement changes in response to a marginal change in carbon prices; specifically, 
$\frac{dA_i(\tau_i)}{d\tau_i} = \frac{1}{C_i''}$. Thus, raising the carbon price increases abatement more in the region with the flatter marginal abatement cost curve.  

As before, it is instructive to rewrite the optimality condition in Equation~\eqref{Def: utilitarian uniform carbon price} in terms of the derivatives with respect to the carbon price (see Appendix \ref{Asssec: Utilitarian uniform carbon price solution}):
\begin{gather}
   \sum_i {u}'(\check{x}_i) \frac{d C_i(\check{A}_i(\check{\tau}))}{d \check{\tau}}
       = 
        - \sum_i {u}'(\check{x}_i) \frac{d{D}_i(\check{A}(\check{\tau}))}{d\check{\tau}}.
    \end{gather}
The utilitarian uniform carbon price equalizes the sum of the marginal \textit{welfare} costs and benefits of abatement from marginally increasing the carbon price. This can be contrasted with the Negishi-weighted carbon price, which equalizes the sum of the marginal \textit{monetary} costs and benefits of abatement from marginally increasing the carbon price.

\subsubsection{The utilitarian solution with differentiated carbon prices}
I now relax the constraint of uniform carbon prices and solve the differentiated carbon price optimization problem.
\begin{definition}
    The \textbf{utilitarian differentiated carbon price} for region i is implicitly defined by
    \begin{gather}\label{Def: utilitarian differentiated carbon price}
        \hat{\tau}_i 
        = {C}'_i(\hat{A}_i) 
        = - \frac{1}{{u}'(\hat{x}_i)}
        \sum_{j \in \mathcal{I}} {u}' (\hat{x}_j) {D}'_j(\hat{A}).
    \end{gather}
\end{definition}
The utilitarian differentiated carbon prices equalize the marginal \textit{welfare} costs of abatement across regions (as opposed to the marginal \textit{monetary} costs of abatement in the Negishi solution), which, in turn, are equal to the marginal welfare benefits of abatement:
\begin{gather}
    {u}'(\hat{x}_N) {C}'_N(\hat{A}_N) = {u}'(\hat{x}_S) {C}'_S(\hat{A}_S)
     = - \sum_{j \in \mathcal{I}} {u}' (\hat{x}_j) {D}'_j(\hat{A}).
\end{gather}
This can be interpreted as a form of equal burden sharing, a common concept in international climate negotiations and the related literature (e.g., \textcite{bretschger_climate_2013, rao_international_2014}).

Thus, the welfare-maximizing differentiated carbon price is higher in the richer region, as it is inversely proportional to the marginal utility of consumption---a result that was first established by \textcite{eyckmans_efficiency_1993} and \textcite{chichilnisky_who_1994}. 
This implies that emissions are not reduced at the lowest \textit{monetary} cost, and emission reductions are therefore not cost-effective. Importantly, however, by equalizing the marginal \textit{welfare} cost of abatement, utilitarian differentiated carbon prices achieve emission reductions at the lowest possible \textit{welfare} cost (in the absence of transfers). Thus, I propose to classify these emission reductions as \textit{welfare-cost-effective}, contrasting it with the concept of (monetary) cost-effectiveness. The concept of welfare-cost-effectiveness may also offer a useful perspective in other public policy contexts\footnote{
It seems especially useful in contexts in which transfers by other means are not feasible.
}, particularly in the context of the new regulatory impact analysis guidelines in the US (Circular A-4), which allow for distributional weighting in cost-benefit analyses \parencite{us_office_of_management_and_budget_circular_2023}. 

A second important point is that the utilitarian differentiated carbon prices are Pareto efficient if international transfers cannot be made\footnote{
Sometimes the notion of \textit{constrained} Pareto efficiency is used to refer to Pareto efficiency in settings with additional constraints (beyond the usual resource and technology constraints), particularly constraints on lump-sum transfers \parencite{chichilnisky_equity_2000, shiell_equity_2003}. Instead, I opt to be explicit about the setting, and the corresponding constraints, which determine the Pareto frontier.
}. This point requires elaboration. It is well known that cost-effective emission reductions are necessary to achieve Pareto efficiency if unrestricted lump-sum transfers can be made \parencite{shiell_equity_2003}. However, this is no longer the case when transfers are infeasible. In such a constrained, second-best setting, the set of feasible allocations shrinks and the Pareto frontier moves inward (except for one point that does not require transfers, which is the Negishi solution). If transfers cannot be made, the only way to move from one Pareto efficient allocation to another is through changing the differentiation of carbon prices. In fact, in this setting, all points on the Pareto frontier require differentiated carbon prices, except for one point, which corresponds to the Negishi solution (see Equation~\eqref{Def: optimal differentiated carbon price} in the appendix). The utilitarian differentiated carbon price yields the point on the Pareto frontier that maximizes global welfare.

\subsection{Comparison of optimal climate policy stringency}\label{sssec:Comparison of optimal climate policy stringency}

I now address the central question of this section: How does the optimal climate policy stringency depend on regional welfare weights? 

\subsubsection{Utilitarian uniform versus Negishi}
 I first compare the uniform carbon prices under utilitarian and Negishi weights.
By construction, the utilitarian carbon price maximizes global utilitarian welfare, while the Negishi-weighted carbon price maximizes global consumption in monetary terms. 
The following proposition and corollary establish the conditions under which one is greater than the other.

\begin{proposition}\label{Proposition: utilitarian vs Negishi}
    The utilitarian uniform carbon price is greater than the Negishi-weighted carbon price, that is $\check{\tau} > \tilde{\tau}$, if and only if $\frac{\check{D}'_S}{\check{D}'_N} > \frac{C''_N}{C''_S}$\footnote{
    I use $\check{D}_i'$ as a short-hand for ${D}'_i(\check{A}_i)$. This notation also applies to other functions and solutions.
    }.
\end{proposition}

\begin{proof}[Proof]
\renewcommand{\qedsymbol}{}
See Appendix \ref{Assec: Proof of Proposition 1}.
\end{proof}

\begin{corollary}\label{Corollary: utilitarian vs Negishi}
    The utilitarian uniform carbon price is greater than the Negishi-weighted carbon price, that is $\check{\tau} > \tilde{\tau}$, if and only if $\frac{-\frac{d\check{D}_S}{d\check{\tau}}}{\frac{d\check{C}_S}{d\check{\tau}}} > 1 > \frac{-\frac{d\check{D}_N}{d\check{\tau}}}{\frac{d\check{C}_N}{d\check{\tau}}}$.
\end{corollary}

\begin{proof}[Proof]
\renewcommand{\qedsymbol}{}
See Appendix \ref{Assec: Proof of Corollary 1}.
\end{proof}

\todo[inline]{Check again why this needs to hold at the utilitarian solution and not generally? Do my assumptions in fact imply that it holds everywhere (I don't think so)? Do I introduce new monotonicity assumptions in the proof? Or does it just indicate in which direction the utilitarian solution has moved away from maximizing dollars?}

Proposition \ref{Proposition: utilitarian vs Negishi} establishes that the welfare-maximizing uniform carbon is greater than the Negishi-weighted carbon price if and only if the relative benefits of increasing global abatement, for the Global South compared to the Global North, exceed the relative costs.
The left-hand side, $\frac{\check{D}'_S}{\check{D}'_N}$, is the relative benefit of an extra unit of global abatement $A$. The right-hand side, $\frac{C''_N}{C''_S}$ is the relative cost of an extra unit of global abatement. 
Since the marginal abatement cost (MAC) is equal across regions, the relative cost of an extra unit of aggregate abatement is determined by the relative fractions of that unit of global abatement that are provided by each region, which in turn is determined by the relative slopes of the MAC function.
A steeper MAC function results in a smaller abatement increase, and therefore a smaller increase in abatement costs\footnote{
To see this, notice that 
    $\frac{C''_N}{C''_S}
    =
    \frac{\frac{d\check{A}_S}{d\check{\tau}}}{\frac{d\check{A}_N}{d\check{\tau}}}
    =
    \frac{\frac{d\check{A}_S}{d\check{A}}}{\frac{d\check{A}_N}{d\check{A}}}
    =
     \frac{\frac{d\check{C}_S}{d\check{A}}}{\frac{d\check{C}_N}{d\check{A}}},$
where $\frac{dA_i}{d\tau_i} = \frac{1}{C_i''}$,
and the third equality follows from $\frac{d\check{C}_S}{d\check{A}_S} = \frac{d\check{C}_N}{d\check{A}_N}$.
}.
Appendix \ref{Assec: Generalization of Proposition 1} generalizes Proposition \ref{Proposition: utilitarian vs Negishi} to $n$ regions and shows that the same insights carry over. 

Using the simplified RICE functions, Proposition \ref{Proposition: utilitarian vs Negishi} can also be expressed in terms of the damage and abatement cost functions per endowment, allowing for a more straightforward comparison of economies of different sizes\footnote{
Here, $d'_i = \frac{dd_i(A)}{dA}$ and $c''_i =\frac{d^2 c_i\left(\frac{A_i}{W_i}\right)}{d\left(\frac{A_i}{W_i}\right)^2}$. Also note that $D'_i = W_i d'_i$ and
$C''_i = c''_i \frac{1}{W_i}$.
}:
\begin{equation*}
    \check{\tau} > \tilde{\tau}
    \iff 
    \frac{\check{d}'_S}{\check{d}'_N} > \frac{c''_N}{c''_S}.
\end{equation*}





Corollary \ref{Corollary: utilitarian vs Negishi} provides an additional piece to understand the condition under which the utilitarian uniform carbon price exceeds the Negishi-weighted carbon price. It states that this is the case if and only if, at the utilitarian uniform carbon price, the ratio of the marginal benefits of abatement to the marginal costs of abatement from marginally increasing the carbon price is greater than one for the South and less than one for the North. Intuitively, this implies that the South would benefit from further increasing the carbon price while the North would be made worse off. The corollary shows that this is necessary and sufficient for the utilitarian uniform carbon price to be greater than the Negishi-weighted carbon price.

We may also be interested in how different factors affect the magnitude of the difference between the two carbon prices. To this end, it is useful to define the \textit{utilitarian-Negishi uniform carbon price ratio}, $\check{\tau}/\tilde{\tau}$. Using the simplified RICE functions (which allow for an easier interpretation), this ratio is given by
\begin{gather}
    \label{E: carbon price ratio (static)}
\begin{aligned}
    \frac{\check{\tau}}{\tilde{\tau}}
    &=
    \frac{\check{u}'_{N} L_{N} w_{N} \check{d}'_{N} + \check{u}'_{S} L_{S} w_{S} \check{d}'_{S}}
    {L_{N} w_{N} \tilde{d}'_{N} + L_{S} w_{S} \tilde{d}'_{S}}
    \frac{c''_{S} \frac{1}{L_{S} w_{S}} + c''_{N} \frac{1}{L_{N} w_{N}}}
    {\check{u}'_{N} c''_{S} \frac{1}{L_{S} w_{S}} + \check{u}'_{S} c''_{N} \frac{1}{L_{N} w_{N}}} \\
    &\approx
    \frac{\frac{L_{S}}{L_{N}} \left(\frac{w_S}{w_N}\right)^{1-\eta} \frac{{d}'_{S}}{{d}'_{N}} + 1}
    {\frac{L_{S}}{L_{N}} \frac{w_S}{w_N} \frac{{d}'_{S}}{{d}'_{N}} +  1}
    \frac{\frac{L_{S}}{L_{N}} \frac{w_{S}}{w_{N}} \frac{c''_{N}}{c''_{S}} + 1 }
    {\frac{L_{S}}{L_{N}} \left(\frac{w_S}{w_N}\right)^{1-\eta} \frac{c''_{N}}{c''_{S}} + 1},
\end{aligned}
\end{gather}
where the second line assumes that the utility function is isoelastic, $u(x) = \frac{x^{1-\eta}}{1-\eta}$ (for $\eta=1$,  $u(x) = \log(x)$, where $\eta$ is the elasticity of the marginal utility of consumption), marginal damages are approximately equal in the utilitarian and Negishi solutions, $\check{d}'_i \approx \tilde{d}'_i$, and the per capita consumption and endowment ratios are approximately equal, $\frac{x_{S}}{x_{N}} \approx \frac{w_{S}}{x_{N}}$. The latter two approximations are useful because they allow us to write the utilitarian-Negishi uniform carbon price ratio simply as a function of the ratios of variable values in the South compared to the North.

Using these approximations, Table \ref{T: Utilitarian-Negishi uniform carbon price ratio (static)} illustrates how the carbon price ratio is affected by the abatement cost and damage functions, inequality and inequality aversion. The default values of the population and endowment per capita ratios are $\frac{L_S}{L_N} = 3.7$ and $\frac{w_N}{w_S} = 3.2$, respectively, which are calibrated to empirical data in 2023 \parencite{world_economics_global_2024}\footnote{
The endowment per capita ratio is calibrated to the GDP per capita ratio (in PPP terms).
}. 

\renewcommand{\arraystretch}{1.4}
\begin{table}[!htb]
    \centering
\begin{threeparttable}
    \caption{Utilitarian-Negishi uniform carbon price ratio ($\check{\tau}/\tilde{\tau}$). Static model.}
    \setlength{\tabcolsep}{13pt}
\begin{tabular}{lcccccccc}
\hline \hline
            & & \multicolumn{3}{c}{$\eta = 1$} &  & \multicolumn{3}{c}{$\eta = 1.5$} \\ \cline{3-5} \cline{7-9} 
\multicolumn{1}{l}{$d_{S}'/d_{N}'$:} & & 0.5 & 1            & 2            &  & 0.5 & 1             & 2             \\ \hline
\multicolumn{9}{l}{\textbf{A) Abatement costs}}                            \\
$c_N''/c_S'' = 0.5$ & & 1.00 & 1.21         & 1.40         &  & 1.00 & 1.29          & 1.57          \\
$c_N''/c_S'' = 1$ & & 0.83 & 1.00         & 1.16         &  & 0.77 & 1.00          & 1.22          \\
$c_N''/c_S'' = 2$ & & 0.71 & 0.86         & 1.00         &  & 0.64 & 0.82          & 1.00          \\
\multicolumn{9}{l}{\textbf{B) Inequality}}                              \\
$w_N/w_S = 1$ & & 1.00 & 1.00         & 1.00         &  & 1.00 & 1.00          & 1.00          \\
$w_N/w_S = 3.2$ & & 0.83 & 1.00         & 1.16         &  & 0.77 & 1.00          & 1.22          \\
$w_N/w_S = 6.4$ & & 0.74 & 1.00         & 1.31         &  & 0.67 & 1.00          & 1.39          \\ \hline \hline
\end{tabular}\scriptsize
\label{T: Utilitarian-Negishi uniform carbon price ratio (static)}
    \begin{minipage}{\linewidth}
        \vspace{1mm}  
        \footnotesize
        \textit{Notes}: The carbon price ratios are approximations based on Equation \eqref{E: carbon price ratio (static)}. Variable values that are not shown are set as follows: In both panels, $L_{S}/L_{N} = 3.7$. In panel A, $w_{N}/w_{S} = 3.2$. In panel B, $c''_{N}/c''_{S} = 1$.
    \end{minipage}
\end{threeparttable}
\end{table}

Panel A of Table \ref{T: Utilitarian-Negishi uniform carbon price ratio (static)} shows that relatively higher marginal damages and a more convex abatement cost function in the South increase the carbon price ratio. Confirming the insight from Proposition \ref{Proposition: utilitarian vs Negishi}, carbon prices are equal if $\frac{d'_S}{d'_N} = \frac{c''_N}{c''_S}$.
Panel B demonstrates that greater inequality amplifies the difference between the utilitarian and Negishi-weighted uniform carbon prices\footnote{This also suggests that accounting for inequality at a more granular resolution (e.g., across countries) may increase the carbon price ratio.}, as does a more concave utility function, which implies a higher inequality aversion.
Furthermore, there is no difference between the carbon prices if there is no inequality or if the utility function is linear (i.e., $\eta = 0$).


\todo[inline]{Finish graph and explain in the text.}

\subsubsection{Utilitarian differentiated versus Negishi}
I now turn to the utilitarian differentiated carbon price solution.
I explore when the utilitarian differentiated carbon price solution leads to higher or lower global emissions compared to the Negishi solution. 
I begin by establishing the following lemma.

\begin{lemma}\label{Lemma: Regional abatement Udiff vs Negishi}
    South's (North's) carbon price under the utilitarian differentiated carbon price solution is less (greater) than the Negishi-weighted carbon price. That is, $\hat{\tau}_S < \tilde{\tau} < \hat{\tau}_N$. Consequently, South's (North's) abatement level is lower (higher) in the utilitarian differentiated carbon price solution than in the Negishi solution; that is $\hat{A}_S < \tilde{A}_S$ and $\hat{A}_N > \tilde{A}_N$.
\end{lemma}
\begin{proof}
    \renewcommand{\qedsymbol}{}
    See Appendix \ref{Assec: Proof of Lemma 1}.
\end{proof}

Therefore, whether global abatement is higher or lower in the utilitarian differentiated carbon price solution than in the Negishi solution depends on whether the additional abatement in the North outweighs the reduced abatement in the South. Proposition \ref{Prop: Aggregate abatement Udiff vs Negishi} establishes the condition under which this holds,
and Appendix \ref{Assec: Generalization of Proposition 2} shows that the same insights extend to the $n$-region setting. 


\begin{proposition}
\label{Prop: Aggregate abatement Udiff vs Negishi}
    The global abatement under utilitarian differentiated carbon prices is greater than under the Negishi-weighted carbon price, that is $\hat{A} > \tilde{A}$, if and only if $\frac{\hat{u}_S'}{\hat{u}_N'} \frac{\hat{D}'_S}{\hat{D}'_N} > \frac{C''_N}{C''_S}$.
\end{proposition}

\begin{proof}
    \renewcommand{\qedsymbol}{}
    See Appendix \ref{Assec: Proof of Proposition 2}.
\end{proof}

The first thing to notice is the similarity of this condition with the corresponding condition for the comparison between the utilitarian uniform carbon price and the Negishi solution detailed in Proposition \ref{Proposition: utilitarian vs Negishi}. The aggregate abatement is again more likely to be higher under the utilitarian solution if the South has relatively high marginal damages and a steep marginal abatement cost curve, compared to the North.

However,  there is an additional term in the condition of Proposition \ref{Prop: Aggregate abatement Udiff vs Negishi}; the ratio of marginal utilities of consumption, $\frac{\hat{u}_S'}{\hat{u}_N'}$. Thus, the marginal damages in the two regions are weighted by their respective marginal utilities, reflecting marginal damages in welfare terms (as opposed to monetary terms). For a poorer South,  $\hat{u}_S' > \hat{u}_N'$ and hence $\frac{\hat{u}_S'}{\hat{u}_N'} > 1$.
The important implication is that the aggregate abatement in the utilitarian differentiated carbon price solution is more likely to be greater than in the Negishi solution if consumption inequality is large.

The attentive reader may wonder why the marginal utilities only appear on the left-hand side of the inequality (representing the relative benefits of abatement), but not on the right-hand side (concerning the costs of abatement). The intuition for this is as follows. The difference in marginal utilities is already accounted for in the region-specific carbon prices which equalize the marginal welfare costs of abatement (i.e., $\hat{u}_N' \hat{C}'_N = \hat{u}_S' \hat{C}'_S$). Consequently, the carbon price in the poorer region is lower because of its higher marginal utility. The term on the right-hand side, $\frac{C''_N}{C''_S}$, simply determines how much the abatement decreases in the South and increases in the North (relative to the Negishi solution). 
A relatively steeper marginal abatement cost in the South and a flatter one in the North make it more likely that the aggregate abatement increases.
It is also worth noting the subtle, but important, difference in intuition behind the $\frac{C''_N}{C''_S}$ term in Propositions \ref{Proposition: utilitarian vs Negishi} and \ref{Prop: Aggregate abatement Udiff vs Negishi}. In Proposition  \ref{Proposition: utilitarian vs Negishi}, this term reflects the relative abatement cost increases to the two regions as a result of a marginal increase in a uniform carbon price. 
In contrast, in Proposition  \ref{Prop: Aggregate abatement Udiff vs Negishi}, it reflects how much abatement in the South decreases and how much it increases in the North when we allow for differentiated carbon prices.

\subsection{Aggregating heterogeneous climate policy preferences}\label{sssec: A region's preferred uniform carbon price}

To sharpen how heterogeneous climate policy preferences map into a single globally uniform carbon price under different social welfare functions, this section derives each region’s preferred uniform carbon price. 
Doing so translates the choice of regional welfare weights into a more tangible preference-aggregation problem: how are heterogeneous regional policy preferences combined into one global price?
In doing so, I establish connections to \textcite{weitzman_can_2014} and \textcite{kotchen_which_2018}, who introduced the notions of preferred uniform carbon prices and the preferred social cost of carbon, respectively.

\subsubsection{Regions' preferred uniform carbon prices}

The preferred uniform carbon price for a region is obtained by solving the uniform carbon price optimization problem with welfare weights fully assigned to that region; thus,  $\alpha_i = 1$ and $\alpha_{-i} = 0$. 

\begin{definition}
    The \textbf{preferred uniform carbon price of region i} is implicitly defined by
    \begin{gather}
    \label{Def: preferred uniform carbon price}
        \mathring{\tau}^i =
        {C}'_{i}(\mathring{A}_i^i) = 
        {C}'_{-i}(\mathring{A}_{-i}^i) =
        - {D}'_i(\mathring{A^i})
        \frac{{C}''_i + {C}''_{-i}}
        {{C}''_{-i}},
    \end{gather}
\end{definition}
\noindent where the superscript $i$ indicates that the functions are evaluated at the solution under the preferred uniform carbon price of region $i$ (for example, $\mathring{A}_S^N$ is the abatement in the South under the preferred uniform carbon price of the North).

Equation~\eqref{Def: preferred uniform carbon price} reveals that a region's preferred uniform carbon price is higher when its marginal benefit of abatement is large and when its abatement cost function is more convex compared to the other region. 
Put simply, this is the case if a region is particularly vulnerable to climate change and if the cost burden of raising a uniform carbon price falls predominantly on the other region\footnote{
To see this, note that
$\frac{{C}''_i + {C}''_{-i}}{{C}''_{-i}}
= 1 + \frac{\frac{dA_{-i}}{d \mathring{\tau}^i}}{\frac{dA_{i}}{d \mathring{\tau}^i}}
= 1 + \frac{\frac{d\mathring{C}_{-i}}{d \mathring{\tau}^i}}{\frac{d\mathring{C}_{i}}{d \mathring{\tau}^i}}$, where the second equality follows from $\frac{d\mathring{C}_{-i}}{d A_{-i}} = \frac{d\mathring{C}_{i}}{d A_{i}}$.
}.
The crucial role of the relative convexities of abatement cost functions for regions’ preferred uniform carbon prices has received limited attention in the existing literature\footnote{
Most studies assume uniform convexities of the abatement cost function across regions \parencite{weitzman_can_2014, weitzman_voting_2017, kotchen_which_2018}.
\textcite{weitzman_world_2017} allows for different convexities of the abatement cost function across regions, but does not highlight their role. 
}.
It is also worth noting that a region's preferred uniform carbon price is greater than its own marginal benefit of abatement, $- {D}'_i$. 
This is because region $i$ accounts for the fact that increasing a uniform carbon price results in additional abatement in the other region. This is represented by the term 
$\frac{{C}''_i + {C}''_{-i}}{{C}''_{-i}} 
=  1 + \frac{A_{-i}'(\mathring{\tau}^i)}{A'_{i}(\mathring{\tau}^i)}
>1$, 
where $A_i'(\tau_i) \equiv \frac{dA_i(\tau_i)}{d\tau_i}$.

It is again instructive to rewrite the optimality condition in Equation~\eqref{Def: preferred uniform carbon price} in terms of the derivatives with respect to the uniform carbon price (see Appendix \ref{Asssec: Regions' preferred uniform carbon price}):
\begin{gather}
\label{E: Preferred uniform carbon price wrt carbon price}
   \frac{d C_i(\mathring{A}_i(\mathring{\tau}^i))}{d \mathring{\tau}^i}
       = 
        - \frac{d{D}_i(\mathring{A}(\mathring{\tau}^i))}{d\mathring{\tau}^i}.
    \end{gather}
Intuitively, the preferred uniform carbon price of region $i$ equalizes the cost and benefits to region $i$  from marginally increasing the uniform carbon price.

\subsubsection{The connection to optimal carbon prices under different welfare weights}

Next, we ask how the preferred uniform carbon prices relate to the optimal uniform carbon prices under the utilitarian solution and the Negishi solution.
I begin by establishing the following lemma, which helps to build intuition and acts as a building block towards proving the proposition that follows. 

\begin{lemma}
    The utilitarian uniform carbon price ($\check{\tau}$) and the Negishi-weighted carbon price ($\tilde{\tau}$) lie strictly between the preferred uniform carbon prices of the Global North ($\mathring{\tau}^N$) and the Global South ($\mathring{\tau}^S$), unless all prices coincide\footnote{
    Strict inequality holds under the assumption that both regions’ marginal utilities are positive and finite, and unless all prices coincide.
    }.
    \label{Lemma: U&N in between preferred carbon price}
\end{lemma}

\begin{proof}
\renewcommand{\qedsymbol}{}
    See Appendix \ref{Assec: Proof of Lemma 2}.
\end{proof}

The intuition behind Lemma \ref{Lemma: U&N in between preferred carbon price} is as follows. Regions' preferred uniform carbon prices are obtained by using ``edge weights" in the SWF, giving full weight to one region and zero weight to the other.
Utilitarian and Negishi weights are linear combinations of these edge weights, giving a positive weight to both regions.
It is therefore not surprising that ``edge weights" results in more extreme carbon prices than ``more balanced" welfare weights.

Using Lemma \ref{Lemma: U&N in between preferred carbon price}, I establish the following relationship between regions' preferred uniform carbon prices and the main result detailed in Proposition \ref{Proposition: utilitarian vs Negishi}. 

\begin{proposition}
    The utilitarian uniform carbon price is greater than the Negishi-weighted carbon price, that is $\check{\tau} > \tilde{\tau}$, if and only if the preferred uniform carbon price of the Global South is greater than the preferred uniform carbon price of the Global North, that is $\mathring{\tau}^S > \mathring{\tau}^N$.
    \label{Proposition: U>N iff S>N}
\end{proposition}

\vspace{-2.1em}
\begin{proof}
\renewcommand{\qedsymbol}{}
    See Appendix \ref{Assec: Proof of Proposition 3}.
\end{proof}

The intuition for the result of Proposition \ref{Proposition: U>N iff S>N} builds on the logic behind Lemma \ref{Lemma: U&N in between preferred carbon price}. Giving a positive weight to both regions, the utilitarian uniform carbon price and the Negishi-weighted carbon price can be understood as ``weighted averages'' of regions' preferred uniform carbon prices, where the welfare weights determine the relative weight given to the preferences of the two regions. 
Since Negishi weights downweight the South,  it is intuitive that the utilitarian uniform carbon price is greater than the Negishi-weighted carbon price if the South prefers a higher uniform carbon price than the North.
This result provides perhaps the clearest intuition for the conditions under which the utilitarian uniform carbon price is higher or lower than the Negishi-weighted carbon price: it simply depends on whether South's preferred uniform carbon price is greater or lower than North's. 



\subsubsection{A voting interpretation}

The next proposition goes one step further: under a set of additional assumptions, it makes explicit how the utilitarian and Negishi objectives aggregate regions’ heterogeneous policy preferences into a single globally uniform carbon price.

\begin{proposition}
    \label{Prop: voting}
    Suppose marginal damages are constant and that the slope of each region’s marginal abatement cost function is inversely proportional to its endowment, that is, $C_i'' = \frac{\kappa}{W_i}$, for some constant $\kappa > 0$ common across regions.
    
    \begin{itemize}
        \item[(i)] Then, the Negishi-weighted carbon price is the endowment-weighted average of regions’ preferred uniform carbon prices:
        \[
            \tilde{\tau} = \sum_i \frac{W_i}{\sum_j W_j} \, \mathring{\tau}^i.
        \]
        
        \item[(ii)] If, in addition, utility is logarithmic, $u(x_i) = \log(x_i)$, and the per capita consumption and endowment ratios are approximately equal, $\frac{\check{x}_S}{\check{x}_N} \approx \frac{w_S}{w_N}$, then the utilitarian uniform carbon price approximately equals the population-weighted average of regions’ preferred uniform carbon prices:
        \[
            \check{\tau} \approx \sum_i \frac{L_i}{\sum_j L_j} \, \mathring{\tau}^i.
        \]
    \end{itemize}
\end{proposition}

\begin{proof}
\renewcommand{\qedsymbol}{}
    See Appendix \ref{Assec: Proof of Proposition 4}.
\end{proof}

The key nonstandard restriction is $C_i'' = \frac{\kappa}{W_i}$, which captures a scale effect: all else equal, a larger economy has proportionally more emissions and thus more abatement opportunities in each marginal cost interval, implying a flatter marginal abatement cost function.
The approximation $\frac{\check{x}_S}{\check{x}_N} \approx \frac{w_S}{w_N}$ is accurate when baseline inequality dominates climate-induced consumption changes, so per-capita consumption remains roughly proportional to endowments.
The assumptions in Proposition \ref{Prop: voting} are relaxed in Appendix \ref{Assec: Generalizations of Proposition 4}. 

Proposition \ref{Prop: voting} admits a simple voting interpretation: under the stated assumptions, the Negishi solution behaves like ``one dollar, one vote'', whereas the utilitarian solution behaves (approximately) like ``one person, one vote''\footnote{These voting analogies are not definitional; they emerge from the maintained structure in Proposition \ref{Prop: voting} and, for utilitarianism, from the approximation in part (ii).}.
Notably, the ``one person, one vote'' interpretation lends the utilitarian uniform carbon price an intuitive democratic appeal\footnote{\textcite[583]{weitzman_world_2017} also notes that ``the inherent democracy of one person, one vote is an attractive feature in and
of itself'', though without connecting this interpretation to a utilitarian social welfare function.}.

Two additional points bear emphasis.
First, to the extent that the utilitarian optimum departs from “one person, one vote”, this departure raises utilitarian welfare---by construction, $\check{\tau}$ maximizes utilitarian welfare among uniform carbon prices.
Second, \textcite{weitzman_can_2014} notes that aggregating preferences based on “one person, one vote” can yield a uniform carbon price close to the efficient benchmark, and he suggests this as an attractive feature\footnote{Note, however, that in Weitzman’s model, “one person, one vote” corresponds to population-weighted majority voting, where the voting outcome is the population \textit{median} of preferred prices.}.
The connection to social welfare functions I establish here flips the logic from a utilitarian perspective: to the extent that the utilitarian uniform carbon price, approximating ``one person, one vote'', is different from the efficient Negishi-weighted carbon price, this difference can be viewed as a feature rather than a drawback, because it raises utilitarian welfare.

\todo[inline]{Maybe provide a table of optimality conditions}

\subsection{Extension: Dynamic model}\label{sssec: Extension: Dynamic model}

An important aspect of climate change is that emission reductions today reduce the impact of climate change in the future. To capture this temporal dimension, I consider a two-period model in this section, which I refer to as ``dynamic''. 
I focus on uniform carbon prices in this extension to illustrate how welfare weights affect optimal carbon prices in a dynamic setting, even when the policy instrument is identical---a globally uniform carbon price.

\subsubsection{Model modifications}
The objective of the dynamic model is to account for the fact that the benefits of abatement come with a delay. To capture this in the simplest way, I assume that abatement occurs in the first period and climate damages in the second period. Aggregate regional consumption is thus given by
$X_{i1} = W_{i1} - C_{i1}(A_{i})$ and
$X_{i2} = W_{i2} - D_{i2}(A)$,
where the second subscript denotes the period $t \in \{1,2\}$.

\subsubsection{Optimal carbon prices}
The optimal uniform carbon prices for the dynamic model are obtained by solving the following optimization problem:
\begin{gather}
\begin{aligned}
    \max_{X_{it}, A_{i1}}
    & \sum_t \sum_i \beta^{t-1} L_{it} \alpha_{it} u_{it}\left(\frac{X_{it}}{L_{it}}\right) \\
\end{aligned} \\
    \begin{aligned}
        \label{E: constraints_diff2period}
    \text{subject to: } 
    &X_{i1} = W_{i1} - C_{i1}(A_{i}), \quad \forall i, \\
    &X_{i2} = W_{i2} - D_{i2}(A), \quad \forall i,\\
    &C'_{N1}(A_{N}) = C'_{S1}(A_{S}),
    \end{aligned}
\end{gather}
where $\beta^{t-1}$ is the utility discount factor (given by $\beta^{t-1}=(1+\rho)^{1-t}$,  where $\rho$ is the utility discount rate or pure rate of time preference). 


The welfare weights are defined as follows. Utilitarian weights are uniform across regions and periods and set to unity; that is, $\alpha_{it}^U = 1$.
Negishi weights are time-variant and defined in accordance with the RICE model:
$  \tilde{\alpha}_{i1} = \frac{1}{\tilde{u}'_{i1}}$ and
$\tilde{\alpha}_{i2} = v \frac{1}{\tilde{u}'_{i2}},$
where $v = \pi \frac{\tilde{u}'_{N2}}{\tilde{u}'_{N1}} + (1 - \pi) \frac{\tilde{u}'_{S2}}{\tilde{u}'_{S1}}$ is the wealth-based component of the social discount factor\footnote{
For a model with a single representative agent, the wealth-based component of the social discount factor is approximated by $\frac{1}{1 + \eta g}$, where $\eta$ is the elasticity of the marginal utility of consumption and $g$ is the growth rate in per capita consumption. Note that $\eta g$ is the wealth-based component of the social discount rate (SDR) in the Ramsey Rule, $SDR \approx \rho + \eta g$, reflecting the rationale for discounting future consumption if future generations are richer.}, which is pinned down as a weighted average of the regional wealth-based discount factors. The discounting weights $\pi \in (0,1)$ and $(1-\pi)$ are given by the regional capital or output shares in previous versions of the RICE model \parencite{nordhaus_warming_2000, nordhaus_economic_2010}. I consider general discounting weights, unless explicitly specified.

Solving the optimization problem above yields the following optimal carbon prices\footnote{
The derivation is largely analogous to the static model (see Appendix \ref{Assec: Derivation of optimal carbon prices}).
}.

\begin{definition}
\label{Def: dynamic Negishi-weighted carbon price}
    The \textbf{dynamic Negishi-weighted carbon price} is implicitly defined by
    \begin{gather}\label{Def: dynamic Negishi-weighted carbon price}
        \tilde{\tau} = C'_{i1}(\tilde{A}_i) = - v \beta 
        \sum_{i} D'_{i2}(\tilde{A}).
    \end{gather}
\end{definition}

\begin{definition}
    \label{Def: dynamic utilitarian uniform carbon price}
    The \textbf{dynamic utilitarian uniform carbon price} is implicitly defined by
    \begin{gather}
        \check{\tau} =
        \tilde{C}'_{i1} =
          - \beta \sum_i {u}'(\check{x}_{i2}) {D}'_{i2}(\check{A})
        \frac{{C}''_{S1} + {C}''_{N1}}
        {{u}'(\check{x}_{N1}) {C}''_{S1} + {u}'(\check{x}_{S1}) {C}''_{N1} } 
    \end{gather}
\end{definition}


These expressions are similar to the ones of the static model with the important difference that damages occur in the second period (and are discounted) while abatement takes place in the first period.
Consequently, optimal carbon prices are generally affected by the developments of endowment, consumption per capita, and population, which is the focus of the comparative analysis below. Moreover,  discounting is affected by the choice of welfare weights; while the utility discount factor is assumed to be the same, the wealth-based component of the social discount factor differs. Under Negishi weights, the wealth-based component of the social discount factor is $v$, which is uniform across regions and given by the weighted average of the regional wealth-based discount factors\footnote{
\textcite{dennig_note_2019} and \textcite{anthoff_differentiated_2021}
show that this distorts regional time-preferences.}.
In contrast, under utilitarian weights, it is simply the regional wealth-based discount factor, $u'(x_{i2})/u'(x_{i1})$, for each region. Notably, utilitarian weights value consumption across regions in the same fashion as across periods.

\todo[inline]{Dynamic region's preferred uniform carbon price + propositions}

\subsubsection{Comparative results}
As before, the central question is how the utilitarian uniform carbon price compares to the Negishi-weighted carbon price. The following proposition establishes this relationship.

\begin{proposition}\label{Proposition: utilitarian vs Negishi (dynamic)}
    The dynamic utilitarian uniform carbon price is greater than the dynamic Negishi-weighted carbon price, that is $\check{\tau} > \tilde{\tau}$, if and only if 
    \begin{gather}
    \frac{\check{u}'_{N2} \check{D}'_{N2} 
        + \check{u}'_{S2} \check{D}'_{S2}}
        {\check{D}'_{N2} + \check{D}'_{S2}}
        > 
        v 
        \frac{{\check{u}'_{N1} C''_{S1} + \check{u}'_{S1} C''_{N1}} }           {C''_{S1} + C''_{N1}},
    \end{gather}
 where $v = \pi \frac{\tilde{u}'_{N2}}{\tilde{u}'_{N1}} + (1 - \pi) \frac{\tilde{u}'_{S2}}{\tilde{u}'_{S1}}$ .
\end{proposition}

\begin{proof}
\renewcommand{\qedsymbol}{}
    See Appendix \ref{Assec: Proof of Proposition 5}.
\end{proof}

As in the static model, this condition is more likely to be satisfied if the South has relatively higher marginal damages and a more convex abatement cost function. 
All else equal, this is also the case for a lower wealth-based component of the social discounting factor under the Negishi-weighted SWF, $v$.

Crucially, the damage and abatement cost functions generally depend on the economy size, which in turn depends on the population size\footnote{To keep the exposition simple, I assume that the endowment per capita is exogenous and does not depend on the population size.}. 
Since the costs and benefits of abatement occur in different periods, economic and population growth affect the relative regional costs and benefits of abatement.
Using the simplified RICE functions, we can rewrite the condition in Proposition \ref{Proposition: utilitarian vs Negishi (dynamic)} as
\begin{gather}
    \label{E: utilitarian vs Negishi (dynamic, RICE)}
    \frac{\check{u}'_{N2} L_{N1} g^L_N w_{N1} g^w_N \check{d}'_{N2} 
    + \check{u}'_{S2} L_{S1} g^L_S w_{S1} g^w_S \check{d}'_{S2}}
    {L_{N1} g^L_N w_{N1} g^w_N \check{d}'_{N2} + L_{S1} g^L_S w_{S1} g^w_S \check{d}'_{N2}}
    > 
    v 
    \frac{{\check{u}'_{N1} \frac{1}{L_{S1} w_{S1}} c''_{S1} + \check{u}'_{S1} \frac{1}{L_{N1} w_{N1}} c''_{N1}} }           
    {\frac{1}{L_{S1} w_{S1}} c''_{S1} + \frac{1}{L_{N1} w_{N1}} c''_{N1}},
\end{gather}
where I used $L_{i2} = L_{i1} g^L_i$  and $w_{i2} = w_{i1} g^w_i$, with $g^L_i$ and $g^w_i$ denoting the population and economic growth factors, respectively.

Equation \eqref{E: utilitarian vs Negishi (dynamic, RICE)} yields a central result: if population growth is faster in the South than the North, then the utilitarian uniform carbon price is more likely to exceed the Negishi-weighted carbon price. 
The intuition is that relatively faster population growth in the South increases the relative damages of climate change in the South, as they manifest in the future, which are given comparatively less weight under the Negishi-weighted SWF.
Simply put, climate change is a bigger problem for the South if its population is growing faster, as this results in more people being harmed by climate change\footnote{
Equivalently, a larger population results in a bigger economy, thereby increasing aggregate marginal damages (which are assumed to be proportional to the economy size). To see the different interpretations formally, note that
$D' = L w d' = W d'$.
Population growth effectively plays an analogous (but opposite) role to time discounting, a point that was formally made by \textcite{budolfson_optimal_2018}.
}.

The role of economic growth is more complicated. This is because economic growth simultaneously affects climate damages and the development of marginal utilities of consumption, which affect discounting under both SWFs\footnote{
For a thorough examination of how interregional inequality and heterogeneous economic growth impact the discount rate under the utilitarian SWF, see \textcite{gollier_discounting_2015}.
} (note that it also affects $v$). However, we can gain traction on the role of economic growth with additional assumptions. 
To obtain intuition for the role of economic growth,
it is again useful to define the utilitarian-Negishi uniform carbon price ratio. Using the simplified RICE functions, this ratio is given by
\begin{gather}
    \label{E: carbon price ratio (dynamic)}
\begin{aligned}
    \frac{\check{\tau}}{\tilde{\tau}}
    &=
    \frac{1}{\pi \frac{\tilde{u}'_{N2}}{\tilde{u}'_{N1}} + (1-\pi) \frac{\tilde{u}'_{S2}}{\tilde{u}'_{S1}}}
    \frac{\check{u}'_{N2} L_{N2} w_{N2} \check{d}'_{N2} + \check{u}'_{S2} L_{S2} w_{S2} \check{d}'_{S2}}
    {L_{N2} w_{N2} \tilde{d}'_{N2} + L_{S2} w_{S2} \tilde{d}'_{S2}}
    \frac{c''_{S1} \frac{1}{L_{S1} w_{S1}} + c''_{N1} \frac{1}{L_{N1} w_{N1}}}
    {\check{u}'_{N1} c''_{S1} \frac{1}{L_{S1} w_{S1}} + \check{u}'_{S1} c''_{N1} \frac{1}{L_{N1} w_{N1}}} \\
    &\approx
    \frac{\frac{L_{S1}}{L_{N1}} \frac{g^L_S}{g^L_N} \frac{ w_{S1}}{w_{N1}} \frac{g^w_S}{g^w_N} + 1}
    {\frac{L_{S1}}{L_{N1}} \frac{g^L_S}{g^L_N} \frac{w_{S1}}{ w_{N1}} \left(\frac{g^w_S}{g^w_N}\right)^{1-\eta} +1}
    \frac{  \frac{L_{S1}}{L_{N1}} \frac{g^L_S}{g^L_N}  \left(\frac{w_{S1}}{w_{N1}} \frac{g^w_S}{g^w_N}\right)^{1-\eta} \frac{{d}'_{S2}}{{d}'_{N2}} +1}
    {\frac{L_{S1}}{L_{N1}} \frac{g^L_S}{g^L_N} \frac{w_{S1}}{w_{N1}} \frac{g^w_S}{g^w_N}\frac{{d}'_{S2}}{{d}'_{N2}} +1}
    \frac{\frac{c''_{N1}}{c''_{S1}} \frac{L_{S1}}{L_{N1}} \frac{w_{S1}}{w_{N1}} +1}
    {\left(\frac{w_{S1}}{w_{N1}}\right)^{-\eta}  \frac{c''_{N1}}{c''_{S1}} \frac{L_{S1}}{L_{N1}} \frac{w_{S1}}{w_{N1}} +1} .
\end{aligned}
\end{gather}
The second line utilizes the following assumptions and approximations:
(1) the utility function is isoelastic, $u(x) = \frac{x^{1-\eta}}{1-\eta}$ (for $\eta=1$,  $u(x) = \log(x)$), 
(2) the discounting weights are given by the regional endowment shares\footnote{
Both endowment and capital shares have been used in previous versions of the RICE model \parencite{nordhaus_warming_2000, nordhaus_economic_2010}. The RICE-2010 model uses capital shares but both approaches are numerically close, according to \textcite{nordhaus_warming_2000}.
}, $\pi = \frac{W_{N2}}{\sum_i W_{i2}}$, 
(3) per capita consumption and endowment growth are approximately equal, $g^x_i \approx g^w_i$, and the per capita consumption and endowment ratios are approximately equal, $\frac{x_{St}}{x_{Nt}} \approx \frac{w_{St}}{x_{Nt}}$, 
(4) per capita consumption growth and marginal damages are approximately equal in the utilitarian and Negishi solutions,  $\check{g}^x_i \approx \tilde{g}^x_i$ and $\check{d}'_i \approx \tilde{d}'_i$. 
These assumptions will generally not hold precisely but can be expected to be good approximations, serving the purpose of obtaining clean intuitions for the role of economic growth.

Using these approximations, the utilitarian-Negishi uniform carbon price ratio only depends on ratios of variable values in the South compared to the North. To demonstrate the role of population and economic growth, an illustrative numerical example of carbon price ratios is shown in Table \ref{T: Utilitarian-Negishi uniform carbon price ratio (dynamic)}. For these calculations, I assume that the first and second periods are 50 years apart and the growth factors are given by $g^y_i = (1 + \bar{g}^y_i)^t$, where $y\in\{L,w\}$, $\bar{g}^y_i$ are the annual growth rates and $t=50$. 

\todo[inline]{Mention data source for population and GDP per capita ratios.}

\renewcommand{\arraystretch}{1.4}
\begin{table}[!htb]
    \centering
\begin{threeparttable}
    \caption{Utilitarian-Negishi uniform carbon price ratio ($\check{\tau}/\tilde{\tau}$). Dynamic model.}
    \setlength{\tabcolsep}{13pt}
\begin{tabular}{lcccccc}
\hline \hline
            & & \multicolumn{2}{c}{$\eta = 1$} &  & \multicolumn{2}{c}{$\eta = 1.5$} \\ \cline{3-4} \cline{6-7} 
\multicolumn{1}{l}{$d_{S2}'/d_{N2}'$:} & & 1            & 2            &  & 1             & 2             \\ \hline
\multicolumn{7}{l}{\textbf{A) Population growth}}                            \\
$\bar{g}^L_S=0\%$, $\bar{g}^L_N=0\%$ & & 1.00         & 1.16         &  & 1.00          & 1.22          \\
$\bar{g}^L_S=1\%$, $\bar{g}^L_N=0\%$ & & 1.12         & 1.26         &  & 1.16          & 1.34          \\
$\bar{g}^L_S=2\%$, $\bar{g}^L_N=0\%$ & & 1.22         & 1.33         &  & 1.29          & 1.44          \\
\multicolumn{7}{l}{\textbf{B) Economic growth}}                              \\
$\bar{g}^w_S=2\%$, $\bar{g}^w_N=2\%$ & & 1.22         & 1.33         &  & 1.29          & 1.44          \\
$\bar{g}^w_S=3\%$, $\bar{g}^w_N=2\%$ & & 1.22         & 1.27         &  & 1.23          & 1.30          \\
$\bar{g}^w_S=4\%$, $\bar{g}^w_N=2\%$ & & 1.22         & 1.23         &  & 1.16          & 1.18          \\ \hline \hline
\end{tabular}\scriptsize
\label{T: Utilitarian-Negishi uniform carbon price ratio (dynamic)}
    \begin{minipage}{\linewidth}
        \vspace{1mm}  
        \footnotesize
        \textit{Notes}: The carbon price ratios are approximations based on Equation \eqref{E: carbon price ratio (dynamic)}.
        Variable values that are not shown are set as follows: In both panels, $w_{N1}/w_{S1} = 3.2$, $L_{S1}/L_{N1} = 3.7$, $c''_{N1}/c''_{S1} = 1$. In panel A, $\bar{g}^w_S/g^w_N = 1$. In panel B, $\bar{g}^L_S = 2\%$, $\bar{g}^L_N = 0\%$.
    \end{minipage}
\end{threeparttable}
\end{table}

Panel A of Table \ref{T: Utilitarian-Negishi uniform carbon price ratio (dynamic)} confirms that faster population growth in the South increases the carbon price ratio. Importantly, this holds even if marginal damages per endowment are homogeneously distributed across regions (i.e., $d'_{S2}=d'_{N2}$). 
Panel B examines the effect of faster economic growth in the South in terms of endowment per capita.
The first thing to note is that  economic growth plays no role if marginal damages per endowment are evenly distributed and $\eta=1$. 
However, economic growth reduces the carbon price ratio if either (1) the utility function is more concave than logarithmic utility ($\eta>1$) and $d'_{S2} \geq d'_{N2}$, or (2) the South has disproportionately high climate damages ($d'_{S2}>d'_{N2}$) and $\eta \geq 1$. 
Since climate damages are expected to be disproportionately large in the South, this last case is the most relevant in practice. Hence, faster economic growth in the South can be expected to reduce the carbon price ratio.

\todo[inline]{Dynamic model. (1) Proof of proposition. (2) Formal derivation of role of economic growth + Intuition. (3) Explain focus on uniform carbon prices. (4) Differentiated carbon prices in Appendix.}

\section{Simulations} \label{sec:Simulations}

This section presents the simulation results. Section \ref{sec:Method} introduces the RICE model and methodology. Section \ref{ssec: Results} discusses how optimal carbon prices are affected by the choice of welfare weights.

\subsection{Method} \label{sec:Method}


\subsubsection{Model} \label{sec:Model}

To provide simulation-based empirical evidence, I use the IAM Mimi-RICE-2010 \parencite{anthoff_mimi-rice-2010jl_2019}, which is an implementation of the RICE-2010 model \parencite{nordhaus_economic_2010} in the Julia programming language using the modular modeling framework Mimi. RICE is the regional variant of the Dynamic Integrated model of Climate and the Economy (DICE), disaggregating the world into 12 regions (see Figure~\ref{fig:RICE regions} for the region classification) \parencite{nordhaus_dice_2013}. It is based on a neoclassical optimal growth model, which is linked to a simple climate model. Economic production is determined by a Cobb-Douglas production function and results in industrial CO$_2$ emissions. The relationship between economic production and emissions depends on the emissions intensity of an economy, which can be reduced by investments in abatement. Emissions then translate to atmospheric CO$_2$ concentrations, radiative forcing, atmospheric and oceanic warming, and finally economic damages resulting from atmospheric temperature changes and sea-level rise.
Importantly, the functions that determine climate damages and abatement costs are region-specific (see Appendix \ref{Assec: Abatement cost and damage functions of the RICE model} for additional information).


\todo[inline]{Put damage and AC functions in the Appendix and comment on some key features here.}

\subsubsection{Optimizations} \label{sec:Optimizations}

I introduce one main modification to the Mimi-RICE-2010 model: the implementation of three different optimization problems.
The final model that includes these modifications is referred to as \textit{Mimi-RICE-plus}.
	
\paragraph{Optimization problems}\mbox{}\\[3pt]
\noindent \vspace{-6pt}The following three optimization problems are implemented:
\begin{enumerate}
    \item \textit{Negishi solution}: Maximization of the discounted Negishi-weighted SWF with no constraints on the marginal abatement costs and the interregional transfers\footnote{
    Note that regions are autarkic in the RICE-2010 model. Thus, the model implicitly contains a constraint of zero transfers. This is also the case in the optimization using the Negishi-weighted objective, even though in this case, zero transfers are also optimal under the Negishi-weighted SWF. 
    }. \vspace{-6pt}
    \item \textit{Utilitarian differentiated carbon price solution}: Maximization of the discounted utilitarian SWF with a constraint on the total level of interregional transfers, but with no constraint on the marginal abatement costs. \vspace{-6pt}
    \item \textit{Utilitarian uniform carbon price solution}: Maximization of the discounted utilitarian SWF with a constraint on the total level of interregional transfers, and an additional constraint of equalized marginal abatement costs across regions in each period. \vspace{-6pt}
\end{enumerate}

 In addition, I also compute regions' preferred uniform carbon prices by maximizing the respective regional SWFs (with welfare weights that equal unity for one region, and zero for all other regions) subject to a zero transfer constraint and a constraint of equalized marginal abatement costs across regions.

 The choice variables are the emissions control rates, which determine carbon prices\footnote{
 Note that I do not optimize the saving rates, as optimizing emission control rates and transfers in each period already results in long convergence times.  Moreover, assuming fixed saving rates is relatively common in the climate economics literature (see \textcite{golosov_optimal_2014, dennig_inequality_2015, budolfson_utilitarian_2021} for more information). I use the saving rates from the base scenario of the original RICE model.
 }.
This is described in more detail below.

\paragraph{Social welfare functions}\mbox{}\\[3pt]	
\noindent The first optimization problem is the maximization of the \textit{discounted Negishi-weighted SWF}
\begin{equation}
    \mathcal{W}^N=\sum_{t \in \mathscr{T}} \sum_{i \in \mathscr{I}} L_{it} \beta^t \tilde{\alpha}_{it} u\left(x_{it}\right)
\end{equation}

where $\mathscr{I}$ denotes the set of the 12 RICE regions, and $\mathscr{T} = \{0, 1, 2, ..., 590\} $ is the time horizon of the RICE model\footnote{
For clarity of exposition, I am omitting the detail that one time period in RICE represents 10 years.
}, corresponding to the model years 2005 to 2595, $L_{it}$ is the population, $x_{it}$ is the per capita consumption, $\beta^t$ is the utility discount factor (given by $\beta^t=(1+\rho)^{-t}$,  where $\rho$ is the utility discount rate), and $\tilde{\alpha}_{it}$ are the time-variant Negishi welfare weights. 
The utility function is given by
\begin{equation*}
    u\left(x_{it}\right)=\left\{\begin{array}{ll}
    \log \left(x_{it}\right) & \text { for } \eta=1 \\
    \frac{x_{it}{ }^{1-\eta}}{1-\eta}+1 & \text { for } \eta \neq 1
\end{array}\right.
\end{equation*}	
where $\eta$ is the elasticity of marginal utility of consumption, which is set to 1.5, consistent with the value employed in the original RICE model.

The time-variant Negishi weights are given by
\begin{equation}
    \label{E:def_NW}
    \tilde{\alpha}_{it} = \frac{1}{u^{\prime}\left(\tilde{x}_{it}\right)} v_t, 
\end{equation}
where $\tilde{x}_{it}$ is the consumption at the Negishi solution\footnote{
The Negishi weights are obtained by solving the optimization multiple times (in the presence of an implicit no transfer constraint, since regions in RICE-2010 are autarkic) and iteratively updating the weights until convergence.  
}, $v_t$ is the wealth-based component of the social discount factor.
In the RICE-2010 model, it is defined as the capital-weighted average of the regional wealth-based discount factors (see \textcite{nordhaus_economic_2010} and Appendix \ref{Assec: Time-variant Negishi weights} for more details). 


The second and third optimization problems maximize the
\textit{discounted utilitarian SWF}
\begin{equation}
\mathcal{W}^U=\sum_{t \in \mathscr{T}} \sum_{i \in \mathscr{I}} L_{it} \beta^t u\left(x_{it}\right) .
\end{equation}
 



\paragraph{Carbon prices}\mbox{}\\[3pt]	
\noindent In optimization problems (1) and (2), carbon prices are allowed to be differentiated across regions. However, recall that in the Negishi solution, uniform carbon prices are optimal by the construction of the Negishi weights. In the third optimization problem, a constraint of equal marginal abatement costs across regions is added\footnote{
The source code for the implementation of this constraint was adopted from the Mimi-NICE model \parencite{dennig_mimi_nice_2017}. 
}. 

\paragraph{Optimization algorithms}\mbox{}\\[3pt]	
\noindent  The optimization problems are solved with the numerical algorithm “\texttt{NLOPT\_LN\_SBPLX}” which is an implementation of the Subplex algorithm \parencite{rowan_functional_1990} in the NLopt (nonlinear-optimization) package \parencite{johnson_nlopt_2020}.  
Parts of the source code were adopted from the mimi-NICE model \parencite{dennig_mimi_nice_2017} and the RICEupdate model \parencite{dennig_riceupdate_2019}.

\todo[inline]{Maybe move the paragraph on optimization algorithms to the appendix.}

\subsection{Results} \label{ssec: Results}

This section investigates how optimal carbon prices depend on the choice of welfare weights in the absence of international transfers. As in the theory section, I distinguish between two utilitarian solutions contingent on whether carbon prices are constrained to be uniform. 
I begin by presenting the main finding: an increased optimal climate policy stringency under both utilitarian approaches compared to the Negishi solution. 
Leveraging the theoretical insights, the remainder of the section explores the drivers of this result.

\subsubsection{The effect on optimal carbon prices} \label{sssec: The effect on optimal climate policy (WW)}

It is useful to first examine the overall stringency of the optimal climate policy paths. To this end, Figure~\ref{F:temp and welfare weights} shows the respective optimal atmospheric temperature trajectories for different optimization problems and different utility discount rates (also referred to as the pure rate of time preference in the literature); 
specifically, I compare the results for the commonly used positive and normative calibrations of the utility discount rates by Nordhaus (1.5\%) and Stern (0.1\%), respectively \parencite{nordhaus_estimates_2011, stern_stern_2006}\footnote{
Like Negishi weights, the utility discount rate also places different weights on the welfare of different people. However, it does so on the basis of time – giving lower weight to the welfare of future generations – rather than on the basis of the wealth (or, more precisely, the consumption level) of an individual. 
The issue of discounting future utilities is heavily debated among economists and has received much more attention than the use of Negishi weights.
}. 

\begin{figure}[!htb]
\centering \includegraphics[width=0.85\textwidth]{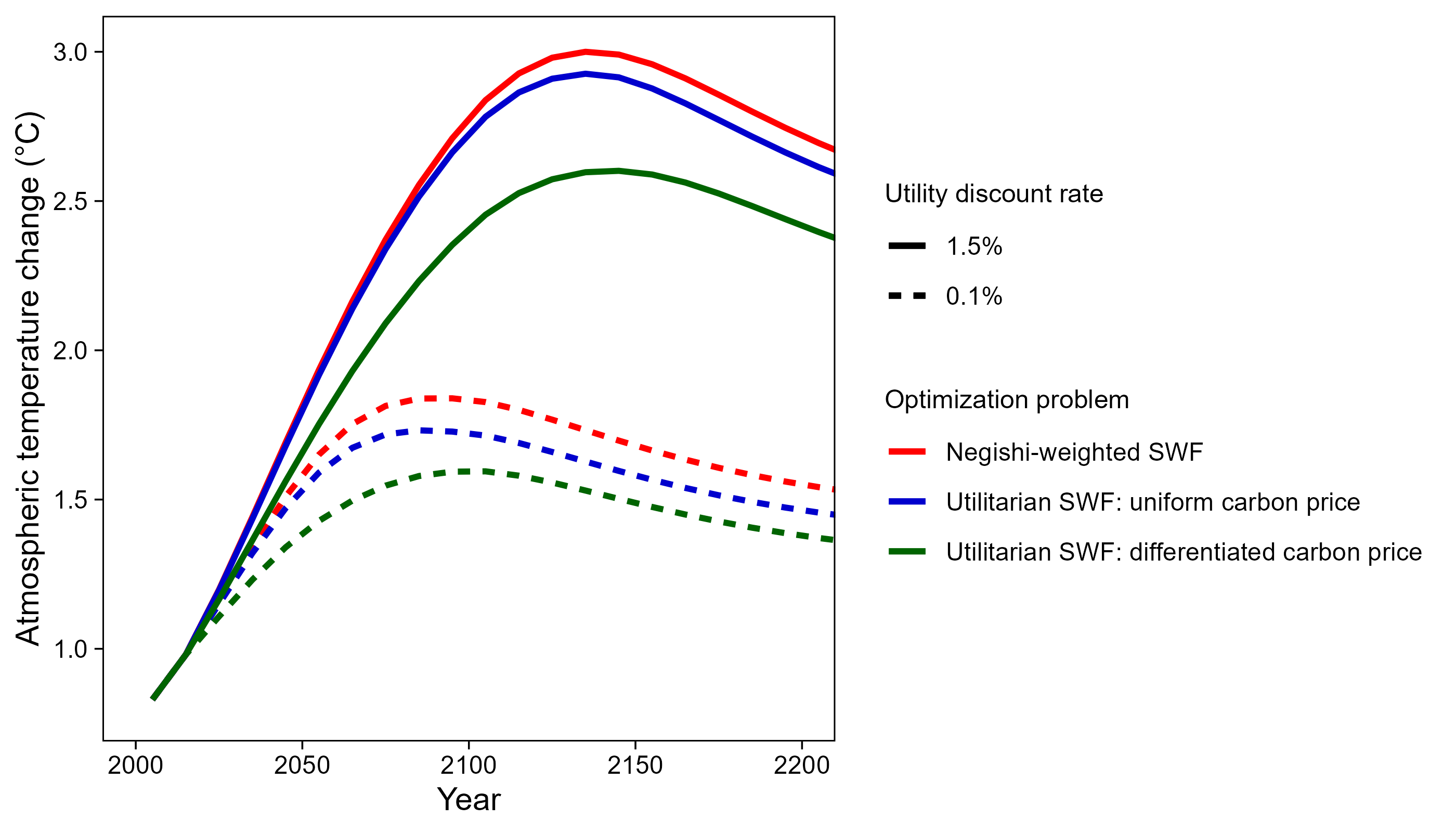}
\caption{Optimal atmospheric temperature trajectories conditional on the optimization problem and the utility discount rate.}
\vspace{1mm}  
    \begin{minipage}{1\textwidth}
        \small \textit{Notes}:
        The Negishi-weighted solutions (red) are compared to the solutions under the utilitarian objective with (green) and without (blue) the additional constraint of equalized regional carbon prices for the Nordhaus (solid lines) and Stern (dashed lines) utility discount rates \parencite{nordhaus_estimates_2011, stern_stern_2006}. Temperature changes are relative to 1900.
    \end{minipage}
\label{F:temp and welfare weights}
\end{figure}

The first main result is that accounting for global inequality increases the optimal climate policy stringency in the RICE model; the utilitarian solutions with uniform and differentiated carbon prices yield lower optimal temperature trajectories than the Negishi solution. Allowing for differentiated carbon prices in the utilitarian optimization results in the lowest warming by accounting for inequality in determining both the carbon price level and differentiation. Figure~\ref{F:temp and welfare weights} also shows the well-known large sensitivity of optimal climate policy to the utility discount rate. Specifically, peak warming is 3.00°C (1.84°C) in the Negishi solution, 2.93°C (1.73°C) in the utilitarian solution with uniform carbon prices, and 2.60°C (1.59°C) in the utilitarian solution with differentiated carbon prices for the 1.5\% (0.1\%) utility discount rate.

The corresponding cumulative global industrial\footnote{
There are two sources of emissions in RICE: endogenous region-level industrial emissions and exogenous emissions from land use change. Industrial emissions constitute the bulk of total emissions. Cumulative emissions from land use change are 29 GtCO$_2$ globally over the entire model horizon from 2005-2595.
} carbon dioxide emissions for the entire model horizon from 2005-2595 are shown in Table~\ref{T: Cumulative global industrial CO2 emissions} in the appendix. 
The effect of increased optimal abatement in the utilitarian solutions relative to the Negishi solution is larger for the lower utility discount rate, when the welfare impacts of future damages are given comparatively greater weight. Specifically, relative to the Negishi solution, cumulative global industrial CO$_2$ emissions are around 5\% (13\%) lower for the utilitarian solution with the additional constraint of uniform carbon prices, and 21\% (27\%) lower for the utilitarian differentiated carbon price solution, using the 1.5\% (0.1\%) utility discount rate.

The optimal carbon prices\footnote{
Note that all dollar values are 2022 USD. I convert the 2005 USD values of the RICE model to 2022 USD values using the World Bank GDP deflator \parencite{world_bank_world_2023}.
} in 2025 are shown in Table~\ref{T: Optimal carbon price in 2025} (the full carbon price trajectories are shown in Figure \ref{F: Optimal trajectories for the carbon price and industrial emissions} in the appendix, along with the corresponding emissions). Welfare-maximizing uniform carbon prices exceed the Negishi-weighted carbon prices for both utility discount rates; specifically, by 15\% (21\%) under the 1.5\% (0.1\%) utility discount rate.

\begin{table}[!htb]
\centering
\begin{threeparttable}
\caption{Optimal carbon price in 2025 (in 2022 \$/tCO$_2$) depending on the optimization problem and the utility discount rate ($\rho$).}
\begin{tabular}{l c c }
\hline \hline
\multirow{2}{*}{Optimization problem} & \multicolumn{2}{c}{Utility discount rate} \\
\cline{2-3}
& \textbf{$\rho = 1.5\%$} & \textbf{$\rho = 0.1\%$} \\
\hline
A) \textit{Negishi-weighted SWF} & 25 & 100* \\
B) \textit{Utilitarian SWF: uniform carbon price} & 29 & 121 \\
C) \textit{Utilitarian SWF: differentiated carbon price} & & \\
\hspace{0.5cm} US & 338 & $> 410$ \\
\hspace{0.5cm} Other High Income & 233 & $> 501$ \\
\hspace{0.5cm} Japan & 232 & $> 638$ \\
\hspace{0.5cm} EU & 199 & $> 638$ \\
\hspace{0.5cm} Russia & 78 & $> 273$ \\
\hspace{0.5cm} Latin America & 48& 202 \\
\hspace{0.5cm} Middle East & 44 & 182 \\
\hspace{0.5cm} China & 32 & 134 \\
\hspace{0.5cm} Eurasia & 24 & 103 \\
\hspace{0.5cm} Other Asia & 10 & 44 \\
\hspace{0.5cm} India & 10 & 41 \\
\hspace{0.5cm} Africa & 5 & 23 \\
\hline \hline
\end{tabular}\scriptsize
\label{T: Optimal carbon price in 2025}
    \begin{minipage}{\linewidth}
        \vspace{1mm}  
        \footnotesize
        \textit{Notes}: Mimi-RICE-plus only yields an approximately equalized carbon price for the Negishi solution. In this case (*), it varied between 98 and 102 \$/tCO$_2$ across regions. The “$>$” sign indicates that the regional backstop price has been reached. Thus, any price above the backstop price is optimal as complete abatement is required.
    \end{minipage}
\end{threeparttable}
\end{table}

Furthermore, to reduce the welfare cost of abatement, there are large differences in welfare-maximizing carbon prices between regions when the constraint of equalized carbon prices is not imposed. 
Consistent with the theoretical results in Section \ref{sssec:Optimal carbon prices under different welfare weights}, this yields high carbon prices in rich regions – exceeding $\sim$\$200/tCO$_2$, even for the high utility discount rate – and much lower carbon prices in poor regions. 
For the lower Stern utility discount rate, the richest five regions already reach their backstop price in 2025, resulting in zero carbon emissions. 
Notably, the utilitarian differentiated carbon prices exceed the Negishi-weighted carbon prices for all regions but the poorest three (or four) regions. This also results in large regional changes in cumulative emissions compared to the Negishi solution, featuring substantial emission reductions in rich regions, smaller reductions in middle-income regions, and emission increases in the poorest three regions (see Figure \ref{EIND_cum_ro_NWvsUtil} in the appendix). 
Importantly, the emission reductions in rich and middle-income regions outweigh the emission increases in the poorest regions, resulting in lower global emissions compared to the Negishi solution. The differentiated carbon price optimum is discussed in more detail in Appendix~\ref{Assec: Discussion of DCPO}.

\subsubsection{Regional heterogeneities and distributional effects}

Why do the utilitarian maximizations lead to greater climate policy ambition compared to the Negishi solution? This section addresses this question by analyzing the distributional impacts of the different climate policy pathways and the regional heterogeneities that drive these outcomes. 

I begin by examining which regions are better off and which are worse off under the utilitarian solutions compared to the Negishi solution.
To provide a simple summary statistic that shows which regions gain or lose overall over the entire model horizon, I focus on regions' aggregate intertemporal welfare changes, expressed as the net present value (NPV) of consumption changes over time. 
These regional NPV consumption changes are shown in Figure \ref{fig:C_pctchange_NPV} along with the global welfare gains of the utilitarian solutions (the temporal trajectories in consumption changes are shown in Figure \ref{fig:C_pctchange} in the appendix). 
More specifically, I compute the consumption changes in the initial period (2005) that would yield a welfare change (in utility terms) that is equivalent to the welfare difference between each of the utilitarian solutions and the Negishi solution. 
Global utilitarian welfare changes are expressed in the welfare-equivalent consumption change in 2005 if consumption were distributed equally.
Details of this calculation are provided in Appendix~\ref{Assec: Calculation of welfare-equivalent consumption changes}.  For the remainder of the numerical results, I focus on the 1.5\% utility discount rate to streamline the discussion. Additional results for the 0.1\% utility discount rate are shown in the appendix.

\begin{figure}[!htb]
    \centering
    \includegraphics[width=0.9\linewidth]{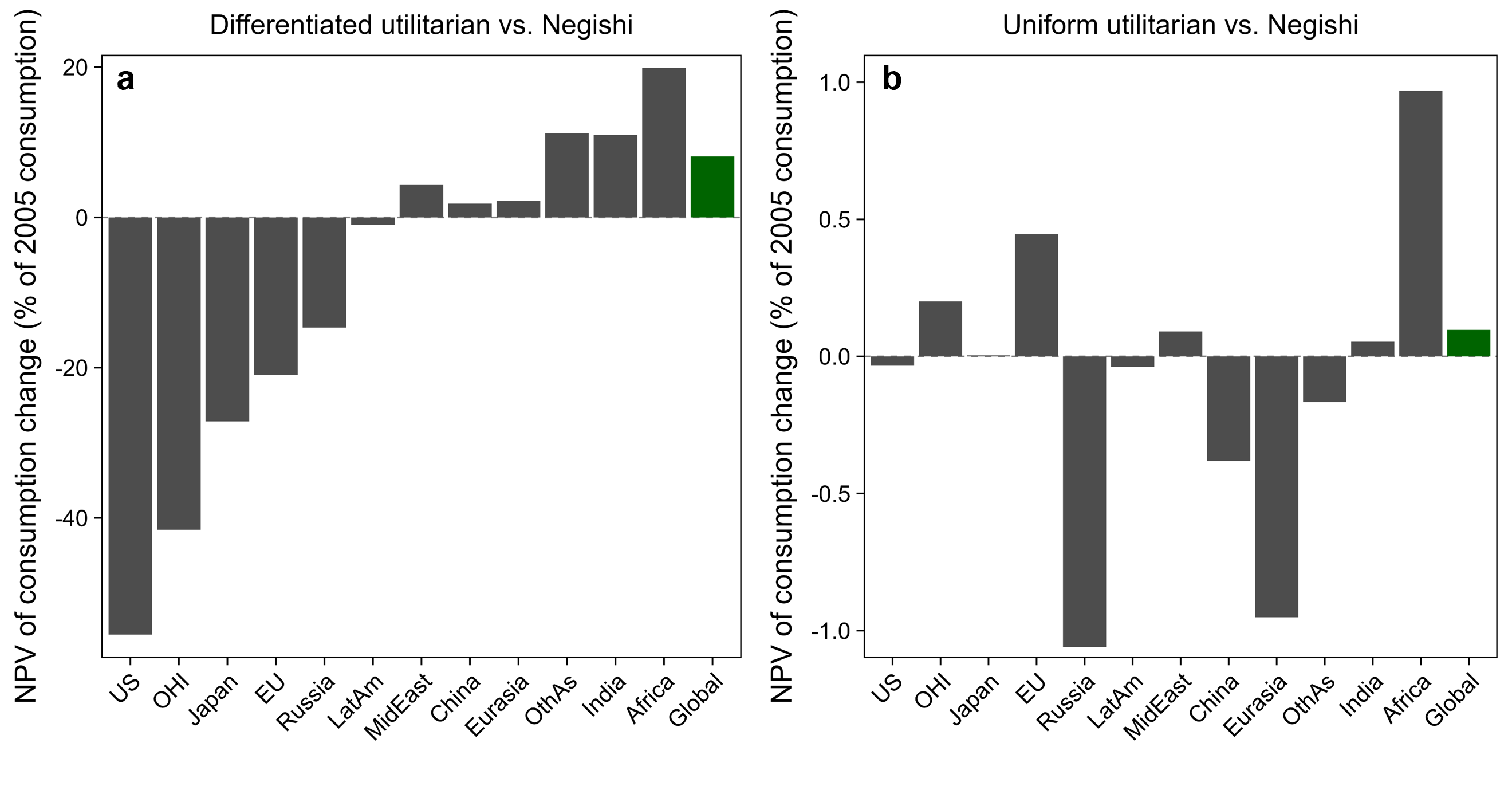}
    \caption{Net present value of consumption changes.}
    \vspace{1mm}  
    \begin{minipage}{1\textwidth}
        \small \textit{Notes}:
        The values show the welfare-equivalent consumption change in 2005, as a percentage of the consumption in 2005. The ``Global'' value expresses the global utilitarian welfare change in the welfare-equivalent consumption change in 2005 if consumption were distributed equally (for details, see Appendix \ref{Assec: Calculation of welfare-equivalent consumption changes}). The figure shows the results for the utility discount rate of 1.5\%.
    \end{minipage}
    \label{fig:C_pctchange_NPV}
\end{figure}
The comparison of the utilitarian differentiated carbon price solution and the Negishi solution is straightforward: rich regions are better off in the Negishi solution and poor regions are better off in the utilitarian solution (see Figure~\ref{fig:C_pctchange_NPV}a). Most importantly, global (unweighted) welfare is higher in the utilitarian solution.
This is of course unsurprising since the utilitarian SWF measures global (unweighted) welfare and Negishi weights upweight the welfare of rich regions and downweight the welfare of poor regions. 
The poorest four regions are better off in all periods in the utilitarian solution, due to both lower abatement costs in their regions and lower global emissions. In contrast, rich regions experience NPV consumption losses as a result of higher abatement costs associated with increased carbon prices. 
Crucially, however, all regions enjoy consumption gains after 2150 because of reduced climate damages due to reduced global emissions (see Figure \ref{fig:C_pctchange}a). Thus, the increased abatement in rich regions does not lead to persistently lower consumption trajectories. 
In addition, it is worth noting that the consumption losses in rich regions do not imply negative consumption per capita growth rates. 
More generally, the consumption per capita trajectories of all regions are not affected substantially, especially compared to the magnitude of the inequality across regions\footnote{
The consumption per capita trajectories for the regions with the largest positive and negative consumption changes, Africa and the US, respectively, are shown in Figure~\ref{fig: Consumption per capita trajectories} in the appendix.
}. Thus, while the utilitarian solutions result in greater global welfare by accounting for global inequality in setting the carbon prices, they do not solve the inequality issue.




The distributional consequences are more complicated for the uniform carbon price solutions. 
The region that benefits the most from higher carbon prices in the utilitarian solution relative to the Negishi solution is the poorest region, Africa (see Figure~\ref{fig:C_pctchange_NPV}b). Indeed, the intertemporal welfare gain in utility terms is by far the largest in Africa (see Figure \ref{fig:C_bar_change_NPV}b in the appendix).
This indicates that the lower carbon prices in the Negishi solution are primarily driven by the down-weighting of Africa in the Negishi-weighted SWF. I have confirmed this through a model run in which Africa is removed from the utilitarian social welfare function\footnote{
To be more specific, setting Africa's welfare weights to zero in the utilitarian social welfare function yields an optimal carbon price trajectory nearly identical to that under Negishi weights, resulting in a peak warming of 3.00°C, as in the Negishi solution.
}.
Further analysis reveals that the difference in optimal policy stringency between the utilitarian and Negishi solutions is driven by disproportionately large climate damages in Africa\footnote{
Specifically, I find that if Africa had the same temperature damage function as the US, the utilitarian uniform carbon price solution again results in peak warming of 3.00°C, the same as the Negishi solution.
}. 
While all regions benefit again in the long-term, Russia and Eurasia experience the greatest consumption losses in NPV terms due to increased abatement costs. 

To understand the underlying drivers of the distributional effects for the uniform carbon price solutions we can leverage the results from the theory model.
Proposition \ref{Proposition: utilitarian vs Negishi (dynamic)} and Equation \eqref{E: utilitarian vs Negishi (dynamic, RICE)} show that relatively higher marginal damages and a more convex abatement cost function in the South, compared to the North, contribute to a higher uniform carbon price in the utilitarian solution than in the Negishi solution.
Figure \ref{fig:ratios_2025} examines whether this is the case in the RICE model, showing the ratios of the regional marginal damages and convexities of the abatement cost functions (both as a percentage of GDP in 2025) relative to the US.
Marginal damages are estimated as the present value (in 2025) of the stream of damages associated with a marginal pulse of emissions in 2025 (using region-specific discount rates)\footnote{
This is effectively the regional social cost of carbon, which is calculated as the welfare-equivalent regional consumption change in 2025.
}.

\begin{figure}[!htb]
    \centering
    \includegraphics[width=0.95\linewidth]{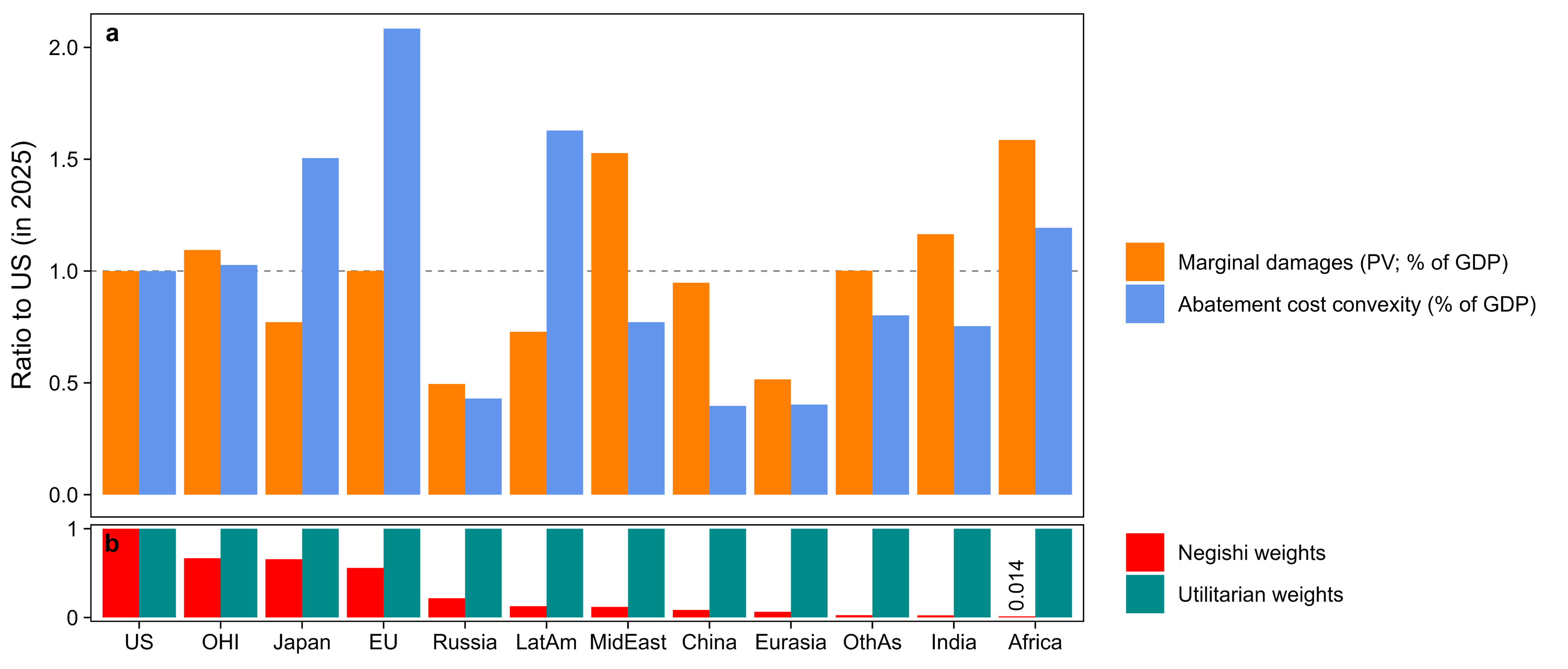} 
    \caption{Relative regional marginal damages and abatement cost convexities.}
    \vspace{1mm}  
    \begin{minipage}{1\textwidth}
        \small \textit{Notes}:
        The ratio of the 2025 present values (PV) of the stream of regional marginal damages as a percentage of the regional GDP in 2025, $\mathbf{d}'_{i}$, is given by $PV(\mathbf{d}'_{i})/PV(\mathbf{d}'_{US})$. The ratio of the convexities in the abatement cost functions is $c''_{i,t}/c''_{US,t}$ (evaluated at uniform carbon prices), where $t$ is the year 2025. 
        Panel (b) shows regions' Negishi and utilitarian welfare weights in 2025 relative to the weights in the US.
        The figure shows the results for the utility discount rate of 1.5\%.
    \end{minipage}
    \label{fig:ratios_2025}
\end{figure}

Consistent with the theoretical results, Figure \ref{fig:ratios_2025} shows that the regions that enjoy NPV consumption gains from the higher utilitarian carbon prices tend to have relatively high marginal damages and/or more convex abatement cost functions.
Most notably, Africa is the region with the largest marginal damages relative to the size of its economy. Additionally, Africa has a fairly convex abatement cost function, which reduces its abatement cost burden as carbon prices are increased.
Together, these two attributes explain why Africa experiences the largest gains in NPV consumption. 
The strongly convex abatement cost function in the EU is also noteworthy, resulting in NPV gains in the EU from the higher utilitarian carbon prices.
Conversely, Russia and Eurasia experience the largest NPV losses due to relatively low marginal damages and flatter marginal abatement cost curves.

As the poorest region, Africa receives the lowest Negishi weight, which is roughly 70 times smaller than that of the richest region, the US, in 2025 (see Figure \ref{fig:ratios_2025}b). Intuitively, the heavy down-weighting of welfare in the region most impacted by climate damages is a key factor behind the lower carbon prices in the Negishi solution. 

Figure \ref{fig:ratios_growth_2025} in the appendix provides a more detailed breakdown of the regional heterogeneities that give rise to differential climate impacts. Rather than computing the present values of the stream of marginal damages, which are affected by population growth and economic growth, it shows the undiscounted marginal damages as a percentage of GDP in a given year alongside population and economic growth. 
Building on the theoretical insights from the dynamic model, this figure shows that Africa is particularly strongly affected by climate change, not only due to its high climate damages as a percentage of GDP but also because it has the fastest population growth, amplifying the aggregate damages by increasing the number of people affected. However, counterbalancing this to some degree is Africa's fast economic growth, which causes Africa's climate damages to be more heavily discounted under the utilitarian SWF.

\todo[inline]{Continue writing after I have confirmed that this is a good way of doing it. Another way would be to show the utility differences, or welfare-equivalent consumption changes for a global representative agent.}


\subsubsection{Regions' preferred uniform carbon prices}

This section presents the preferred uniform carbon prices for different regions, offering a complementary perspective on why utilitarian welfare weights lead to higher uniform carbon prices than Negishi weights. 
Drawing on the theoretical two-region model from Section \ref{sssec: A region's preferred uniform carbon price}, we know that this occurs if and only if the poorer region prefers a higher uniform carbon price than the richer region. By examining these preferences within the RICE model, we can gain further intuition for this result. Additionally, understanding regions' preferences regarding uniform carbon prices is valuable in its own right, as it helps to identify which regions might advocate for more or less stringent global climate policies in international negotiations\footnote{
However, it is important to note that within the framework of international negotiations under the Paris Agreement, which emphasizes nationally determined contributions, a globally uniform carbon price has not been the central focus. 
}. 

Table~\ref{T: Preffered uniform carbon prices} shows each region's preferred uniform carbon prices and the resulting peak temperature increase (the full carbon price and temperature trajectories are shown in Figure \ref{fig: preferred uniform carbon prices} in the appendix). 
The preferred uniform carbon prices vary widely across regions. In 2025, they differ by nearly an order of magnitude, from \$7 per ton of CO$_2$ in Russia and Eurasia to over \$60 per ton in Africa and the EU. This also leads to significant differences in peak temperatures, with Africa's preferred policy limiting warming to 2.4°C, while Russia's preferred policy allows for nearly 4°C.

\renewcommand{\arraystretch}{1.2}
\begin{table}[!htb]
\centering
\begin{threeparttable}
\caption{Regions preferred uniform carbon prices (in 2022 \$/tCO$_2$) and resulting peak warming.
For comparison, the uniform utilitarian and Negishi-weighted carbon prices are also shown.}
\setlength{\tabcolsep}{10pt}
\begin{tabular}{lcccc}
\hline \hline
                        & \multicolumn{3}{c}{Carbon prices} &                   \\ \cline{2-4}
Welfare weights & 2025      & 2055      & 2085      & Peak warming (°C) \\ \hline
Negishi                 & 25        & 60*        & 116**       & 3.00              \\
Utilitarian             & 29        & 68        & 128       & 2.93              \\ \hdashline
US                      & 30        & 64        & 111       & 2.99              \\
OHI                     & 34        & 77        & 146       & 2.86              \\
Japan                   & 35        & 90        & 174       & 2.72              \\
EU                      & 63        & 133       & 224       & 2.42              \\
Russia                  & 7         & 17        & 35        & 3.98              \\
LatAm                   & 38        & 79        & 138       & 2.91              \\
MidEast                 & 37        & 72        & 118       & 3.00              \\
China                   & 13        & 40        & 94        & 3.00              \\
Eurasia                 & 7         & 20        & 43        & 3.65              \\
OthAs                   & 27        & 62        & 120       & 3.00              \\
India                   & 29        & 67        & 129       & 2.94              \\
Africa                  & 64        & 134       & 225       & 2.40              \\ \hline \hline
\end{tabular}\scriptsize
\label{T: Preffered uniform carbon prices}
    \begin{minipage}{\linewidth}
        \vspace{1mm}  
        \footnotesize
        \textit{Notes}: The table shows the results for the utility discount rate of 1.5\%. Temperature changes are relative to 1900.
        Mimi-RICE-plus only yields an approximately equalized carbon price for the Negishi solution. Specifically, it varied across regions between (*) 59 and 63 \$/tCO$_2$ in 2055, and (**) 113 and 121 \$/tCO$_2$ in 2085.
    \end{minipage}
\end{threeparttable}
\end{table}

The large differences in regions' preferred uniform carbon prices highlight the importance of how these preferences are weighted in the SWF. 
This perspective provides additional intuition for why optimal carbon prices are lower in the Negishi solution, which downweights the preferences of poorer regions. 
Notably, Africa, the poorest region that is most heavily downweighted in the Negishi-weighted SWF, has the highest preferred carbon prices, mainly due to its disproportionately large climate damages.
In contrast, the US, the region with the largest Negishi weights, prefers comparatively low carbon prices, particularly after 2050. 
Specifically, Africa's preferred uniform carbon prices are more than twice as high as those of the US. 
While the overall effect depends on all regions, the downweighting of Africa's preferences is the primary reason for lower carbon prices in the Negishi solution compared to the utilitarian solution, which assigns equal weight to all regions' welfare, as discussed above. 

\section{Conclusion} \label{sec:conclusion}

This paper investigates how accounting for inequality affects 
optimal carbon prices.
Specifically, it compares the optimal carbon prices under two optimization approaches: the conventional, positive approach which maximizes the Negishi-weighted social welfare function (SWF), and a normative approach which focuses on maximizing global welfare, employing constrained maximizations of the utilitarian SWF.

Using a theoretical model, I show that, in the absence of international transfers, accounting for inequality may result in higher or lower optimal carbon prices and that this depends on 
regional differences in marginal climate damages and the burden of abatement costs, and disparities in population and economic growth.
Intuitively, global welfare maximization warrants more stringent climate policy if poor countries are more vulnerable to future climate change---due to higher marginal damages, faster population growth and slow economic catch-up---and if the cost burden of abatement predominately falls on rich countries.
This highlights the importance of accurately accounting for regional heterogeneities in climate economic model. 
In numerical simulations with the integrated assessment model RICE, I find that accounting for inequality results in lower optimal global emissions, both if carbon prices are allowed to be regionally differentiated and if they are constrained to be globally uniform. 



There are some limitations of this study which are left for future research. First, the RICE model masks inequality within its twelve regions. Thus, a valuable avenue for future research would be to account for inequality at a finer resolution and examine how the quantitative results change. Existing modifications of the RICE model may be used for this analysis, such as NICE and RICE50+ \parencite{dennig_inequality_2015, gazzotti_persistent_2021}.

Second, the numerical simulations of this study are performed with a single integrated assessment model (IAM). As different IAMs are known to produce different results, it would be worthwhile to replicate the analysis with other IAMs to assess the robustness of the findings of this paper. 
Furthermore, it would be valuable to strengthen the empirical evidence on regional heterogeneities highlighted by this paper’s theoretical results, including differences in damage and abatement cost functions, as well as in economic and population growth. 


Third, this study relies on deterministic models. Given the substantial uncertainties in both human and physical systems and the associated economic effects of climate change, extending the analysis to a probabilistic framework would be valuable. Building on the findings of this paper, it could be particularly insightful to explore potential disparities in the uncertainties that different regions face.

\clearpage
\singlespacing

\printbibliography

\todo[inline]{Include all references}







\clearpage


\onehalfspacing

\appendix

\renewcommand{\thefigure}{A\arabic{figure}}
\renewcommand{\thetable}{A\arabic{table}}
\renewcommand{\theequation}{A\arabic{equation}}

\setcounter{figure}{0}
\setcounter{table}{0}
\setcounter{equation}{0}

\setcounter{proposition}{0}
\setcounter{lemma}{0}
\setcounter{corollary}{0}
\setcounter{theorem}{0}
\renewcommand{\theproposition}{A\arabic{proposition}}
\renewcommand{\thelemma}{A\arabic{lemma}}
\renewcommand{\thecorollary}{A\arabic{corollary}}
\renewcommand{\thetheorem}{A\arabic{theorem}}

\renewcommand{\thesection}{\Alph{section}}

\stepcounter{section}

\section*{Appendix \thesection: Proofs and Derivations}
\addcontentsline{toc}{section}{Appendix \thesection: Proofs and Derivations} 

\subsection{Derivation of optimal carbon prices}\label{Assec: Derivation of optimal carbon prices}

This section shows the derivation of the optimal uniform carbon price for arbitrary welfare weights (Equation~\ref{Def: optimal uniform carbon price}). As discussed briefly below, the derivation of the optimal differentiated carbon prices is largely analogous (and, in fact, simpler) and therefore not explicitly shown.

The Lagrangian of the uniform carbon price optimization problem is
\begin{gather*} \label{objective}
    \begin{aligned}
    \mathcal{L} = L_i \alpha_i u\left(\frac{X_i}{L_i}\right)
        &- \sum_i \lambda_i \left(X_i - W_i + C_i(A_{i}) + D_i(A) \right) \\
        & - \mu \left(C'_N(A_N) - C'_S(A_S) \right), 
    \end{aligned}
\end{gather*}
where $\lambda_i$ and $\mu$ are Lagrange multipliers.

The first-order conditions (FOCs) are:
\begin{gather*}
    \begin{aligned}
        &[X_i]: \alpha_i u'(x_i) = \lambda_i, \quad \forall i\\
        &[A_{N}]: \lambda_N C'_N(A_{N}) = -\sum_i \lambda_i D_i'(A) - \mu C''_N\\
        &[A_{S}]: \lambda_S C'_S(A_{S}) = -\sum_i \lambda_i D_i'(A) + \mu C''_S \\
        &[\lambda_i]: X_i = W_i - C_i(A_i) - D_i(A), \quad \forall i \\
        &[\mu]: C'_N(A_N) = C'_S(A_S) .
    \end{aligned}
\end{gather*}

Combining the FOCs, we get the following two optimality conditions:
\begin{gather}
     \alpha_N u'(x_N) C'_N(A_{N}) = - \sum_i \alpha_i u'(x_i) D_i'(A) - \mu C''_N \label{optimality condition 1_uniform}, \\
     \alpha_S u'(x_S) C'_S(A_{S}) = - \sum_i \alpha_i u'(x_i) D_i'(A) + \mu C''_S \label{optimality condition 2_uniform} .
\end{gather}

We can now solve these two equations for the optimal uniform carbon price, noting that $C'_N(A_N) = C'_S(A_S)$ by the uniform carbon price constraint. Eliminating $\mu$ by dividing Equation~\eqref{optimality condition 1_uniform} by Equation~\eqref{optimality condition 2_uniform}, and simple manipulations yield the optimal uniform carbon price (Equation~\eqref{Def: optimal uniform carbon price} in the main text),
\begin{gather*}
        \tau^{uni} = C'_i(A_i^*) = - \sum_i \alpha_i u'(x_i^*) D'_i(A^*) 
        \frac{C''_S + C''_N}
        {\alpha_N u'(x_N^*) C''_S + \alpha_S u'(x_S^*) C''_N} .
\end{gather*}

The derivation of the optimal differentiated carbon price is largely analogous. The main difference is that the uniform carbon price constraint is missing (i.e., $\mu (C'_N(A_N) - C'_S(A_S))$ is missing in the Lagrangian). Hence, we (generally) get different optimal carbon prices in the two regions.

\todo[inline]{Continue Appendix A.1 derivation}

\subsection{Derivation of optimality conditions}\label{Assec: Derivation of optimality conditions}

\subsubsection{Negishi solution}\label{Asssec: Negishi solution}

We start with Equation~\eqref{Def: Negishi-weighted carbon price}, 
\begin{gather*}
        \frac{d{C}_{i}(\tilde{A}_i)}{d\tilde{A}_i} 
        = 
        - \sum_i \frac{d{D}_i(\tilde{A})}{d\tilde{A}} .
    \end{gather*}
    
Multiplying both sides by $\frac{d\tilde{A}}{d\tilde{\tau}}$, and using  $d\tilde{A} = d\tilde{A}_S + d\tilde{A}_N$, yields
%
\begin{gather*}
        \frac{d{C}_{i}(\tilde{A}_i)}{d\tilde{A}_i} \frac{d\tilde{A}_S + d\tilde{A}_N}{d\tilde{\tau}}
        = 
        - \sum_i \frac{d{D}_i(\tilde{A})}{d\tilde{\tau}} .
    \end{gather*}

Using $\tilde{\tau} = \frac{d{C}_N(\tilde{A}_N)}{d\tilde{A}_N} = \frac{d{C}_S(\tilde{A}_S)}{d\tilde{A}_S}$, we obtain
%
%
\begin{gather*}
        \sum_i
        \frac{d{C}_i(\tilde{A}_i)}
        {d\tilde{\tau}}
        = 
        - \sum_i \frac{d{D}_i(\tilde{A})}{d\tilde{\tau}} .
    \end{gather*}

\subsubsection{Utilitarian uniform carbon price solution}\label{Asssec: Utilitarian uniform carbon price solution}

We start with Equation~\eqref{Def: utilitarian uniform carbon price},
    \begin{gather*}
        \frac{d{C}_{i}(\check{A}_i)}{d\check{A}_i} = 
        - \sum_i {u}'(\check{x}_i) \frac{d{D}_i(\check{A})}{d \check{A}}
        \frac{{C}''_S + {C}''_N}
        {{u}'(\check{x}_N) {C}''_S + {u}'(\check{x}_S) {C}''_N } .
    \end{gather*}

Using $C''_i = \frac{d \check{C}'_i}{d \check{A}_i} = \frac{d \check{\tau}}{d \check{A}_i}$, multiplying both sides by $\frac{d\check{A}}{d\check{\tau}}$ and rearranging, we have 
\begin{gather*}
       \frac{d {C}_{i}(\check{A}_i)}{d \check{A}_i} 
        \frac{{u}'(\check{x}_N) \frac{d\check{A}}{d \check{A}_S} 
        + {u}'(\check{x}_S) \frac{d\check{A}}{d \check{A}_N}}
        {\frac{d \check{\tau}}{d \check{A}_S} + \frac{d \check{\tau}}{d \check{A}_N}}
       = 
        - \sum_i {u}'(\check{x}_i) \frac{d{D}_i(\check{A})}{d\check{\tau}}.
    \end{gather*}

Using $\frac{d {C}_{S}(\check{A}_S)}{d \check{A}_S} = \frac{d {C}_{N}(\check{A}_N)}{d \check{A}_N}$, equalizing the denominators of the ratios in the denominator and rearranging yields
\begin{equation*}
    \begin{gathered}
        {u}'(\check{x}_N) \frac{d\check{A}}{d \check{A}_S d \check{\tau}}
           \frac{d C_S(\check{A}_S) d C_N(\check{A}_N)}{d C_S(\check{A}_S) + d C_N(\check{A}_N)}
            + {u}'(\check{x}_S) \frac{d\check{A}}{d \check{A}_N d \check{\tau}}
            \frac{d C_S(\check{A}_S) d C_N(\check{A}_N)}{d C_S(\check{A}_S) + d C_N(\check{A}_N)} \\
           = 
            - \sum_i {u}'(\check{x}_i) \frac{d{D}_i(\check{A})}{d\check{\tau}}.
    \end{gathered} 
    \end{equation*}

    Using $\check{\tau} = \frac{d\check{C}_S}{d\check{A}_S} = \frac{d\check{C}_N}{d\check{A}_N}$, and thus $\check{\tau} d\check{A}_i  = d\check{C}_i$ for all $i$, we rewrite the previous equation as
\begin{gather*}
\check{\tau} \frac{d\check{A}_S + d\check{A}_N}{\check{\tau} d\check{A}_S + \check{\tau} d\check{A}_N}
\left(
       {u}'(\check{x}_N) \frac{d C_N(\check{A}_N)}{d \check{\tau}}
        + {u}'(\check{x}_S) \frac{d C_S(\check{A}_S)}{d \check{\tau}}
        \right) 
       = 
        - \sum_i {u}'(\check{x}_i) \frac{d{D}_i(\check{A})}{d\check{\tau}},
    \end{gather*}

which simplifies to
\begin{gather*}
   \sum_i {u}'(\check{x}_i) \frac{d C_i(\check{A}_i)}{d \check{\tau}}
       = 
        - \sum_i {u}'(\check{x}_i) \frac{d{D}_i(\check{A})}{d\check{\tau}}.
    \end{gather*}

\subsubsection{Regions' preferred uniform carbon price}\label{Asssec: Regions' preferred uniform carbon price}

We start with Equation~\eqref{Def: preferred uniform carbon price},
\begin{gather*}
        \frac{d{C}_{i}(\mathring{A}_{i}^i)}{d\mathring{A}_{i}^i}
        =
        - \frac{d{D}_i(\mathring{A^i})}{d\mathring{A^i}}
        \frac{{C}''_i + {C}''_{-i}}
        {{C}''_{-i}},
    \end{gather*}

Using ${C}''_i = \frac{d {C}'_i(\mathring{A}_i^i)}{d \mathring{A}_i^i} = \frac{d \mathring{\tau}^i}{d \mathring{A}_i^i}$ and ${C}''_i = \frac{d {C}'_{-i}(\mathring{A}_{-i}^i)}{d \mathring{A}_{-i}^i} = \frac{d \mathring{\tau}^i}{d \mathring{A}_{-i}^i}$, multiplying both sides by $\frac{d\check{A}}{d\check{\tau}}$ and rearranging, we have 
\begin{gather*}
\begin{aligned}
    \frac{d{C}_{i}(\mathring{A}_{i}^i)}{d\mathring{A}_{i}^i}
        & =
        - \frac{d{D}_i(\mathring{A^i})}{d\mathring{A^i}}
        \frac{\frac{d \mathring{\tau}^i}{d \mathring{A}_i^i} +
        \frac{d \mathring{\tau}^i}{d \mathring{A}_{-i}^i}}
        {\frac{d \mathring{\tau}^i}{d \mathring{A}_{-i}^i}} \\
        & =
        - \frac{d{D}_i(\mathring{A^i})}{d\mathring{A^i}}
        \left(
        \frac{d\mathring{A}_{-i}^i}{d \mathring{A}_i^i} + 1
        \right) \\
        & =
        - \frac{d{D}_i(\mathring{A^i})}{d\mathring{A^i}}
        \left(
        \frac{d\mathring{A}_{-i}^i + d \mathring{A}_i^i}{d \mathring{A}_i^i}
        \right) \\
        & = - \frac{d{D}_i(\mathring{A^i})}{d\mathring{A^i_i}}.
\end{aligned}
\end{gather*}

Multiplying both sides by $\frac{d \mathring{A}^i_i}{d \mathring{\tau}^i}$, we obtain
\begin{gather*}
    \frac{d{C}_{i}(\mathring{A}_{i}^i)}{d\mathring{\tau}^i}
        = - \frac{d{D}_i(\mathring{A^i})}{d\mathring{\tau}^i}.
\end{gather*}


\todo[inline]{Rewrite proofs as proof by contrapositive.}

\todo[inline]{Shorten proof and check consistency of notation.}

\subsection{Proof of Proposition 1} \label{Assec: Proof of Proposition 1}

\todo[inline]{Shorten proof and check consistency of notation.}

\begin{proof}
    We split the proof into the forward and backward implications.


    \vspace{6pt}
    \myuline{Proof of forward implication}: $\tilde{\tau} < \check{\tau}$  $\implies$ $\frac{\check{D}'_S}{\check{D}'_N} > \frac{C''_N}{C''_S}$.
    
    Let us ask under which conditions $\tilde{\tau} < \check{\tau}$, or equivalently, $\tilde{C}' < \check{C}'$. First note that, for strictly convex abatement cost functions, $\tilde{C}' < \check{C}'$ implies $\tilde{A}_i < \check{A}_i$ for all $i$, and thus $\tilde{A} < \check{A}$. For strictly convex damage functions, this implies $\tilde{D}'_i < \check{D}'_i $ for all $i$ (note that marginal damages of abatement are negative).

    We have $\tilde{C}' < \check{C}'$ if and only if
    \begin{gather*}
        - \tilde{D}'_N - \tilde{D}'_S < (-\check{u}'_N \check{D}'_N - \check{u}'_S \check{D}'_S) \frac{C''_S + C''_N}{\check{u}'_N C''_S + \check{u}'_S C''_N} .
    \end{gather*}
    
    Multiplying both sides by the denominator on the right-hand side (which is positive), and rearranging, we get
    \begin{gather*}
         (\check{u}'_N \check{D}'_N + \check{u}'_S \check{D}'_S) (C''_S + C''_N) < (\tilde{D}'_N + \tilde{D}'_S) (\check{u}'_N C''_S + \check{u}'_S C''_N) .
    \end{gather*}
    
    Multiplying out and collecting terms, we have
    \begin{gather} \label{E: UvN inequality}
        C''_N \Bigl(\check{u}'_N \check{D}'_N + \check{u}'_S \check{D}'_S - \check{u}'_S \tilde{D}'_N  - \check{u}'_S \tilde{D}'_S  \Bigr)
        <
        C''_S \Bigl(\check{u}'_N \tilde{D}'_N  + \check{u}'_N \tilde{D}'_S  - \check{u}'_N \check{D}'_N - \check{u}'_S \check{D}'_S \Bigr) .
    \end{gather}
    
    We know that $\Bigl(\check{u}'_N \check{D}'_N + \check{u}'_S \check{D}'_S - \check{u}'_S \tilde{D}'_N  - \check{u}'_S \tilde{D}'_S  \Bigr) > 0$ if $\check{u}'_S > \check{u}'_N$ and $\tilde{D}'_i < \check{D}'_i$ because $\check{u}'_S \check{D}'_S - \check{u}'_S \tilde{D}'_S > 0$ since $\tilde{D}'_i < \check{D}'_i$ and $\check{u}'_N \check{D}'_N - \check{u}'_S \tilde{D}'_N >0$ since $\check{u}'_S > \check{u}'_N$ and $\tilde{D}'_i < \check{D}'_i$. Hence, we can divide by it and the sign of the inequality does not flip. Moreover, note that we must also have 
    \begin{gather}\label{E: UvN positivity condition}
        \Bigl(\check{u}'_N \tilde{D}'_N  + \check{u}'_N \tilde{D}'_S  - \check{u}'_N \check{D}'_N - \check{u}'_S \check{D}'_S \Bigr) >0
    \end{gather}
     for the inequality in Equation~\eqref{E: UvN inequality} to hold, since $C_i''>0$ for all $i$.
    
    Cross-division, collecting common terms, and rearranging yields
    \begin{gather*}
        \frac{C''_N}{C''_S} 
        <
         \frac{\check{u}'_N (\tilde{D}'_N  -  \check{D}'_N + \tilde{D}'_S  ) - \check{u}'_S \check{D}'_S }
         {\check{u}'_N \check{D}'_N - \check{u}'_S (\tilde{D}'_S  -\check{D}'_S + \tilde{D}'_N )}.
    \end{gather*}
    
    Note that
    \begin{gather*}
        \underbrace{\tilde{D}'_N  -  \check{D}'_N}_{<0} + \tilde{D}'_S < \tilde{D}'_S
    \end{gather*}
    
    and
    \begin{gather*}
        \underbrace{\tilde{D}'_S  -\check{D}'_S}_{<0} + \tilde{D}'_N < \tilde{D}'_N
    \end{gather*}
    
    Thus, for the numerator we have
    \begin{gather*}
    \begin{aligned}
        \overbrace{\check{u}'_N \underbrace{(\tilde{D}'_N  -  \check{D}'_N + \tilde{D}'_S )}_{< \tilde{D}'_S} - \check{u}'_S \check{D}'_S}^{>0 \text{ by Equation~\eqref{E: UvN positivity condition} }} &< \check{u}'_N \tilde{D}'_S - \check{u}'_S \check{D}'_S \\
        &< \check{u}'_N \check{D}'_S - \check{u}'_S \check{D}'_S \\
        &> 0 ,
    \end{aligned}  
    \end{gather*}
    where the second inequality holds since $\tilde{D}'_S < \check{D}'_S$ and the last inequality holds since $\check{u}'_S > \check{u}'_N$.
    
    Similarly, for the denominator, we have
    \begin{gather*}
    \begin{aligned}
        \check{u}'_N \check{D}'_N - \check{u}'_S \underbrace{(\tilde{D}'_S  -\check{D}'_S + \tilde{D}'_N)}_{<\tilde{D}'_N}
        &> \check{u}'_N \check{D}'_N - \check{u}'_S \tilde{D}'_N \\
        &> \check{u}'_N \check{D}'_N - \check{u}'_S \check{D}'_N \\
        &> 0 ,
    \end{aligned}
    \end{gather*}
    where the second inequality holds since $\tilde{D}'_S < \check{D}'_S$ and the last inequality holds since $\check{u}'_S > \check{u}'_N$.
    
    Compared to the case ``Negishi = Utilitarian", we have a greater (positive) denominator, and a smaller (positive, by Equation~\eqref{E: UvN positivity condition}) numerator. 
    
    Putting this together we have
    \begin{gather*}
    \begin{aligned}
        \frac{C''_N}{C''_S} 
        &< \frac{\check{u}'_N (\tilde{D}'_N  -  \check{D}'_N + \tilde{D}'_S  ) - \check{u}'_S \check{D}'_S }
        {\check{u}'_N \check{D}'_N - \check{u}'_S (\tilde{D}'_S  -\check{D}'_S + \tilde{D}'_N )} \\
        &< \frac{\check{u}'_N \check{D}'_S - \check{u}'_S \check{D}'_S}{\check{u}'_N \check{D}'_N - \check{u}'_S \check{D}'_N} \\
         &= \frac{\check{D}'_S}{\check{D}'_N} .
    \end{aligned} 
    \end{gather*}
    
    We have thus shown that $\tilde{C}' < \check{C}' \implies \frac{\check{D}'_S}{\check{D}'_N} > \frac{C''_N}{C''_S}$.



    \vspace{6pt}
    \myuline{Proof of backward implication}: $\frac{\check{D}'_S}{\check{D}'_N} > \frac{C''_N}{C''_S}$  $\implies$ $\tilde{\tau} < \check{\tau}$.

    In order to derive a contradiction, suppose that $\frac{\check{D}'_S}{\check{D}'_N} > \frac{C''_N}{C''_S}$  $\implies$ $\tilde{\tau} \geq \check{\tau}$.

    We begin by establishing the implications of $\tilde{\tau} \geq \check{\tau}$, or equivalently, $\tilde{C}' \geq \check{C}'$. First note that, for strictly convex abatement cost functions, $\tilde{C}' \geq \check{C}'$ implies $\tilde{A}_i \geq \check{A}_i$ for all $i$, and thus $\tilde{A} \geq \check{A}$. For strictly convex damage functions, this implies $\tilde{D}'_i \geq \check{D}'_i$ for all $i$ (note that marginal damages of abatement are negative).

    Next, note that $\tilde{C}' \geq \check{C}'$ if and only if
    \begin{gather*}
        - \tilde{D}'_N - \tilde{D}'_S \geq (-\check{u}'_N \check{D}'_N - \check{u}'_S \check{D}'_S) \frac{C''_S + C''_N}{\check{u}'_N C''_S + \check{u}'_S C''_N} .
    \end{gather*}
    
    Multiplying both sides by the denominator on the right-hand side (which is positive), and rearranging, we get
    \begin{gather*}
         (\check{u}'_N \check{D}'_N + \check{u}'_S \check{D}'_S) (C''_S + C''_N)
         \geq (\tilde{D}'_N + \tilde{D}'_S) (\check{u}'_N C''_S + \check{u}'_S C''_N) .
    \end{gather*}
    
    Multiplying this out and collecting common terms gives
    \begin{gather} \label{E: UvN inequality direction2}
        C''_N \Bigl(\check{u}'_N \check{D}'_N + \check{u}'_S \check{D}'_S - \check{u}'_S \tilde{D}'_N  - \check{u}'_S \tilde{D}'_S  \Bigr)
        \geq
        C''_S \Bigl(\check{u}'_N \tilde{D}'_N  + \check{u}'_N \tilde{D}'_S  - \check{u}'_N \check{D}'_N - \check{u}'_S \check{D}'_S \Bigr) .
    \end{gather}

    We know that
    \begin{gather}\label{E: UvN positivity condition part1 direction2}
        \Bigl(\check{u}'_N \tilde{D}'_N  + \check{u}'_N \tilde{D}'_S  - \check{u}'_N \check{D}'_N - \check{u}'_S \check{D}'_S \Bigr) > 0 
    \end{gather}

    because $\check{u}'_N \tilde{D}'_S - \check{u}'_S \check{D}'_S > 0$ since $\tilde{D}'_i \geq \check{D}'_i$ and $\check{u}'_S > \check{u}'_N$ and $\check{u}'_N \tilde{D}'_N - \check{u}'_N \check{D}'_N \geq 0$ since $\tilde{D}'_i \geq \check{D}'_i$.

    Moreover, note that Equations \eqref{E: UvN inequality direction2} and \eqref{E: UvN positivity condition part1 direction2} imply
    \begin{gather}\label{E: UvN positivity condition direction2}
        \Bigl(\check{u}'_N \check{D}'_N + \check{u}'_S \check{D}'_S - \check{u}'_S \tilde{D}'_N  - \check{u}'_S \tilde{D}'_S  \Bigr) >0
    \end{gather}

    since $C_i''>0$ for all $i$. Hence, we can divide by it and the sign of the inequality does not flip.
    
    Cross-division and collecting common terms yields
    \begin{gather*}
        \frac{C''_N}{C''_S} 
        \geq
        \frac{\check{u}'_N (\tilde{D}'_N  -  \check{D}'_N + \tilde{D}'_S  ) - \check{u}'_S \check{D}'_S }
        {\check{u}'_N \check{D}'_N - \check{u}'_S (\tilde{D}'_S  -\check{D}'_S + \tilde{D}'_N )} .
    \end{gather*}

    It is worthwhile to take stock at this point. So far, we have established that
    \begin{gather}\label{E: N>U implication}
        \tilde{C}' \geq \check{C}'
        \iff
        \frac{C''_N}{C''_S} 
        \geq
        \frac{\check{u}'_N (\tilde{D}'_N  -  \check{D}'_N + \tilde{D}'_S  ) - \check{u}'_S \check{D}'_S }
        {\check{u}'_N \check{D}'_N - \check{u}'_S (\tilde{D}'_S  -\check{D}'_S + \tilde{D}'_N )} .
    \end{gather}

    Next, we show that $\frac{\check{D}'_S}{\check{D}'_N} > \frac{C''_N}{C''_S}$  $\implies$ $\tilde{C}' \geq \check{C}'$ yields a contradiction:
    \begin{gather*}
    \begin{aligned}
        \frac{C''_N}{C''_S} &< \frac{\check{D}'_S}{\check{D}'_N} \\
        &= \frac{\check{u}'_N \check{D}'_S - \check{u}'_S \check{D}'_S}{\check{u}'_N \check{D}'_N - \check{u}'_S \check{D}'_N} \\
        &\leq \frac{\check{u}'_N \tilde{D}'_S - \check{u}'_S \check{D}'_S}{\check{u}'_N \check{D}'_N - \check{u}'_S \tilde{D}'_N} \\
        &\leq \frac{\check{u}'_N (\tilde{D}'_N  -  \check{D}'_N + \tilde{D}'_S  ) - \check{u}'_S \check{D}'_S }
        {\check{u}'_N \check{D}'_N - \check{u}'_S (\tilde{D}'_S  -\check{D}'_S + \tilde{D}'_N )} \\
        &\leq \frac{C''_N}{C''_S} ,
    \end{aligned}
    \end{gather*}

    where the second and third inequalities follow from $\tilde{D}'_i \geq \check{D}'_i$ for all $i$ and the fact that the denominator (and, trivially, the numerator) is positive by Equation~\eqref{E: UvN positivity condition direction2}\footnote{
    Note that the denominator in the third line is positive because it is greater than the positive denominator in the fourth line. This can be seen as follows:
    \begin{gather*}
        \check{u}'_N \check{D}'_N - \check{u}'_S \tilde{D}'_N
        > \check{u}'_N \check{D}'_N - \check{u}'_S ( \underbrace{\tilde{D}'_S  -\check{D}'_S}_{\geq 0} + \tilde{D}'_N )
        > 0.
    \end{gather*}
    }.
    The last inequality follows from the implication of $\tilde{C}' \geq \check{C}'$ documented in Equation~\eqref{E: N>U implication}.

    We have reached the contradiction $\frac{C''_N}{C''_S} < \frac{C''_N}{C''_S}$. Hence, $\frac{\check{D}'_S}{\check{D}'_N} > \frac{C''_N}{C''_S}$  $\implies$ $\tilde{C}' \geq \check{C}'$ is incorrect,
    and we have thus shown that we must have $\frac{\check{D}'_S}{\check{D}'_N} > \frac{C''_N}{C''_S} \implies \tilde{C}' < \check{C}'$.

    Together, the proofs of the forward and backward implications yield the equivalence
    \begin{gather*}
        \check{\tau} > \tilde{\tau}
        \iff
        \frac{\check{D}'_S}{\check{D}'_N} > \frac{C''_N}{C''_S}.
    \end{gather*}
\end{proof}

\subsection{Generalization of Proposition 1} \label{Assec: Generalization of Proposition 1}

This section generalizes Proposition \ref{Proposition: utilitarian vs Negishi} to a setting with $n$ regions. For an arbitrary number of regions $i \in \mathcal{I}=\{1,\dots,n\}$, the \textit{uniform carbon price optimization problem} is 
\begin{gather*}
    \max_{X_i, \tau} \sum_i L_i \alpha_i u \left(\frac{X_i}{L_i}\right) \\
    \begin{aligned}
    \text{subject to: }
    \label{E: constraints diff1period}
    &X_i = W_i - C_i(A_{i}(\tau)) - D_i(A(\tau)), \quad \forall i. \\
    \end{aligned}
\end{gather*}
Solving this problem with utilitarian weights yields the \textit{utilitaran uniform carbon price}, which is implicitely defined by
\begin{gather}
    \check{\tau} = C'(\check{A}_i) 
    = - \left( \sum_i \frac{1}{C''_i} \right) \frac{\sum_i  u'(\check{x}_i) D_i'(\check{A}(\check{\tau})) }
      {\sum_i u'(\check{x}_i) \frac{1}{C''_i}}.
\end{gather}
The Negishi-weighted carbon price is implicitely defined by
\begin{gather}
\label{eq: Negishi carbon price (n regions)}
    \tilde{\tau} = C'(\tilde{A}_i) 
    = - \sum_i D_i'(\tilde{A}(\tilde{\tau})).
\end{gather}
Proposition \ref{Proposition: utilitarian vs Negishi} then generalizes as follows.
\begin{proposition}[{\(n\)-region analogue of Proposition~\ref{Proposition: utilitarian vs Negishi}}]
\label{Proposition: utilitarian vs Negishi (n regions)}
    Let $\mathcal{I}=\{1,\dots,n\}$.
    The utilitarian uniform carbon price is greater than the Negishi-weighted carbon price, that is $\check{\tau} > \tilde{\tau}$, if and only if 
    $\frac{\sum_i \check{u}_i ' \check{D} _i'}{ \sum_i \check{D}'_i}
    >
    \frac{\sum_i \check{u}'_i \frac{1}{C''_i} }{ \sum_i \frac{1}{C''_i} } $.
\end{proposition}

\begin{proof}
We split the proof into the forward and backward implications.

    \vspace{6pt}
    \myuline{Proof of forward implication}: $\check{\tau} > \tilde{\tau}$  $\implies$ 
    $\frac{\sum_i \check{u}_i ' \check{D} _i'}{ \sum_i \check{D}'_i}
 >
 \frac{\sum_i \check{u}'_i \frac{1}{C''_i} }{ \sum_i \frac{1}{C''_i} }$ .
    Let us ask under which conditions $\check{\tau} > \tilde{\tau}$, or equivalently, $\check{C}' > \tilde{C}'$. First note that, for strictly convex abatement cost functions, $\check{C}' > \tilde{C}'$ implies $\check{A}_i > \tilde{A}_i$ for all $i$, and thus $\check{A} > \tilde{A}$. For strictly convex damage functions, this implies $\check{D}'_i > \tilde{D}'_i$ for all $i$ (note that marginal damages of abatement are negative).
    We have $\check{C}' > \tilde{C}'$ if and only if
    \begin{gather*}
        - \sum_i \frac{1}{C''_i} \frac{\sum_i \check{u}_i ' \check{D} _i' }
      {\sum_i \check{u}'_i \frac{1}{C''_i}}
        >
        - \sum_i \tilde{D}'_i .
    \end{gather*}
Multiply by the denominator on the LHS (which is positive)
  \begin{gather*}
        - \sum_i \frac{1}{C''_i} 
        \sum_i \check{u}_i ' \check{D} _i'
        >
        - \sum_i \check{u}'_i \frac{1}{C''_i} \sum_i \tilde{D}'_i .
    \end{gather*}
Divide by $\sum_i \frac{1}{C''_i}$ and  $(- \sum_i \tilde{D}'_i)$ (which are both positive)
\begin{gather*}
       \frac{\sum_i \check{u}_i ' \check{D} _i'}{ \sum_i \tilde{D}'_i }
       >
\frac{\sum_i \check{u}'_i \frac{1}{C''_i} }
{ \sum_i \frac{1}{C''_i} }.
    \end{gather*}
    Since $0>\check{D}'_i > \tilde{D}'_i$ for all $i$ (note that marginal damages of abatement are negative) and $0 > \sum_i \check{u}_i ' \check{D} _i'$, we have
\begin{gather*}
       \frac{\sum_i \check{u}_i ' \check{D} _i'}{ \sum_i \check{D}'_i }
       > 
       \frac{\sum_i \check{u}_i ' \check{D} _i'}{ \sum_i \tilde{D}'_i }
       >
\frac{\sum_i \check{u}'_i \frac{1}{C''_i} }
{ \sum_i \frac{1}{C''_i} }.
    \end{gather*}
We have thus shown that  $\check{\tau} > \tilde{\tau}$  $\implies$ 
    $\frac{\sum_i \check{u}_i ' \check{D} _i'}{ \sum_i \check{D}'_i}
    > \frac{\sum_i \check{u}'_i \frac{1}{C''_i} }{ \sum_i \frac{1}{C''_i} }$.

 \myuline{Proof of backward implication}: 
 $\frac{\sum_i \check{u}_i ' \check{D} _i'}{ \sum_i \check{D}'_i}
 >
 \frac{\sum_i \check{u}'_i \frac{1}{C''_i} }{ \sum_i \frac{1}{C''_i} }$  $\implies$ $\check{\tau} > \tilde{\tau}$.
We will prove the contrapositive of the stated result. That is, we will prove 
\begin{gather*}
    \check{\tau} \leq \tilde{\tau}
    \implies
    \frac{\sum_i \check{u}_i ' \check{D} _i'}{ \sum_i \check{D}'_i}
    \leq
    \frac{\sum_i \check{u}'_i \frac{1}{C''_i} }{ \sum_i \frac{1}{C''_i} }.
\end{gather*}
    We start by establishing the implications of $\check{\tau} \leq \tilde{\tau}$, or equivalently, $\check{C}' \leq \tilde{C}'$.

    First note that, for strictly convex abatement cost functions, $\check{C}' \leq \tilde{C}'$ implies $\check{A}_i \leq \tilde{A}_i$ for all $i$, and thus $\check{A} \leq \tilde{A}$. For strictly convex damage functions, this implies $\check{D}'_i \leq \tilde{D}'_i$ for all $i$ (note that marginal damages of abatement are negative).
    Next, note that $\check{C}' \leq \tilde{C}'$ if and only if
     \begin{gather*}
        - \sum_i \frac{1}{C''_i} \frac{\sum_i \check{u}_i ' \check{D} _i' }
      {\sum_i \check{u}'_i \frac{1}{C''_i}}
        \leq
        - \sum_i \tilde{D}'_i .
    \end{gather*}
Multiply by the denominator on the LHS (which is positive)
  \begin{gather*}
        - \sum_i \frac{1}{C''_i} 
        \sum_i \check{u}_i ' \check{D} _i'
        \leq
        - \sum_i \check{u}'_i \frac{1}{C''_i} \sum_i \tilde{D}'_i .
    \end{gather*}
Divide by $\sum_i \frac{1}{C''_i}$ and  $(- \sum_i \tilde{D}'_i)$ (which are both positive)
\begin{gather*}
       \frac{\sum_i \check{u}_i ' \check{D} _i'}{ \sum_i \tilde{D}'_i }
       \leq
\frac{\sum_i \check{u}'_i \frac{1}{C''_i} }
{ \sum_i \frac{1}{C''_i} }.
    \end{gather*}
  Since $\check{D}'_i \leq \tilde{D}'_i$ for all $i$ (note that marginal damages of abatement are negative) and $\sum_i \check{u}_i ' \check{D} _i' < 0$, we have
\begin{gather*}
       \frac{\sum_i \check{u}_i ' \check{D} _i'}{ \sum_i \check{D}'_i }
       \leq
       \frac{\sum_i \check{u}_i ' \check{D} _i'}{ \sum_i \tilde{D}'_i }
       \leq
\frac{\sum_i \check{u}'_i \frac{1}{C''_i} }
{ \sum_i \frac{1}{C''_i} }.
    \end{gather*}
We have thus shown that 
    \begin{gather*}
        \check{\tau} \leq \tilde{\tau}
        \implies
        \frac{\sum_i \check{u}_i ' \check{D} _i'}{ \sum_i \check{D}'_i}
        \leq
        \frac{\sum_i \check{u}'_i \frac{1}{C''_i} }{ \sum_i \frac{1}{C''_i} }.
    \end{gather*}
By contraposition, we therefore have
\begin{gather*}
        \check{\tau} > \tilde{\tau}
        \implies
       \frac{\sum_i \check{u}_i ' \check{D} _i'}{ \sum_i \check{D}'_i}
       >
       \frac{\sum_i \check{u}'_i \frac{1}{C''_i} }{ \sum_i \frac{1}{C''_i} }.
\end{gather*}

Together, the proofs of the forward and backward implications yield the equivalence
   \begin{gather*}
   \check{\tau} > \tilde{\tau}
   \iff
    \frac{\sum_i \check{u}_i ' \check{D} _i'}{ \sum_i \check{D}'_i}
    >
    \frac{\sum_i \check{u}'_i \frac{1}{C''_i} }{ \sum_i \frac{1}{C''_i} }  
    .
\end{gather*}
\end{proof}

The intuition of the two-region case carries over to the setting with an arbirary number of regions: 
the utilitarian uniform carbon price exceeds the Negishi-weighted carbon price
when poorer regions tend to face relatively higher marginal damages and steeper marginal abatement cost functions.


\subsection{Proof of Corollary 1} \label{Assec: Proof of Corollary 1}

Proposition \ref{Proposition: utilitarian vs Negishi} establishes that $\check{\tau} > \tilde{\tau}$, if and only if $\frac{\check{D}'_S}{\check{D}'_N} > \frac{C''_N}{C''_S}$. We can rewrite this condition as
\begin{gather*}
    \frac{\frac{d\check{D}_S}{d\check{\tau}}}{\frac{d\check{D}_S}{d\check{\tau}}}
    =
    \frac{\frac{d\check{D}_S}{d\check{A}}}{\frac{d\check{D}_S}{d\check{A}}}
    >
    \frac{\frac{d\check{C}_N'}{d\check{A}_N}}{\frac{d\check{C}_S'}{d\check{A}_S}}
    =
    \frac{\frac{d\check{A}_S}{d\check{\tau}}}{\frac{d\check{A}_N}{d\check{\tau}}}
    =
    \frac{\frac{d\check{A}_S}{d\check{A}}}{\frac{d\check{A}_N}{d\check{A}}}
    =
     \frac{\frac{d\check{C}_S}{d\check{A}}}{\frac{d\check{C}_N}{d\check{A}}}
    =
     \frac{\frac{d\check{C}_S}{d\check{\tau}}}{\frac{d\check{C}_N}{d\check{\tau}}},
\end{gather*}
where $\frac{dA_i}{d\tau_i} = \frac{1}{C_i''}$, and the third equality on the right-hand side follows from $\frac{d\check{C}_S}{d\check{A}_S} = \frac{d\check{C}_N}{d\check{A}_N}$ .

This establishes that 
\begin{gather*}
    \frac{-\frac{d\check{D}_S}{d\check{\tau}}}{\frac{d\check{C}_S}{d\check{\tau}}} > \frac{-\frac{d\check{D}_N}{d\check{\tau}}}{\frac{d\check{C}_N}{d\check{\tau}}}.
\end{gather*}

It remains to be shown that the left-hand side is greater than one, while the right-hand side is less than one. I utilize Proposition \ref{Proposition: U>N iff S>N} and Lemma \ref{Lemma: U&N in between preferred carbon price} to show this.

From Lemma \ref{Lemma: U&N in between preferred carbon price} we know that the utilitarian uniform carbon price ($\check{\tau}$) and the Negishi-weighted carbon price ($\tilde{\tau}$) are in between the preferred uniform carbon prices of the Global North ($\mathring{\tau}^N$) and the Global South ($\mathring{\tau}^S$), unless they all coincide. Moreover, from Proposition \ref{Proposition: U>N iff S>N} we know that the utilitarian uniform carbon price ($\check{\tau}$)  is greater than the Negishi-weighted carbon price ($\tilde{\tau}$) if and only if the preferred uniform carbon price of the Global South ($\mathring{\tau}^S$) is greater than the preferred uniform carbon price of the Global North ($\mathring{\tau}^N$). Therefore, $\check{\tau} > \tilde{\tau}$ implies $\mathring{\tau}^S > \check{\tau} > \tilde{\tau} > \mathring{\tau}^N$.

From Equation~\eqref{E: Preferred uniform carbon price wrt carbon price} we know that the marginal benefit-cost ratio with respect to the uniform carbon price equals one at the preferred uniform carbon price. That is, 
\begin{gather*}
\frac{- \frac{d{D}_i(\mathring{A}(\mathring{\tau}^i))}{d\mathring{\tau}^i}}
{\frac{d C_i(\mathring{A}_i(\mathring{\tau}^i))}{d \mathring{\tau}^i}}
=
1.
    \end{gather*}
We can relate this to the marginal benefit-cost ratios at the utilitarian uniform carbon price. 

For the North, we have
$- \frac{d{D}_N(\mathring{A}(\mathring{\tau}^N))}{d\mathring{\tau}^N}
>
- \frac{d{D}_N(\check{A}(\check{\tau}))}{d\check{\tau}}$. To see this, note that we have $\frac{d\mathring{\tau}^N}{d\mathring{A}(\mathring{\tau}^N)} = \frac{d\check{\tau}}{d\check{A}(\check{\tau})}$ since $\frac{d^3C(A_i)}{dA_i^3}=0$ for all $A_i$\footnote{
This can be seen as follows:
\begin{gather*}
    \frac{d\tau}{dA(\tau)}
    = \frac{d\tau}{dA_N(\tau) + dA_S(\tau)}
    = \frac{d\tau}{\frac{1}{C_N''} d\tau + \frac{1}{C_S''} d\tau}
    = \frac{1}{\frac{1}{C_N''} + \frac{1}{C_S''}},
\end{gather*}
where the second equality holds since $C''_i) = \frac{d\tau}{dA_i(\tau)}$. Notice that the last term is constant since $\frac{d^3C(A_i)}{dA_i^3}=0$.
}. Therefore,  $- \frac{d{D}_N(\mathring{A}(\mathring{\tau}^N))}{d\mathring{\tau}^N}
>
- \frac{d{D}_N(\check{A}(\check{\tau}))}{d\check{\tau}}$ if and only if $- \frac{d{D}_N(\mathring{A}(\mathring{\tau}^N))}{d\mathring{\tau}^N}
\frac{d\mathring{\tau}^N}{d\mathring{A}(\mathring{\tau}^N)}
>
- \frac{d{D}_N(\check{A}(\check{\tau}))}{d\check{\tau}}
\frac{d\check{\tau}}{d\check{A}(\check{\tau})}$,  which simplifies to 
$- \frac{d{D}_N(\mathring{A}(\mathring{\tau}^N))}{d\mathring{A}(\mathring{\tau}^N)}
>
- \frac{d{D}_N(\check{A}(\check{\tau}))}{d\check{A}(\check{\tau})}$. We know that this inequality holds from the strict convexity of the damage function and since $\check{\tau} > \mathring{\tau}^N$ implies $\check{A}_i(\check{\tau}) > \mathring{A}_i(\mathring{\tau}^N)$ for all $i$ for strictly convex abatement cost functions. 

Moreover, we have $\frac{d{C}_N(\mathring{A}(\mathring{\tau}^N))}{d\mathring{\tau}^N}
<
\frac{d{C}_N(\check{A}(\check{\tau}))}{d\check{\tau}}$ for the North. To see this, note that we can write this as 
$\frac{d{C}_N(\mathring{A}(\mathring{\tau}^N))}{d\mathring{A}(\mathring{\tau}^N)}
\frac{d\mathring{A}(\mathring{\tau}^N)}{d\mathring{\tau}^N}
<
\frac{d{C}_N(\check{A}(\check{\tau}))}{d\check{A}(\check{\tau})}
\frac{d\check{A}(\check{\tau})}{d\check{\tau}}$, which in turn can be rewritten as 
$\mathring{\tau}^N
\frac{1}{C''_N(\mathring{A}(\mathring{\tau}^N))}
<
\check{\tau}
\frac{1}{C''_N(\check{A}(\check{\tau}))}$. This inequality holds since  $\check{\tau} > \mathring{\tau}^N$ and $\frac{d^3C(A_i)}{dA_i^3}=0$ for all $A_i$.

Together, this establishes the following inequalities for the North:
\begin{gather*}
1 =
\frac{- \frac{d{D}_N(\mathring{A}(\mathring{\tau}^N))}{d\mathring{\tau}^N}}
{\frac{d C_N(\mathring{A}_N(\mathring{\tau}^N))}{d \mathring{\tau}^N}}
>
\frac{- \frac{d{D}_N(\check{A}(\check{\tau}))}{d\check{\tau}}}
{\frac{d C_N(\mathring{A}_N(\mathring{\tau}^N))}{d \mathring{\tau}^N}}
>
\frac{- \frac{d{D}_N(\check{A}(\check{\tau}))}{d\check{\tau}}}
{\frac{d C_N(\check{A}_N(\check{\tau}))}{d \check{\tau}}}.
    \end{gather*}

Conversely, for the South, we have
$- \frac{d{D}_S(\mathring{A}(\mathring{\tau}^S))}{d\mathring{\tau}^S}
<
- \frac{d{D}_S(\check{A}(\check{\tau}))}{d\check{\tau}}$. To see this, note that we have $\frac{d\mathring{\tau}^S}{d\mathring{A}(\mathring{\tau}^S)} = \frac{d\check{\tau}}{d\check{A}(\check{\tau})}$ since $\frac{d^3C(A_i)}{dA_i^3}=0$ for all $A_i$. Therefore,  $- \frac{d{D}_S(\mathring{A}(\mathring{\tau}^S))}{d\mathring{\tau}^S}
<
- \frac{d{D}_S(\check{A}(\check{\tau}))}{d\check{\tau}}$ if and only if $- \frac{d{D}_N(\mathring{A}(\mathring{\tau}^S))}{d\mathring{\tau}^S}
\frac{d\mathring{\tau}^S}{d\mathring{A}(\mathring{\tau}^S)}
<
- \frac{d{D}_S(\check{A}(\check{\tau}))}{d\check{\tau}}
\frac{d\check{\tau}}{d\check{A}(\check{\tau})}$,  which simplifies to 
$- \frac{d{D}_S(\mathring{A}(\mathring{\tau}^S))}{d\mathring{A}(\mathring{\tau}^S)}
<
- \frac{d{D}_S(\check{A}(\check{\tau}))}{d\check{A}(\check{\tau})}$. We know that this inequality holds from the strict convexity of the damage function and since $\check{\tau} < \mathring{\tau}^S$ implies $\check{A}_i(\check{\tau}) < \mathring{A}_i(\mathring{\tau}^S)$ for all $i$ for strictly convex abatement cost functions. 

Moreover, we have $\frac{d{C}_S(\mathring{A}(\mathring{\tau}^S))}{d\mathring{\tau}^S}
>
\frac{d{C}_S(\check{A}(\check{\tau}))}{d\check{\tau}}$ for the South. To see this, note that we can write this as 
$\frac{d{C}_S(\mathring{A}(\mathring{\tau}^S))}{d\mathring{A}(\mathring{\tau}^S)}
\frac{d\mathring{A}(\mathring{\tau}^S)}{d\mathring{\tau}^S}
>
\frac{d{C}_S(\check{A}(\check{\tau}))}{d\check{A}(\check{\tau})}
\frac{d\check{A}(\check{\tau})}{d\check{\tau}}$, which in turn can be rewritten as 
$\mathring{\tau}^S
\frac{1}{C''_S(\mathring{A}(\mathring{\tau}^S))}
>
\check{\tau}
\frac{1}{C''_S(\check{A}(\check{\tau}))}$. This inequality holds since  $\check{\tau} < \mathring{\tau}^S$ and $\frac{d^3C(A_i)}{dA_i^3}=0$ for all $A_i$.

Together, this establishes the following inequalities for the South:
\begin{gather*}
1 =
\frac{- \frac{d{D}_S(\mathring{A}(\mathring{\tau}^S))}{d\mathring{\tau}^S}}
{\frac{d C_S(\mathring{A}_S(\mathring{\tau}^S))}{d \mathring{\tau}^S}}
>
\frac{- \frac{d{D}_S(\check{A}(\check{\tau}))}{d\check{\tau}}}
{\frac{d C_S(\mathring{A}_S(\mathring{\tau}^S))}{d \mathring{\tau}^S}}
>
\frac{- \frac{d{D}_S(\check{A}(\check{\tau}))}{d\check{\tau}}}
{\frac{d C_S(\check{A}_S(\check{\tau}))}{d \check{\tau}}}.
    \end{gather*}

We have thus shown that
\begin{gather*}
    \frac{-\frac{d\check{D}_S}{d\check{\tau}}}{\frac{d\check{C}_S}{d\check{\tau}}}
    >
    1
    >
    \frac{-\frac{d\check{D}_N}{d\check{\tau}}}{\frac{d\check{C}_N}{d\check{\tau}}}.
\end{gather*}

\subsection{Proof of Lemma 1} \label{Assec: Proof of Lemma 1}

\begin{proof}
    We start by showing that $\hat{\tau}_S < \tilde{\tau}$. Suppose, towards a contradiction, that $\hat{\tau}_S \geq \tilde{\tau}$, which is the case if and only if
\begin{equation*}
    - \tilde{D}'_N - \tilde{D}'_S \leq - \hat{D}'_S - \frac{\hat{u}_N'} {\hat{u}_S'} \hat{D}'_N .
\end{equation*}

Since $\frac{\hat{u}_N'} {\hat{u}_S'} < 1$ this inequality is satisfied if and only if $\hat{A} < \tilde{A}$, and thus $\hat{D}'_i < \tilde{D}'_i$ for all $i$.  From the definition of the utilitarian differentiated carbon price, we know that $\hat{\tau}_S < \hat{\tau}_N$. However,  $\tilde{\tau} \leq \hat{\tau}_S < \hat{\tau}_N$ implies  $\hat{A} > \tilde{A}$, and we have thus arrived at a contradiction. Therefore, we must have $\hat{\tau}_S < \tilde{\tau}$.  $\hat{A}_S < \tilde{A}_S$ follows from the strict convexity of the abatement cost function (and the definitions of the optimal carbon prices).

Next, we show that $\tilde{\tau} < \hat{\tau}_N$. Suppose, towards a contradiction, that $\tilde{\tau} \geq \hat{\tau}_N$, which is the case if and only if
\begin{equation*}
    - \tilde{D}'_N - \tilde{D}'_S
    \geq - \frac{\hat{u}_S'} {\hat{u}_N'} \hat{D}'_S - \hat{D}'_N .
\end{equation*}

Since $\frac{\hat{u}_S'} {\hat{u}_N'} > 1$ this inequality is satisfied if and only if $\hat{A} > \tilde{A}$, and thus $\hat{D}'_i > \tilde{D}'_i$ for all $i$.  From the definition of the utilitarian differentiated carbon price, we know that $\hat{\tau}_S < \hat{\tau}_N$. However,  $\tilde{\tau} \geq \hat{\tau}_N > \hat{\tau}_S$ implies $\hat{A} < \tilde{A}$, and we have thus arrived at a contradiction. Therefore, we must have $\tilde{\tau} < \hat{\tau}_N$.  $\hat{A}_N > \tilde{A}_N$ follows from the strict convexity of the abatement cost function (and the definitions of the optimal carbon prices).
\end{proof}

\subsection{Proof of Proposition 2} \label{Assec: Proof of Proposition 2}

\begin{proof}

    We first need to obtain expressions for the abatement as a function of the marginal abatement cost. 
Since I assume that the abatement cost functions are smooth, strictly increasing, strictly convex, and the third derivative is equal to zero for all $A_i>0$, they have the following form\footnote{
To give an idea about what affects this constant, the characteristics that determine $k_i$ in the RICE model are the size of the economy, the baseline emissions intensity, the price of a backstop technology, and the parameter that determines the convexity of the abatement cost function (see Equation~\ref{E: AC function RICE}).
}: $C_i(A_i) = k_i A_i^2 + m_i A_i + n_i$, with $k_i >0$. 
The marginal abatement cost is thus $C_i'(A_i) = 2 k_i A_i + m_i$ and the second derivative is $C_i'' = 2 k_i$. Therefore,  $k_i = \frac{C''_i}{2}$.
We invert the marginal abatement cost function to obtain an expression for the abatement:
\begin{gather}
    \label{eq: inverted MAC}
    A_i = \frac{1}{2} (C'_i - m_i) k_i ^{-1}.
\end{gather}


    We split the proof into the forward and backward implications.

    
    \vspace{6pt}
    \myuline{Proof of forward direction}: $\hat{A} > \tilde{A}$  $\implies$ $\frac{\hat{u}_S'}{\hat{u}_N'} \frac{\hat{D}'_S}{\hat{D}'_N} > \frac{C''_N}{C''_S}$.
    
    We begin by establishing the implications of $\hat{A} > \tilde{A}$. 
First, $\hat{A} > \tilde{A}$ implies $\hat{A}_N + \hat{A}_S > \tilde{A}_N + \tilde{A}_S$. 
Therefore, $\hat{A} > \tilde{A}$ if and only if

\begin{gather*}
\begin{aligned}
     \tilde{C}'_N k_N^{-1}+
     \tilde{C}'_S k_S^{-1}
    &<  \hat{C}'_N k_N^{-1}+
     \hat{C}'_S k_S^{-1}.
\end{aligned}
\end{gather*}

Plugging in the expressions for the marginal abatement costs detailed in Definitions \ref{Def: Negishi-weighted carbon price} and \ref{Def: utilitarian differentiated carbon price}, and rewriting, we get
\begin{gather*}
\begin{aligned}
    \left( \hat{D}'_N - \tilde{D}'_N + \frac{\hat{u}_S'}{\hat{u}_N'} \hat{D}'_S- \tilde{D}'_S  \right) k_N^{-1}
    < \left(\tilde{D}'_S - \hat{D}'_S + \tilde{D}'_N     
     - \frac{\hat{u}_N'}{\hat{u}_S'} \hat{D}'_N   \right) k_S^{-1}.
\end{aligned}
\end{gather*}

Next, note that $\tilde{A} < \hat{A}$ implies  $\tilde{D'_i} < \hat{D'_i}$, for all $i$.
Therefore, the previous inequality implies\footnote{
To see this, note that $\tilde{D'_i} < \hat{D'_i}$ implies the following inequalities:
\begin{gather*}
\begin{aligned}
\left(\frac{\hat{u}_S'}{\hat{u}_N'} \hat{D}'_S - \hat{D}'_S \right) k_N^{-1} 
< \left(\frac{\hat{u}_S'}{\hat{u}_N'} \hat{D}'_S - \tilde{D}'_S \right) k_N^{-1}
   < \left( \hat{D}'_N - \tilde{D}'_N + \frac{\hat{u}_S'}{\hat{u}_N'} \hat{D}'_S - \tilde{D}'_S  \right) k_N^{-1} \\
    < \left(\tilde{D}'_S - \hat{D}'_S + \tilde{D}'_N     
     - \frac{\hat{u}_N'}{\hat{u}_S'} \hat{D}'_N   \right) k_S^{-1}
     < \left(\tilde{D}'_N     
     - \frac{\hat{u}_N'}{\hat{u}_S'} \hat{D}'_N   \right) k_S^{-1} 
     < \left(\hat{D}'_N     
     - \frac{\hat{u}_N'}{\hat{u}_S'} \hat{D}'_N   \right) k_S^{-1}.
\end{aligned}
\end{gather*}
}
\begin{gather*}
\begin{aligned}
 \left(\frac{\hat{u}_S'}{\hat{u}_N'} \hat{D}'_S - \hat{D}'_S \right) k_N^{-1}
     < \left(\hat{D}'_N     
     - \frac{\hat{u}_N'}{\hat{u}_S'} \hat{D}'_N   \right) k_S^{-1},
\end{aligned}
\end{gather*}
which we can rewrite as (recall that $\hat{D}'_i < 0$ so the inequality flips)
\begin{gather*}\label{E:Udiff > N, quadratic AC}
\begin{aligned}
    \frac{\hat{u}_S'}{\hat{u}_N'}
     \frac{\hat{D}'_S}{\hat{D}'_N}
    > 
    \frac{k_N}{k_S}.  
\end{aligned}
\end{gather*}







Using $k_i = \frac{C''_i}{2}$, we get
\begin{gather*}
\begin{aligned}
    \frac{\hat{u}'_S}{\hat{u}'_N}
    \frac{\hat{D}'_S}{\hat{D}'_N}
    >
    \frac{C''_N}{C''_S}.
\end{aligned}
\end{gather*}

\vspace{6pt}
 \myuline{Proof of backward direction}:   $\frac{\hat{D}'_S}{\hat{D}'_N} > \frac{\hat{u}_N'}{\hat{u}_S'} \frac{C''_N}{C''_S}$ $\implies$  $\hat{A} > \tilde{A}$.

 In order to derive a contradiction, suppose that $\frac{\hat{D}'_S}{\hat{D}'_N} > \frac{\hat{u}_N'}{\hat{u}_S'} \frac{C''_N}{C''_S}$ $\implies$  $\hat{A} \leq \tilde{A}$.
We start by establishing the implications of $\hat{A} \leq \tilde{A}$. 
$\hat{A} \leq \tilde{A}$ implies $\hat{A}_N + \hat{A}_S \leq \tilde{A}_N + \tilde{A}_S$. 
Therefore, $\hat{A} \leq \tilde{A}$ if and only if
\begin{gather*}
\begin{aligned}
     \tilde{C}'_N k_N^{-1}+
     \tilde{C}'_S k_S^{-1}
    &\geq  \hat{C}'_N k_N^{-1}+
     \hat{C}'_S k_S^{-1}.
\end{aligned}
\end{gather*}

Plugging in the expressions for the marginal abatement costs detailed in Definitions \ref{Def: Negishi-weighted carbon price} and \ref{Def: utilitarian differentiated carbon price}, and rewriting, we get 
\begin{gather*}
\begin{aligned}
   \left( \hat{D}'_N - \tilde{D}'_N + \frac{\hat{u}_S'}{\hat{u}_N'} \hat{D}'_S- \tilde{D}'_S  \right) k_N^{-1}
    \geq \left(\tilde{D}'_S - \hat{D}'_S + \tilde{D}'_N     
     - \frac{\hat{u}_N'}{\hat{u}_S'} \hat{D}'_N   \right) k_S^{-1}.
\end{aligned}
\end{gather*}
Using $k_i = \frac{C_i''}{2}$, we get
\begin{gather}\label{E: Udiff inequality1}
\begin{aligned}
   \left( \hat{D}'_N - \tilde{D}'_N + \frac{\hat{u}_S'}{\hat{u}_N'} \hat{D}'_S- \tilde{D}'_S  \right) \frac{1}{C_N''}
    \geq \left(\tilde{D}'_S - \hat{D}'_S + \tilde{D}'_N     
     - \frac{\hat{u}_N'}{\hat{u}_S'} \hat{D}'_N   \right) \frac{1}{C_S''}.
\end{aligned}
\end{gather}


Next, note that $\tilde{A} \geq \hat{A}$ implies  $\tilde{D'_i} \geq \hat{D'_i}$, for all $i$.
$\tilde{D'_i} \geq \hat{D'_i}$ for all $i$ and $\hat{u}_S' > \hat{u}_N'$ imply
\begin{equation}\label{E: Udiff inequality2}
     \hat{D}'_N - \tilde{D}'_N + \frac{\hat{u}_S'}{\hat{u}_N'} \hat{D}'_S- \tilde{D}'_S  < 0.
\end{equation}

Moreover, Equations \eqref{E: Udiff inequality1} and \eqref{E: Udiff inequality2} imply
\begin{equation*}\label{E: Udiff inequality3}
    \tilde{D}'_S - \hat{D}'_S + \tilde{D}'_N     
     - \frac{\hat{u}_N'}{\hat{u}_S'} \hat{D}'_N < 0,
\end{equation*}
since $C_i''>0$ for all $i$.

We therefore know that the inequality flips upon cross-division, yielding
\begin{gather*}
\begin{aligned}
    \frac{C_N''}{C_S''}
    \geq
   \frac{ \hat{D}'_N - \tilde{D}'_N + \frac{\hat{u}_S'}{\hat{u}_N'} \hat{D}'_S- \tilde{D}'_S  }
   {\tilde{D}'_S - \hat{D}'_S + \tilde{D}'_N     
     - \frac{\hat{u}_N'}{\hat{u}_S'} \hat{D}'_N}.
\end{aligned}
\end{gather*}

Multiplying both sides by $\frac{\hat{u}_N'}{\hat{u}_S'}$ and collecting common terms, we have
\begin{gather*}
\begin{aligned}
    \frac{\hat{u}_N'}{\hat{u}_S'} \frac{C_N''}{C_S''}
    \geq
   \frac{ \hat{u}_N' \left( \hat{D}'_N - \tilde{D}'_N - \tilde{D}'_S \right)+ \hat{u}_S' \hat{D}'_S  }
   {\hat{u}_S' \left( \tilde{D}'_S -  \hat{D}'_S + \tilde{D}'_N \right) - \hat{u}_N' \hat{D}'_N}.
\end{aligned}
\end{gather*}

So far, we have established that
\begin{gather}\label{E: N>Udiff implication}
\begin{aligned}
\tilde{A} \geq \hat{A} \iff
   \frac{\hat{u}_N'}{\hat{u}_S'} \frac{C_N''}{C_S''}
    \geq
   \frac{ \hat{u}_N' \left( \hat{D}'_N - \tilde{D}'_N - \tilde{D}'_S \right)+ \hat{u}_S' \hat{D}'_S  }
   {\hat{u}_S' \left( \tilde{D}'_S -  \hat{D}'_S + \tilde{D}'_N \right) - \hat{u}_N' \hat{D}'_N}.
\end{aligned}
\end{gather}
 Next, we show that $\frac{\hat{u}_S'}{\hat{u}_N'} \frac{\hat{D}'_S}{\hat{D}'_N} > \frac{C''_N}{C''_S}$ $\implies$  $\hat{A} \leq \tilde{A}$ yields a contradiction.   We start by rearranging $\frac{\hat{u}_S'}{\hat{u}_N'} \frac{\hat{D}'_S}{\hat{D}'_N} > \frac{C''_N}{C''_S}$ to $\frac{\hat{D}'_S}{\hat{D}'_N} > \frac{\hat{u}_N'}{\hat{u}_S'} \frac{C''_N}{C''_S}$. We then obtain the following contradiction:
 \begin{gather*}
     \begin{aligned}
         \frac{\hat{u}_N'}{\hat{u}_S'} \frac{C''_N}{C''_S}
         &< \frac{\hat{D}'_S}{\hat{D}'_N} \\
         &= \frac{\hat{u}_S' \hat{D}'_S - \hat{u}_N' \hat{D}'_S}{\hat{u}_S' \hat{D}'_N - \hat{u}_N' \hat{D}'_N}\\
         &\leq \frac{\hat{u}_S' \hat{D}'_S - \hat{u}_N' \tilde{D}'_S}{\hat{u}_S' \tilde{D}'_N - \hat{u}_N' \hat{D}'_N} \\
         &\leq \frac{ \hat{u}_N' \left( \hat{D}'_N - \tilde{D}'_N - \tilde{D}'_S \right)+ \hat{u}_S' \hat{D}'_S  }
   {\hat{u}_S' \left( \tilde{D}'_S -  \hat{D}'_S + \tilde{D}'_N \right) - \hat{u}_N' \hat{D}'_N} \\
        &\leq \frac{\hat{u}_N'}{\hat{u}_S'} \frac{C''_N}{C''_S}.
     \end{aligned}
 \end{gather*}

where the second and third inequalities follow from $\tilde{D}'_i \geq \check{D}'_i$ for all $i$ and the fact that the denominator and the numerator are negative  by Equations \eqref{E: Udiff inequality2} and \eqref{E: Udiff inequality3} \footnote{
    Note that the denominator in the third line is negative because it is less than the negative denominator in the fourth line. This can be seen as follows:
    \begin{gather*}
        \hat{u}_S' \tilde{D}'_N - \hat{u}_N' \hat{D}'_N
        < \hat{u}_S' ( \underbrace{\tilde{D}'_S -  \hat{D}'_S}_{\geq 0} + \tilde{D}'_N ) - \hat{u}_N' \hat{D}'_N
        < 0.
    \end{gather*}
    }.
    The last inequality follows from the implication of $\tilde{A} \geq \hat{A}$ documented in Equation~\eqref{E: N>Udiff implication}.

    We have reached the contradiction $\frac{\hat{u}_N'}{\hat{u}_S'} \frac{C''_N}{C''_S} < \frac{\hat{u}_N'}{\hat{u}_S'} \frac{C''_N}{C''_S}$. Hence,   $\frac{\hat{u}_S'}{\hat{u}_N'} \frac{\hat{D}'_S}{\hat{D}'_N} > \frac{C''_N}{C''_S}$  $\implies$  $\hat{A} \leq \tilde{A}$ is incorrect,
    and we have thus shown that we must have   $\frac{\hat{u}_S'}{\hat{u}_N'} \frac{\hat{D}'_S}{\hat{D}'_N} > \frac{C''_N}{C''_S}$  $\implies$  $\hat{A} > \tilde{A}$.

    Together, the proofs of the forward and backward implications yield the equivalence
    \begin{gather*}
        \hat{A} > \tilde{A}
        \iff
        \frac{\hat{u}_S'}{\hat{u}_N'} \frac{\hat{D}'_S}{\hat{D}'_N} > \frac{C''_N}{C''_S} .
    \end{gather*}
 
\end{proof}

\subsection{Generalization of Proposition 2} \label{Assec: Generalization of Proposition 2}

This section generalizes Proposition \ref{Prop: Aggregate abatement Udiff vs Negishi} to a setting with $n$ regions. For an arbitrary number of regions $i \in \mathcal{I} =\{1,\dots,n\} $, the \textit{differentiated carbon price optimization problem} is 
\begin{gather*}
    \max_{X_i, A_i} \sum_i L_i \alpha_i u \left(\frac{X_i}{L_i}\right) \\
    \begin{aligned}
    \text{subject to: }
    \label{E: constraints diff1period}
    &X_i = W_i - C_i(A_{i}) - D_i(A), \quad \forall i. \\
    \end{aligned}
\end{gather*}
Solving this problem with utilitarian weights yields the \textit{utilitaran differentiated carbon price}, which is implicitely defined by
\begin{gather*}
\hat{\tau}_i = C_i'(\hat{A}_i) = - \frac{1}{u'(\hat{x}_i)} \sum_{j \in \mathcal{I}} u'(\hat{x}_j) D_j'(\hat{A}).
\end{gather*}
The Negishi-weighted carbon price is given by Equation \eqref{eq: Negishi carbon price (n regions)}.
Proposition \ref{Prop: Aggregate abatement Udiff vs Negishi} then generalizes as follows.

\begin{proposition}[{\(n\)-region analogue of Proposition~\ref{Prop: Aggregate abatement Udiff vs Negishi}}]
Let $\mathcal{I}=\{1,\dots,n\}$. The global abatement under utilitarian differentiated carbon prices is greater than under the Negishi-weighted carbon price, that is $\hat{A}>\tilde{A}$, if and only if
\begin{gather*}
\frac{\sum_{j\in\mathcal{I}} \hat{u}_j' \hat{D}'_j}
{\sum_{j\in\mathcal{I}} \hat{D}'_j}
>
\frac{\sum_{i\in\mathcal{I}} \frac{1}{C_i''}}
{\sum_{i\in\mathcal{I}} \frac{1}{\hat{u}_i' C_i''}} .
\end{gather*}
\end{proposition}

\begin{proof}

We again split the proof into the forward and backward implications.

\vspace{6pt}
\myuline{Proof of forward direction}:
\begin{gather*}
\hat{A} > \tilde{A}
\implies
\frac{\sum_{j\in\mathcal{I}} \hat{u}_j' \hat{D}'_j}
{\sum_{j\in\mathcal{I}} \hat{D}'_j}
>
\frac{\sum_{i\in\mathcal{I}} \frac{1}{C_i''}}
{\sum_{i\in\mathcal{I}} \frac{1}{\hat{u}_i' C_i''}} .
\end{gather*}

We begin by establishing the implications of $\hat{A} > \tilde{A}$, using the expression for abatement in Equation \eqref{eq: inverted MAC}.
First, $\hat{A} > \tilde{A}$ implies $\sum_{i\in\mathcal{I}} \hat{A}_i > \sum_{i\in\mathcal{I}} \tilde{A}_i$.
Therefore, $\hat{A} > \tilde{A}$ if and only if
\begin{gather*}
\begin{aligned}
     \sum_{i\in\mathcal{I}} \hat{C}'_i k_i^{-1}
    &>
     \sum_{i\in\mathcal{I}} \tilde{C}'_i k_i^{-1}.
\end{aligned}
\end{gather*}

Plugging in the expressions for the marginal abatement costs under the Negishi and utilitarian differentiated solutions, and rewriting, we get
\begin{gather*}
\begin{aligned}
\left(\sum_{j\in\mathcal{I}} - \hat{u}_j' \hat{D}'_j\right)
\left(\sum_{i\in\mathcal{I}} \frac{k_i^{-1}}{\hat{u}_i'}\right)
>
\left(\sum_{j\in\mathcal{I}} - \tilde{D}'_j\right)
\left(\sum_{i\in\mathcal{I}} k_i^{-1}\right).
\end{aligned}
\end{gather*}

Next, note that $\tilde{A} < \hat{A}$ implies $\tilde{D}'_j < \hat{D}'_j$, for all $j\in\mathcal{I}$.
Since $\hat{D}'_j<0$ and $\tilde{D}'_j<0$, this implies
\begin{gather*}
\begin{aligned}
\sum_{j\in\mathcal{I}} - \hat{D}'_j
<
\sum_{j\in\mathcal{I}} - \tilde{D}'_j .
\end{aligned}
\end{gather*}
Therefore, the previous inequality implies
\begin{gather*}
\begin{aligned}
\left(\sum_{j\in\mathcal{I}} - \hat{u}_j' \hat{D}'_j\right)
\left(\sum_{i\in\mathcal{I}} \frac{k_i^{-1}}{\hat{u}_i'}\right)
>
\left(\sum_{j\in\mathcal{I}} - \hat{D}'_j\right)
\left(\sum_{i\in\mathcal{I}} k_i^{-1}\right).
\end{aligned}
\end{gather*}

Using $k_i = \frac{C_i''}{2}$ and rearranging yields

\begin{gather*}
\frac{\sum_{j\in\mathcal{I}} \hat{u}_j' \hat{D}'_j}
{\sum_{j\in\mathcal{I}} \hat{D}'_j}
>
\frac{\sum_{i\in\mathcal{I}} \frac{1}{C_i''}}
{\sum_{i\in\mathcal{I}} \frac{1}{\hat{u}_i' C_i''}} .
\end{gather*}
This proves the forward direction.

\vspace{6pt}
\myuline{Proof of backward direction}:
\begin{gather*}
\frac{\sum_{j\in\mathcal{I}} \hat{u}_j' \hat{D}'_j}
{\sum_{j\in\mathcal{I}} \hat{D}'_j}
>
\frac{\sum_{i\in\mathcal{I}} \frac{1}{C_i''}}
{\sum_{i\in\mathcal{I}} \frac{1}{\hat{u}_i' C_i''}}
\implies
\hat{A} > \tilde{A}.
\end{gather*}

In order to derive a contradiction, suppose that
\[
\frac{\sum_{j\in\mathcal{I}} \hat{u}_j' \hat{D}'_j}
{\sum_{j\in\mathcal{I}} \hat{D}'_j}
>
\frac{\sum_{i\in\mathcal{I}} \frac{1}{C_i''}}
{\sum_{i\in\mathcal{I}} \frac{1}{\hat{u}_i' C_i''}}
\implies
\hat{A} \leq \tilde{A}.
\]

We start by establishing the implications of $\hat{A} \leq \tilde{A}$.
$\hat{A} \leq \tilde{A}$ implies $\sum_{i\in\mathcal{I}} \hat{A}_i \leq \sum_{i\in\mathcal{I}} \tilde{A}_i$.
Therefore, $\hat{A} \leq \tilde{A}$ if and only if
\begin{gather*}
\begin{aligned}
     \sum_{i\in\mathcal{I}} \hat{C}'_i k_i^{-1}
    &\leq
     \sum_{i\in\mathcal{I}} \tilde{C}'_i k_i^{-1}.
\end{aligned}
\end{gather*}

Plugging in the expressions for the marginal abatement costs under the Negishi and utilitarian differentiated solutions, and rewriting, we get
\begin{gather}
\label{E: Udiff N inequality1}
\begin{aligned}
\left(\sum_{j\in\mathcal{I}} - \hat{u}_j' \hat{D}'_j\right)
\left(\sum_{i\in\mathcal{I}} \frac{k_i^{-1}}{\hat{u}_i'}\right)
\leq
\left(\sum_{j\in\mathcal{I}} - \tilde{D}'_j\right)
\left(\sum_{i\in\mathcal{I}} k_i^{-1}\right).
\end{aligned}
\end{gather}

Next, note that $\hat{A} \leq \tilde{A}$ implies $\hat{D}'_j \leq \tilde{D}'_j$, for all $j\in\mathcal{I}$.
Since $\hat{D}'_j<0$ and $\tilde{D}'_j<0$, this implies
\begin{gather}
\label{E: Udiff N inequality2}
\begin{aligned}
\sum_{j\in\mathcal{I}} - \hat{D}'_j
\geq
\sum_{j\in\mathcal{I}} - \tilde{D}'_j .
\end{aligned}
\end{gather}

Using $k_i=\frac{C_i''}{2}$, the assumed condition can be rewritten as
\begin{gather}
\label{E: Udiff N inequality3}
\begin{aligned}
\left(\sum_{j\in\mathcal{I}} - \hat{u}_j' \hat{D}'_j\right)
\left(\sum_{i\in\mathcal{I}} \frac{k_i^{-1}}{\hat{u}_i'}\right)
>
\left(\sum_{j\in\mathcal{I}} - \hat{D}'_j\right)
\left(\sum_{i\in\mathcal{I}} k_i^{-1}\right).
\end{aligned}
\end{gather}

Combining Equations~\eqref{E: Udiff N inequality2} and \eqref{E: Udiff N inequality3}, we obtain
\begin{gather*}
\begin{aligned}
\left(\sum_{j\in\mathcal{I}} - \hat{u}_j' \hat{D}'_j\right)
\left(\sum_{i\in\mathcal{I}} \frac{k_i^{-1}}{\hat{u}_i'}\right)
>
\left(\sum_{j\in\mathcal{I}} - \tilde{D}'_j\right)
\left(\sum_{i\in\mathcal{I}} k_i^{-1}\right),
\end{aligned}
\end{gather*}
which contradicts Equation~\eqref{E: Udiff N inequality1}.
Hence, 
\begin{gather*}
\frac{\sum_{j\in\mathcal{I}} \hat{u}_j' \hat{D}'_j}
{\sum_{j\in\mathcal{I}} \hat{D}'_j}
>
\frac{\sum_{i\in\mathcal{I}} \frac{1}{C_i''}}
{\sum_{i\in\mathcal{I}} \frac{1}{\hat{u}_i' C_i''}}
\implies\hat{A} > \tilde{A}.
\end{gather*}
Together, the proofs of the forward and backward implications yield the equivalence
\begin{gather*}
\hat{A} > \tilde{A}
\iff
\frac{\sum_{j\in\mathcal{I}} \hat{u}_j' \hat{D}'_j}
{\sum_{j\in\mathcal{I}} \hat{D}'_j}
>
\frac{\sum_{i\in\mathcal{I}} \frac{1}{C_i''}}
{\sum_{i\in\mathcal{I}} \frac{1}{\hat{u}_i' C_i''}} .
\end{gather*}

\end{proof}

As before, the intuition of the two-region case carries over to the setting with an arbirary number of regions: 
aggregate abatement under utilitarian differentiated  carbon prices exceeds the aggregate abatement under Negishi-weighted carbon prices
when poorer regions tend to face relatively higher marginal damages and steeper marginal abatement cost functions.

\subsection{Proof of Lemma 2} \label{Assec: Proof of Lemma 2}

I first prove foundational lemmas which act as building blocks to prove Lemma \ref{Lemma: U&N in between preferred carbon price}.








\begin{lemma}
    North's preferred uniform carbon price is less than the utilitarian uniform carbon price, that is $\mathring{\tau}^{N} < \check{\tau}$, if and only if $\frac{\check{D}'_S}{\check{D}'_N} > \frac{C''_N}{C''_S}$.
    \label{lemma: North<U}
\end{lemma}

\begin{proof}
    We split the proof into the forward and backward implications.

    \vspace{6pt}
    \myuline{Proof of forward direction}: $\check{\tau} > \mathring{\tau}^{N}$  $\implies$ $\frac{\check{D}'_S}{\check{D}'_N} > \frac{C''_N}{C''_S}$.
    
    I begin by establishing the conditions under which $\check{\tau} > \mathring{\tau}^{N}$, or equivalently, $\check{C}' > \mathring{C}^{N}\hspace{-0.3mm}'$\footnote{
    While this notation is somewhat cumbersome, I use the notation $\mathring{C}^{i}\hspace{-0.3mm}'$ for clarity and conciseness as a short-hand for ${C}'_i(\mathring{A}_i^i)$, and I drop the subscript to reflect that ${C}'_i(\mathring{A}_i^i) = {C}'_{-i}(\mathring{A}_{-i}^i)$ under uniform carbon prices.
    }. First note that, for strictly convex abatement cost functions, $\check{C}' > \mathring{C}^{N}\hspace{-0.3mm}'$ implies $\check{A}_i > \mathring{A}_i^N$ for all $i$, and thus $\check{A} > \mathring{A}^N$.
    For strictly convex damage functions, this implies $\check{D}'_i > \mathring{D}_i^N\hspace{-0.3mm}'$ for all $i$ (note that marginal damages of abatement are negative).

    We have $\check{C}' > \mathring{C}^{N}\hspace{-0.3mm}'$ if and only if
    \begin{gather*}
        (-\check{u}'_N \check{D}'_N - \check{u}'_S \check{D}'_S) \frac{C''_S + C''_N}{\check{u}'_N C''_S + \check{u}'_S C''_N} >
        - \mathring{D}^N_N\hspace{-0.3mm}'
        \frac{C''_S + C''_{N}}{C''_{S}},
    \end{gather*}



    which can be rewritten as
    \begin{gather*}
        - \frac{\check{u}'_S}{\check{u}'_N} \check{D}'_S >
        - \mathring{D}^N_N\hspace{-0.3mm}'
        \left( 1 + \frac{\check{u}'_S C''_N}
        {\check{u}'_N C''_{S}}\right) + \check{D}'_N.
    \end{gather*}

    Using $\check{D}'_i > \mathring{D}_i^N\hspace{-0.3mm}'$, we have
    \begin{gather*}
    \begin{aligned}
        - \frac{\check{u}'_S}{\check{u}'_N} \check{D}'_S 
        &> - \mathring{D}^N_N\hspace{-0.3mm}' 
        \left( 1 + \frac{\check{u}'_S C''_N}{\check{u}'_N C''_{S}}\right) + \check{D}'_N \\
        &> - \check{D}'_N 
        \left( 1 + \frac{\check{u}'_S C''_N}{\check{u}'_N C''_{S}}\right) + \check{D}'_N \\
        &= - \check{D}'_N 
        \left(\frac{\check{u}'_S C''_N}{\check{u}'_N C''_{S}}\right).
    \end{aligned}
    \end{gather*}

    Rewriting yields
    \begin{gather*}
        \frac{\check{D}'_S}{\check{D}'_N}
        > \frac{C''_{N}}{C''_{S}}.
    \end{gather*}

We have thus shown that $\check{C}' > \mathring{C}^{N}\hspace{-0.3mm}' \implies \frac{\check{D}'_S}{\check{D}'_N} > \frac{C''_N}{C''_S}$.



    \vspace{6pt}
    \myuline{Proof of backward direction}: $\frac{\check{D}'_S}{\check{D}'_N} > \frac{C''_N}{C''_S}$  $\implies$ $\check{\tau} > \mathring{\tau}^N$.

    In order to derive a contradiction, suppose that $\frac{\check{D}'_S}{\check{D}'_N} > \frac{C''_N}{C''_S}$  $\implies$ $\check{\tau} \leq \mathring{\tau}^{N}$.

    We begin by establishing the implications of $\check{\tau} \leq \mathring{\tau}^{N}$, or equivalently, $\check{C}' \leq \mathring{C}^{N}\hspace{-0.3mm}'$. First note that, for strictly convex abatement cost functions, $\check{C}' \leq \mathring{C}^{N}\hspace{-0.3mm}'$ implies $\check{A}_i \leq \mathring{A}_i^N$ for all $i$, and thus $\check{A} \leq \mathring{A}^N$. For strictly convex damage functions, this implies $\check{D}'_i \leq \mathring{D}_i^N\hspace{-0.3mm}'$ for all $i$ (note that marginal damages of abatement are negative).

    Next, note that $\check{C}' \leq \mathring{C}^{N}\hspace{-0.3mm}'$ if and only if
    \begin{gather*}
        (-\check{u}'_N \check{D}'_N - \check{u}'_S \check{D}'_S) \frac{C''_S + C''_N}{\check{u}'_N C''_S + \check{u}'_S C''_N} \leq
        - \mathring{D}^N_N\hspace{-0.3mm}'
        \frac{C''_S + C''_{N}}
        {C''_{S}},
    \end{gather*}



    which can be rewritten as
    \begin{gather*}
        - \frac{\check{u}'_S}{\check{u}'_N} \check{D}'_S 
        \leq
        - \mathring{D}^N_N\hspace{-0.3mm}'
        \left( 1 + \frac{\check{u}'_S C''_N}
        {\check{u}'_N C''_{S}}\right) + \check{D}'_N.
    \end{gather*}

    Let us temporarily define $\delta_N \equiv \mathring{D}_N^N\hspace{-0.3mm}' - \check{D}'_N$\footnote{
    Note that I redefine $\delta_i$ below, keeping the same notation for simplicity.
    }. We know that $\delta_N \geq 0$ since $\check{D}'_N \leq \mathring{D}_N^N\hspace{-0.3mm}'$.
    Substitute $\check{D}'_N = \mathring{D}_N^N\hspace{-0.3mm}' - \delta_N$ into the previous expression to obtain
    \begin{gather*}
        - \frac{\check{u}'_S}{\check{u}'_N} \check{D}'_S 
        \leq
        - \mathring{D}^N_N\hspace{-0.3mm}'
        \left( 1 + \frac{\check{u}'_S C''_N}
        {\check{u}'_N C''_{S}}\right) + \mathring{D}^N_N\hspace{-0.3mm}' - \delta_N.
    \end{gather*}


    Rewriting yields
    \begin{gather*}
        \frac{\check{D}'_S}{\mathring{D}^N_N\hspace{-0.3mm}'} 
        \leq
        \frac{C''_N}{C''_{S}}
        - \underbrace{\frac{\delta_N}{- \mathring{D}^N_N\hspace{-0.3mm}'} 
        \frac{\check{u}'_N}{\check{u}'_S}}_{\geq 0}.
    \end{gather*}

    So far, we have established that
    \begin{gather}\label{E: U<P_N implication v2}
        \check{C}' \leq \mathring{C}^{N}\hspace{-0.3mm}'
        \iff
        \frac{\check{D}'_S}{\mathring{D}^N_N\hspace{-0.3mm}'} 
        \leq
        \frac{C''_N}{C''_{S}}
        - \underbrace{\frac{\delta_N}{- \mathring{D}^N_N\hspace{-0.3mm}'} 
        \frac{\check{u}'_N}{\check{u}'_S}}_{\geq 0}.
    \end{gather} 

     Next, we show that $\frac{\check{D}'_S}{\check{D}'_N} > \frac{C''_N}{C''_S}$  $\implies$ $\check{C}' \leq \mathring{C}^{N}\hspace{-0.3mm}'$ yields a contradiction:
     \begin{gather*}
    \begin{aligned}
        \frac{C''_N}{C''_S} < \frac{\check{D}'_S}{\check{D}'_N}
        \leq \frac{\check{D}'_S}{\mathring{D}^N_N\hspace{-0.3mm}'}
        \leq \frac{C''_N}{C''_{S}}
        - \underbrace{\frac{\delta_N}{- \mathring{D}^N_N\hspace{-0.3mm}'} 
        \frac{\check{u}'_N}{\check{u}'_S}}_{\geq 0}
        \leq \frac{C''_N}{C''_S},
    \end{aligned}
    \end{gather*}


    where the second inequality follows from $\check{D}'_i \leq \mathring{D}_i^N\hspace{-0.3mm}'$,
    and the third inequality follows from the implication of $\check{C}' \leq \mathring{C}^{N}\hspace{-0.3mm}'$ documented in Equation~\eqref{E: U<P_N implication v2}.

    We have reached the contradiction $\frac{C''_N}{C''_S} < \frac{C''_N}{C''_S}$. Hence, $\frac{\check{D}'_S}{\check{D}'_N} > \frac{C''_N}{C''_S}$  $\implies$ $\check{C}' \leq \mathring{C}^{N}\hspace{-0.3mm}'$ is incorrect,
    and we have thus shown that we must have $\frac{\check{D}'_S}{\check{D}'_N} > \frac{C''_N}{C''_S}  \implies \check{C}' > \mathring{C}^{N}\hspace{-0.3mm}'$.

    Together, the proofs of the forward and backward directions yield the equivalence
    \begin{gather*}
        \check{\tau} > \mathring{\tau}^{N}
        \iff
        \frac{\check{D}'_S}{\check{D}'_N} > \frac{C''_N}{C''_S}.
    \end{gather*}

\end{proof}




\begin{lemma}
    South's preferred uniform carbon price is greater than the utilitarian uniform carbon price, that is $\mathring{\tau}^{S} > \check{\tau}$, if and only if $\frac{\check{D}'_S}{\check{D}'_N} > \frac{C''_N}{C''_S}$.
    \label{lemma: South>U}
\end{lemma}

\begin{proof}
    We split the proof into the forward and backward implications.

    \vspace{6pt}
    \myuline{Proof of forward direction}: $\mathring{\tau}^{S} > \check{\tau}$  $\implies$ $\frac{\check{D}'_S}{\check{D}'_N} > \frac{C''_N}{C''_S}$.
    
   I begin by establishing the conditions under which $\mathring{\tau}^{S} > \check{\tau}$, or equivalently, $\check{C}' < \mathring{C}^{S}\hspace{-0.3mm}'$. First note that, for strictly convex abatement cost functions, $\check{C}' < \mathring{C}^{S}\hspace{-0.3mm}'$ implies $\check{A}_i < \mathring{A}_i^S$ for all $i$, and thus $\check{A} < \mathring{A}^S$. For strictly convex damage functions, this implies $\check{D}'_i < \mathring{D}_i^S\hspace{-0.3mm}'$ for all $i$ (note that marginal damages of abatement are negative).

    We have  $\check{C}' < \mathring{C}^{S}\hspace{-0.3mm}'$ if and only if
    \begin{gather*}
        (-\check{u}'_N \check{D}'_N - \check{u}'_S \check{D}'_S) \frac{C''_S + C''_N}{\check{u}'_N C''_S + \check{u}'_S C''_N} <
        - \mathring{D}^S_S\hspace{-0.3mm}'
        \frac{C''_S + C''_{N}}{C''_{N}},
    \end{gather*}



    which can be rewritten as
    \begin{gather*}
        - \frac{\check{u}'_N}{\check{u}'_S}\check{D}'_N <
        - \mathring{D}^S_S\hspace{-0.3mm}'
        \left( 1 + \frac{\check{u}'_N C''_S}
        {\check{u}'_S C''_{N}}\right)
        + \check{D}'_S.
    \end{gather*}

    Using $\check{D}'_i < \mathring{D}_i^S\hspace{-0.3mm}'$, we have
    \begin{gather*}
    \begin{aligned}
        - \frac{\check{u}'_N}{\check{u}'_S}\check{D}'_N
        &< - \mathring{D}^S_S\hspace{-0.3mm}'
        \left( 1 + \frac{\check{u}'_N C''_S} {\check{u}'_S C''_{N}}\right) + \check{D}'_S \\
        &< - \check{D}'_S
        \left( 1 + \frac{\check{u}'_N C''_S} {\check{u}'_S C''_{N}}\right) + \check{D}'_S \\
        &= - \check{D}'_S
        \left(\frac{\check{u}'_N C''_S} {\check{u}'_S C''_{N}}\right).
    \end{aligned}
    \end{gather*}

    Rewriting yields
    \begin{gather*}
        \frac{\check{D}'_S}{\check{D}'_N}
        > \frac{C''_{N}}{C''_{S}}.
    \end{gather*}

We have thus shown that $\check{C}' < \mathring{C}^{S}\hspace{-0.3mm}' \implies \frac{\check{D}'_S}{\check{D}'_N} > \frac{C''_N}{C''_S}$.



    \vspace{6pt}
    \myuline{Proof of backward direction}: $\frac{\check{D}'_S}{\check{D}'_N} > \frac{C''_N}{C''_S}$  $\implies$ $\check{\tau} < \mathring{\tau}^{S}$.

    In order to derive a contradiction, suppose that $\frac{\check{D}'_S}{\check{D}'_N} > \frac{C''_N}{C''_S}$  $\implies$ $\check{\tau} \geq \mathring{\tau}^{S}$.

    We begin by establishing the implications of $\check{\tau} \geq \mathring{\tau}^{S}$, or equivalently $\check{C}' \geq \mathring{C}^{S}\hspace{-0.3mm}'$. First note that, for strictly convex abatement cost functions, $\check{C}' \geq \mathring{C}^{S}\hspace{-0.3mm}'$ implies $\check{A}_i \geq \mathring{A}_i^S$ for all $i$, and thus $\check{A} \geq \mathring{A}^S$. For strictly convex damage functions, this implies $\check{D}'_i \geq \mathring{D}_i^S\hspace{-0.3mm}'$ for all $i$ (note that marginal damages of abatement are negative).

    Next, note that $\check{C}' \geq \mathring{C}^{S}\hspace{-0.3mm}'$ if and only if
    \begin{gather*}
        (-\check{u}'_N \check{D}'_N - \check{u}'_S \check{D}'_S) \frac{C''_S + C''_N}{\check{u}'_N C''_S + \check{u}'_S C''_N} \geq
        - \mathring{D}^S_S\hspace{-0.3mm}'
        \frac{C''_S + C''_{N}}
        {C''_{N}},
    \end{gather*}



    which can be rewritten as
    \begin{gather*}
        - \frac{\check{u}'_N}{\check{u}'_S}\check{D}'_N
        \geq
        - \mathring{D}^S_S\hspace{-0.3mm}'
        \left( 1 + \frac{\check{u}'_N C''_S}
        {\check{u}'_S C''_{N}}\right)
        + \check{D}'_S.
    \end{gather*}

    Let us temporarily define $\delta_S \equiv \check{D}'_S - \mathring{D}_S^S\hspace{-0.3mm}'$. We know that $\delta_S \geq 0$ since $\check{D}'_S \geq \mathring{D}_S^S\hspace{-0.3mm}'$.
    Substitute $\check{D}'_S = \mathring{D}_S^S\hspace{-0.3mm}' + \delta_S$ into the previous expression to obtain
    \begin{gather*}
        - \frac{\check{u}'_N}{\check{u}'_S}\check{D}'_N
        \geq
        - \mathring{D}^S_S\hspace{-0.3mm}'
        \left( 1 + \frac{\check{u}'_N C''_S}
        {\check{u}'_S C''_{N}}\right)
        + \mathring{D}_i^S\hspace{-0.3mm}' + \delta_i.
    \end{gather*}


    Simplifying and rearranging yields
    \begin{gather*}
        \frac{\check{D}'_N}{\mathring{D}^S_S\hspace{-0.3mm}'} 
        \geq
        \frac{C''_S}{C''_{N}}
        + \underbrace{\frac{\delta_S}{- \mathring{D}^S_S\hspace{-0.3mm}'} 
        \frac{\check{u}'_S}{\check{u}'_N}}_{\geq 0} .
    \end{gather*}
    
    So far, we have established that
    \begin{gather}\label{E: U<P_S implication v2}
        \check{C}' \geq \mathring{C}^{S}\hspace{-0.3mm}'
        \iff
        \frac{\check{D}'_N}{\mathring{D}^S_S\hspace{-0.3mm}'} 
        \geq
        \frac{C''_S}{C''_{N}}
        + \underbrace{\frac{\delta_S}{- \mathring{D}^S_S\hspace{-0.3mm}'} 
        \frac{\check{u}'_S}{\check{u}'_N}}_{\geq 0}.
    \end{gather} 

     Next, we show that $\frac{\check{D}'_S}{\check{D}'_N} > \frac{C''_N}{C''_S}$  $\implies$ $\check{C}' \geq \mathring{C}^{S}\hspace{-0.3mm}'$ yields a contradiction.
     We start by rearranging $\frac{\check{D}'_S}{\check{D}'_N} > \frac{C''_N}{C''_S}$ to $\frac{\check{D}'_N}{\check{D}'_S} < \frac{C''_S}{C''_N}$. We then obtain the following contradiction:
    \begin{gather*}
    \begin{aligned}
        \frac{C''_S}{C''_N} > \frac{\check{D}'_N}{\check{D}'_S}
        \geq \frac{\check{D}'_N}{\mathring{D}^S_S\hspace{-0.3mm}'}
        \geq \frac{C''_S}{C''_{N}}
        + \underbrace{\frac{\delta_S}{- \mathring{D}^S_S\hspace{-0.3mm}'} 
        \frac{\check{u}'_S}{\check{u}'_N}}_{\geq 0}
        \geq \frac{C''_S}{C''_N} .
    \end{aligned}
    \end{gather*}
    

    where the second inequality follows from $\check{D}'_i \geq \mathring{D}_i^S\hspace{-0.3mm}'$, and the third inequality follows from the implication of $\check{C}' \geq \mathring{C}^{S}\hspace{-0.3mm}'$ documented in Equation~\eqref{E: U<P_S implication v2}.
    
    We have reached the contradiction $\frac{C''_S}{C''_N} > \frac{C''_S}{C''_N}$. Hence, $\frac{\check{D}'_S}{\check{D}'_N} > \frac{C''_N}{C''_S}$  $\implies$ $\check{C}' \geq \mathring{C}^{S}\hspace{-0.3mm}'$ is incorrect, and we have thus shown that we must have $\frac{\check{D}'_S}{\check{D}'_N} > \frac{C''_N}{C''_S}  \implies \check{C}' < \mathring{C}^{S}\hspace{-0.3mm}'$.
    
    Together, the proofs of the forward and backward directions yield the equivalence
    \begin{gather*}
        \check{\tau} < \mathring{\tau}^{S}
        \iff
        \frac{\check{D}'_S}{\check{D}'_N} > \frac{C''_N}{C''_S}.
    \end{gather*}

\end{proof}


\begin{lemma}
    North's preferred uniform carbon price is less than the Negishi-weighted carbon price, that is $\mathring{\tau}^{N} < \tilde{\tau}$, if and only if $\frac{\tilde{D}'_S}{\tilde{D}'_N} > \frac{C''_N}{C''_S}$.
    \label{lemma: N<Negishi}
\end{lemma}

\begin{proof}
    We split the proof into the forward and backward implications.


    \vspace{6pt}
    \myuline{Proof of forward direction}: $\tilde{\tau} > \mathring{\tau}^{N}$  $\implies$ $\frac{\tilde{D}'_S}{\tilde{D}'_N} > \frac{C''_N}{C''_S}$.
    
    I begin by establishing the conditions under which $\tilde{\tau} > \mathring{\tau}^{N}$, or equivalently, $\tilde{C}' > \mathring{C}^{N}\hspace{-0.3mm}'$. First note that, for strictly convex abatement cost functions, $\tilde{C}' > \mathring{C}^{N}\hspace{-0.3mm}'$ implies $\tilde{A}_i > \mathring{A}_i^N$ for all $i$, and thus $\tilde{A} > \mathring{A}^N$. For strictly convex damage functions, this implies $\tilde{D}'_i > \mathring{D}_i^N\hspace{-0.3mm}'$ for all $i$ (note that marginal damages of abatement are negative).

    We have  $\tilde{C}' > \mathring{C}^{N}\hspace{-0.3mm}'$ if and only if
    \begin{gather*}
        - \tilde{D}'_N - \tilde{D}'_S >
        - \mathring{D}^N_N\hspace{-0.3mm}'
        \frac{C''_S + C''_{N}}
        {C''_{S}},
    \end{gather*}
    

    which we can rewrite as (note that $\mathring{D}^N_N\hspace{-0.3mm}'$ is negative so the sign of the inequality flips)
    \begin{gather*}
            \frac{C''_S + C''_{N}}
            {C''_{S}} <
            \frac{\tilde{D}'_N}{\mathring{D}^N_N\hspace{-0.3mm}'} + \frac{\tilde{D}'_S}{\mathring{D}^N_N\hspace{-0.3mm}'}.
    \end{gather*}

    Let us temporarily define $\delta_N \equiv \mathring{D}_N^N\hspace{-0.3mm}' - \tilde{D}'_N$. We know that $\delta_N < 0$ since $\tilde{D}'_N > \mathring{D}_N^N\hspace{-0.3mm}'$.
    Substitute $\tilde{D}'_N = \mathring{D}_N^N\hspace{-0.3mm}' - \delta_N$ into the previous expression to obtain
    \begin{gather*}
            \frac{C''_S + C''_{N}}
            {C''_{S}} <
            \frac{\mathring{D}^N_N\hspace{-0.3mm}' - \delta_N}{\mathring{D}^N_N\hspace{-0.3mm}'} + \frac{\tilde{D}'_S}{\mathring{D}^N_N\hspace{-0.3mm}'},
    \end{gather*}

    which simplifies to
    \begin{gather*}
            \frac{C''_N} {C''_{S}} <
            \frac{\tilde{D}'_S}{\mathring{D}^N_N\hspace{-0.3mm}'} + \frac{-\delta_N}{\mathring{D}^N_N\hspace{-0.3mm}'}.
    \end{gather*}

    We can now establish the following inequalities:
    \begin{gather*}
    \begin{aligned}
        \frac{C''_N} {C''_{S}} <
        \frac{\tilde{D}'_S}{\mathring{D}^N_N\hspace{-0.3mm}'} + \underbrace{\frac{-\delta_N}{\mathring{D}^N_N\hspace{-0.3mm}'}}_{<0}
        < \frac{\tilde{D}'_S}{\mathring{D}^N_N\hspace{-0.3mm}'}
        < \frac{\tilde{D}'_S}{\tilde{D}'_N}.
    \end{aligned}  
    \end{gather*}

    where the last inequality follows from $\tilde{D}'_N > \mathring{D}_N^N\hspace{-0.3mm}'$.

    We have thus shown that $\tilde{C}' > \mathring{C}^{N}\hspace{-0.3mm}' \implies \frac{\tilde{D}'_S}{\tilde{D}'_N} > \frac{C''_N}{C''_S}$.
    


    \vspace{6pt}
    \myuline{Proof of backward direction}: $\frac{\tilde{D}'_S}{\tilde{D}'_N} > \frac{C''_N}{C''_S}$  $\implies$ $\tilde{\tau} > \mathring{\tau}^N$.

    In order to derive a contradiction, suppose that $\frac{\tilde{D}'_S}{\tilde{D}'_N} > \frac{C''_N}{C''_S}$  $\implies$ $\tilde{\tau} \leq \mathring{\tau}^N$.
    
    We begin by establishing the implications of $\tilde{\tau} \leq \mathring{\tau}^N$, or equivalently, $\tilde{C}' \leq \mathring{C}^{N}\hspace{-0.3mm}'$. First note that, for strictly convex abatement cost functions, $\tilde{C}' \leq \mathring{C}^{N}\hspace{-0.3mm}'$ implies $\tilde{A}_i \leq \mathring{A}_i^N$ for all $i$, and thus $\tilde{A} \leq \mathring{A}^N$. For strictly convex damage functions, this implies $\tilde{D}'_i \leq \mathring{D}_i^N\hspace{-0.3mm}'$ for all $i$ (note that marginal damages of abatement are negative).

    Next, note that $\tilde{C}' \leq \mathring{C}^{N}\hspace{-0.3mm}'$ if and only if
    \begin{gather*}
        - \tilde{D}'_N - \tilde{D}'_S \leq - \mathring{D}^N_N\hspace{-0.3mm}'
        \frac{C''_S + C''_{N}}
        {C''_{S}},
    \end{gather*}
    

which we can rewrite as (note that $\mathring{D}^N_N\hspace{-0.3mm}'$ is negative so the sign of the inequality flips)
    \begin{gather*}
            \frac{C''_S + C''_{N}}
            {C''_{S}} \geq
            \frac{\tilde{D}'_N}{\mathring{D}^N_N\hspace{-0.3mm}'} + \frac{\tilde{D}'_S}{\mathring{D}^N_N\hspace{-0.3mm}'}.
    \end{gather*}

     Let us temporarily define $\delta_N \equiv \mathring{D}_N^N\hspace{-0.3mm}' - \tilde{D}'_N$. We know that $\delta_N \geq 0$ since $\tilde{D}'_N \leq \mathring{D}_N^N\hspace{-0.3mm}'$.
    Substitute $\tilde{D}'_N = \mathring{D}_N^N\hspace{-0.3mm}' - \delta_N$ into the previous expression to obtain    
   \begin{gather*}
            \frac{C''_S + C''_{N}}
            {C''_{S}} \geq
            \frac{\mathring{D}^N_N\hspace{-0.3mm}' - \delta_N}{\mathring{D}^N_N\hspace{-0.3mm}'} + \frac{\tilde{D}'_S}{\mathring{D}^N_N\hspace{-0.3mm}'},
    \end{gather*}

    which simplifies to
    \begin{gather*}
    \begin{aligned}
        \frac{C''_N} {C''_{S}} &\geq
        \frac{\tilde{D}'_S}{\mathring{D}^N_N\hspace{-0.3mm}'} + \underbrace{\frac{-\delta_N}{\mathring{D}^N_N\hspace{-0.3mm}'}}_{\geq0} .\\
    \end{aligned}  
    \end{gather*}
    
     So far, we have established that
    \begin{gather}\label{E: N>P_N implication v2}
        \tilde{C}' \leq \mathring{C}^{N}\hspace{-0.3mm}'
        \iff
        \frac{C''_N} {C''_{S}} \geq
        \frac{\tilde{D}'_S}{\mathring{D}^N_N\hspace{-0.3mm}'} + \underbrace{\frac{-\delta_N}{\mathring{D}^N_N\hspace{-0.3mm}'}}_{\geq0} .
    \end{gather}

    Next, we show that $\frac{\tilde{D}'_S}{\tilde{D}'_N} > \frac{C''_N}{C''_S}$  $\implies$ $\tilde{C}' \leq \mathring{C}^{N}\hspace{-0.3mm}'$ yields a contradiction.
    \begin{gather*}
    \begin{aligned}
        \frac{C''_N}{C''_S} < \frac{\tilde{D}'_S}{\tilde{D}'_N}
        \leq \frac{\tilde{D}'_S}{\mathring{D}^N_N\hspace{-0.3mm}'}
        \leq \frac{\tilde{D}'_S}{\mathring{D}^N_N\hspace{-0.3mm}'} + \underbrace{\frac{-\delta_N}{\mathring{D}^N_N\hspace{-0.3mm}'}}_{\geq 0}
        \leq \frac{C''_N}{C''_S}.
    \end{aligned}
    \end{gather*}


    where the second and third inequalities follow from $\tilde{D}'_i \leq \mathring{D}_i^N\hspace{-0.3mm}'$ for all $i$, and the last inequality follows from the implication of $\tilde{C}' \leq \mathring{C}^{N}\hspace{-0.3mm}'$ documented in Equation~\eqref{E: N>P_N implication v2}.

    We have reached the contradiction $\frac{C''_N}{C''_S} < \frac{C''_N}{C''_S}$. Hence, $\frac{\tilde{D}'_S}{\tilde{D}'_N} > \frac{C''_N}{C''_S}$  $\implies$ $\tilde{C}' \leq \mathring{C}^{N}\hspace{-0.3mm}'$ is incorrect, 
    and we have thus shown that we must have $\frac{\tilde{D}'_S}{\tilde{D}'_N} > \frac{C''_N}{C''_S}\implies \tilde{C}' > \mathring{C}^{N}\hspace{-0.3mm}'$.

    Together, the proofs of the forward and backward directions yield the equivalence
    \begin{gather*}
        \tilde{\tau} > \mathring{\tau}^{N}
        \iff
        \frac{\tilde{D}'_S}{\tilde{D}'_N} > \frac{C''_N}{C''_S}.
    \end{gather*}

\end{proof}


\begin{lemma}
    South's preferred uniform carbon price is greater than the Negishi-weighted carbon price, that is $\mathring{\tau}^{S} > \tilde{\tau}$, if and only if $\frac{\tilde{D}'_S}{\tilde{D}'_N} > \frac{C''_N}{C''_S}$.
    \label{lemma: S>Negishi}
\end{lemma}


\begin{proof}
    We split the proof into the forward and backward implications.


    \vspace{6pt}
    \myuline{Proof of forward direction}: $\tilde{\tau} < \mathring{\tau}^{S}$  $\implies$ $\frac{\tilde{D}'_S}{\tilde{D}'_N} > \frac{C''_N}{C''_S}$.
    
    I begin by establishing the conditions under which $\tilde{\tau} < \mathring{\tau}^{S}$, or equivalently, $\tilde{C}' < \mathring{C}^{S}\hspace{-0.3mm}'$. First note that, for strictly convex abatement cost functions, $\tilde{C}' < \mathring{C}^{S}\hspace{-0.3mm}'$ implies $\tilde{A}_i < \mathring{A}_i^S$ for all $i$, and thus $\tilde{A} < \mathring{A}^S$. For strictly convex damage functions, this implies $\tilde{D}'_i < \mathring{D}_i^S\hspace{-0.3mm}'$ for all $i$ (note that marginal damages of abatement are negative).

    We have  $\tilde{C}' < \mathring{C}^{S}\hspace{-0.3mm}'$ if and only if
    \begin{gather*}
        - \tilde{D}'_N - \tilde{D}'_S <
        - \mathring{D}^S_S\hspace{-0.3mm}'
        \frac{C''_S + C''_{N}}
        {C''_{N}},
    \end{gather*}
    

    which we can rewrite as (note that $\mathring{D}^S_S\hspace{-0.3mm}'$ is negative so the sign of the inequality flips)
    \begin{gather*}
            \frac{C''_S + C''_{N}}
            {C''_{N}} >
            \frac{\tilde{D}'_N}{\mathring{D}^S_S\hspace{-0.3mm}'} + \frac{\tilde{D}'_S}{\mathring{D}^S_S\hspace{-0.3mm}'}.
    \end{gather*}

 Let us temporarily define $\delta_S \equiv \mathring{D}_S^S\hspace{-0.3mm}' - \tilde{D}'_S$. We know that $\delta_S > 0$ since $\tilde{D}'_S < \mathring{D}_S^S\hspace{-0.3mm}'$.
    Substitute $\tilde{D}'_S = \mathring{D}_S^S\hspace{-0.3mm}' - \delta_S$ into the previous expression to obtain
    \begin{gather*}
            \frac{C''_S + C''_{N}}
            {C''_{N}} >
            \frac{\tilde{D}'_N}{\mathring{D}^S_S\hspace{-0.3mm}'} + \frac{\mathring{D}^S_S\hspace{-0.3mm}' - \delta_S}{\mathring{D}^S_S\hspace{-0.3mm}'},
    \end{gather*}

    which simplifies to
    \begin{gather*}
            \frac{C''_S} {C''_{N}} >
            \frac{\tilde{D}'_N}{\mathring{D}^S_S\hspace{-0.3mm}'} + \frac{-\delta_S}{\mathring{D}^S_S\hspace{-0.3mm}'} .
    \end{gather*}

    We can now establish the following inequalities: 
    \begin{gather*}
    \begin{aligned}
        \frac{C''_S} {C''_{N}} >
        \frac{\tilde{D}'_N}{\mathring{D}^S_S\hspace{-0.3mm}'} + \underbrace{\frac{-\delta_S}{\mathring{D}^S_S\hspace{-0.3mm}'}}_{>0} 
        > \frac{\tilde{D}'_N}{\mathring{D}^S_S\hspace{-0.3mm}'} 
        > \frac{\tilde{D}'_N}{\tilde{D}'_S} .
    \end{aligned}  
    \end{gather*}

where the last inequality follows from $\tilde{D}'_S < \mathring{D}_S^S\hspace{-0.3mm}'$.

We have thus shown that $\tilde{C}' < \mathring{C}^{S}\hspace{-0.3mm}' \implies \frac{\tilde{D}'_S}{\tilde{D}'_N} > \frac{C''_N}{C''_S}$.



    \vspace{6pt}
    \myuline{Proof of backward direction}: $\frac{\tilde{D}'_S}{\tilde{D}'_N} > \frac{C''_N}{C''_S}$  $\implies$ $\tilde{\tau} < \mathring{\tau}^{S}$.

    In order to derive a contradiction, suppose that $\frac{\tilde{D}'_S}{\tilde{D}'_N} > \frac{C''_N}{C''_S}$  $\implies$ $\tilde{\tau} \geq \mathring{\tau}^{S}$.
    
    We start by establishing the implications of $\tilde{\tau} \geq \mathring{\tau}^{S}$, or equivalently, $\tilde{C}' \geq \mathring{C}^{S}\hspace{-0.3mm}'$. First note that, for strictly convex abatement cost functions, $\tilde{C}' \geq \mathring{C}^{S}\hspace{-0.3mm}'$ implies $\tilde{A}_i \geq \mathring{A}_i^S$ for all $i$, and thus $\tilde{A} \geq \mathring{A}^S$. For strictly convex damage functions, this implies $\tilde{D}'_i \geq \mathring{D}_i^S\hspace{-0.3mm}'$ for all $i$ (note that marginal damages of abatement are negative).

    Next, note that $\tilde{C}' \geq \mathring{C}^{S}\hspace{-0.3mm}'$ if and only if
    \begin{gather*}
        - \tilde{D}'_N - \tilde{D}'_S \geq - \mathring{D}^S_S\hspace{-0.3mm}'
        \frac{C''_S + C''_{N}}
        {C''_{N}} ,
    \end{gather*}
    

    which we can rewrite as (note that $\mathring{D}^S_S\hspace{-0.3mm}'$ is negative so the sign of the inequality flips)
    \begin{gather*}
            \frac{C''_S + C''_{N}}
            {C''_{N}} \leq
            \frac{\tilde{D}'_N}{\mathring{D}^S_S\hspace{-0.3mm}'} + \frac{\tilde{D}'_S}{\mathring{D}^S_S\hspace{-0.3mm}'} .
    \end{gather*}
    
   Let us temporarily define $\delta_S \equiv \mathring{D}_S^S\hspace{-0.3mm}' - \tilde{D}'_S$. We know that $\delta_S \leq 0$ since $\tilde{D}'_S \geq \mathring{D}_S^S\hspace{-0.3mm}'$.
    Substitute $\tilde{D}'_S = \mathring{D}_S^S\hspace{-0.3mm}' - \delta_S$ into the previous expression to obtain
\begin{gather*}
            \frac{C''_S + C''_{N}}
            {C''_{N}} \leq
            \frac{\tilde{D}'_N}{\mathring{D}^S_S\hspace{-0.3mm}'} + \frac{\mathring{D}^S_S\hspace{-0.3mm}' - \delta_S}{\mathring{D}^S_S\hspace{-0.3mm}'},
    \end{gather*}


    which simplifies to
    \begin{gather*}
    \begin{aligned}
        \frac{C''_S} {C''_{N}} &\leq
        \frac{\tilde{D}'_N}{\mathring{D}^S_S\hspace{-0.3mm}'} + \underbrace{\frac{-\delta_S}{\mathring{D}^S_S\hspace{-0.3mm}'}}_{\leq 0}. \\
    \end{aligned}  
    \end{gather*}

     So far, we have established that
    \begin{gather}\label{E: N>P_S implication v2}
        \tilde{C}' \geq \mathring{C}^{S}\hspace{-0.3mm}'
        \iff
        \frac{C''_S} {C''_{N}} \leq
        \frac{\tilde{D}'_N}{\mathring{D}^S_S\hspace{-0.3mm}'} + \underbrace{\frac{-\delta_S}{\mathring{D}^S_S\hspace{-0.3mm}'}}_{\leq 0} .
    \end{gather}

    Next, we show that $\frac{\tilde{D}'_S}{\tilde{D}'_N} > \frac{C''_N}{C''_S}$  $\implies$ $\tilde{C}' \geq \mathring{C}^{S}\hspace{-0.3mm}'$ yields a contradiction.
    We start by rearranging $\frac{\tilde{D}'_S}{\tilde{D}'_N} > \frac{C''_N}{C''_S}$ to $\frac{\tilde{D}'_N}{\tilde{D}'_S} < \frac{C''_S}{C''_N}$.
    We then obtain the following contradiction:
    \begin{gather*}
    \begin{aligned}
        \frac{C''_S}{C''_N} > \frac{\tilde{D}'_N}{\tilde{D}'_S} 
        \geq \frac{\tilde{D}'_N}{\mathring{D}^S_S\hspace{-0.3mm}'} 
        \geq \frac{\tilde{D}'_N}{\mathring{D}^S_S\hspace{-0.3mm}'} + \underbrace{\frac{-\delta_S}{\mathring{D}^S_S\hspace{-0.3mm}'}}_{\leq 0} 
        \geq \frac{C''_S}{C''_N},
    \end{aligned}
    \end{gather*}


    where the second and third inequalities follow from $\tilde{D}'_i \geq \mathring{D}_i^S\hspace{-0.3mm}'$ for all $i$,
    and the last inequality follows from the implication of $\tilde{C}' \geq \mathring{C}^{S}\hspace{-0.3mm}'$ documented in Equation~\eqref{E: N>P_S implication v2}.

    We have reached the contradiction $\frac{C''_S}{C''_N} > \frac{C''_S}{C''_N}$. Hence, $\frac{\tilde{D}'_S}{\tilde{D}'_N} > \frac{C''_N}{C''_S}$  $\implies$ $\tilde{C}' \geq \mathring{C}^{S}\hspace{-0.3mm}'$ is incorrect, 
    and we have thus shown that we must have $\frac{\tilde{D}'_S}{\tilde{D}'_N} > \frac{C''_N}{C''_S}\implies \tilde{C}' < \mathring{C}^{S}\hspace{-0.3mm}'$.

    Together, the proofs of the forward and backward directions yield the equivalence
    \begin{gather*}
        \tilde{\tau} < \mathring{\tau}^{S}
        \iff
        \frac{\tilde{D}'_S}{\tilde{D}'_N} > \frac{C''_N}{C''_S}.
    \end{gather*}

\end{proof}


Using Lemmas \ref{lemma: North<U}-\ref{lemma: S>Negishi}, we can now prove Lemma \ref{Lemma: U&N in between preferred carbon price}, which is restated below.

\renewcommand{\thelemma}{\arabic{lemma}}

\setcounter{lemma}{1}
\begin{lemma}
    The utilitarian uniform carbon price ($\check{\tau}$) and the Negishi-weighted carbon price ($\tilde{\tau}$) are in between the preferred uniform carbon prices of the Global North ($\mathring{\tau}^N$) and the Global South ($\mathring{\tau}^S$), unless they all coincide.
\end{lemma}

\setcounter{lemma}{4}
\renewcommand{\thelemma}{A\arabic{lemma}}

\begin{proof}
    Let us begin by showing that the utilitarian uniform carbon price lies between North's and South's preferred uniform carbon prices, unless they coincide. Lemma \ref{lemma: North<U} and \ref{lemma: South>U} imply that $\mathring{\tau}^S > \check{\tau} > \mathring{\tau}^N$ if and only if $\frac{\check{D}'_S}{\check{D}'_N} > \frac{C''_N}{C''_S}$. 
    Similarly, it can be shown that $\mathring{\tau}^S < \check{\tau} < \mathring{\tau}^N$ if and only if $\frac{\check{D}'_S}{\check{D}'_N} < \frac{C''_N}{C''_S}$.
    Hence, we have $\mathring{\tau}^S = \check{\tau} = \mathring{\tau}^N$ if and only if $\frac{\check{D}'_S}{\check{D}'_N} = \frac{C''_N}{C''_S}$, as this is the only remaining possibility for each inequality.

    Analogously, we can show that the Negishi-weighted uniform carbon price lies between North's and South's preferred uniform carbon prices, unless they coincide. Lemma \ref{lemma: N<Negishi} and \ref{lemma: S>Negishi} imply that $\mathring{\tau}^S > \tilde{\tau} > \mathring{\tau}^N$ if and only if $\frac{\tilde{D}'_S}{\tilde{D}'_N} > \frac{C''_N}{C''_S}$. 
     Similarly, it can be shown that $\mathring{\tau}^S < \tilde{\tau} < \mathring{\tau}^N$ if and only if $\frac{\tilde{D}'_S}{\tilde{D}'_N} < \frac{C''_N}{C''_S}$.
    Hence, we have $\mathring{\tau}^S = \tilde{\tau} = \mathring{\tau}^N$ if and only if $\frac{\tilde{D}'_S}{\tilde{D}'_N} = \frac{C''_N}{C''_S}$, as this is the only remaining possibility for each inequality.
    
    It follows that  $\mathring{\tau}^S = \check{\tau} = \tilde{\tau} = \mathring{\tau}^N$ if and only if $\frac{{D}'_S}{{D}'_N} = \frac{C''_N}{C''_S}$, where $D'_i  = \check{D}'_i = \tilde{D}'_i$.
    
\end{proof}

\subsection{Proof of Proposition 3} \label{Assec: Proof of Proposition 3}




\begin{proof}

We split the proof into the forward and backward implications.

    \vspace{6pt}
    \myuline{Proof of forward direction}: $\mathring{\tau}^S > \mathring{\tau}^N \implies \check{\tau} > \tilde{\tau}$.

South's preferred uniform carbon price is greater than North's preferred uniform carbon price if and only if
\begin{equation*}
    - \mathring{D}_{S}^S\hspace{-0.3mm}' \frac{{C}''_S + {C}''_{N}}{{C}''_{N}}
    >
    - \mathring{D}_{N}^N\hspace{-0.3mm}' \frac{{C}''_S + {C}''_{N}}{{C}''_{S}}.
\end{equation*}

Simplifying and rearranging yields
\begin{equation*}
    \frac{\mathring{D}_{S}^S\hspace{-0.3mm}'}{\mathring{D}_{N}^N\hspace{-0.3mm}'}
    >
    \frac{{C}''_{N}}{{C}''_{S}}.
\end{equation*}

From Lemma \ref{Lemma: U&N in between preferred carbon price}, we know that the utilitarian uniform carbon price lies between the two preferred uniform carbon prices. For strictly convex abatement cost functions, we know that if South's preferred uniform carbon price is greater than North's preferred uniform carbon price, then $\mathring{A}_{i}^S > \mathring{A}_{i}^N$ for all $i$, and thus $\mathring{A}^S > \mathring{A}^N$. For strictly convex damage functions, and recalling that marginal damages of abatement are negative, this implies

\begin{equation*}
    \mathring{D}_{i}^S\hspace{-0.3mm}' > \tilde{D}'_i > \mathring{D}_{i}^N\hspace{-0.3mm}', \quad \forall i.
\end{equation*}

We thus have
\begin{equation*}
    \frac{\tilde{D}'_S}{\tilde{D}'_N}
    > \frac{\mathring{D}_{S}^S\hspace{-0.3mm}'}{\mathring{D}_{N}^N\hspace{-0.3mm}'}
    > \frac{{C}''_{N}}{{C}''_{S}}.
\end{equation*}

We have thus shown that 
\begin{equation*}
    \mathring{\tau}^S > \mathring{\tau}^N \implies \check{\tau} > \tilde{\tau}.
\end{equation*}

\vspace{6pt}
\myuline{Proof of backward direction}: $\check{\tau} > \tilde{\tau} \implies \mathring{\tau}^S > \mathring{\tau}^N$.

Proposition \ref{Proposition: utilitarian vs Negishi} establishes that $\check{\tau} > \tilde{\tau}$ if and only if $\frac{\tilde{D}'_S}{\tilde{D}'_N} > \frac{C''_N}{C''_S}$.
From Lemma \ref{Lemma: U&N in between preferred carbon price}, we know that $\frac{\tilde{D}'_S}{\tilde{D}'_N} > \frac{C''_N}{C''_S}$ implies $\mathring{\tau}^S > \check{\tau} > \mathring{\tau}^N$. Therefore,
\begin{equation*}
    \check{\tau} > \tilde{\tau} \implies \mathring{\tau}^S > \mathring{\tau}^N.
\end{equation*}

Together, the proofs of the forward and backward directions yield the equivalence
\begin{gather*}
    \check{\tau} > \tilde{\tau}
    \iff
    \mathring{\tau}^S > \mathring{\tau}^N.
\end{gather*}

\end{proof}

\subsection{Proof of Proposition 4} \label{Assec: Proof of Proposition 4}

\begin{proof}
    \myuline{Part (i)}:

    We begin by writing the Negishi-weighted carbon price as a weighted average of regions' preferred uniform carbon prices:
    \begin{gather*}
    \begin{aligned}
        \underbrace{- \sum_i {D}'_i(\tilde{A})}_{ \tilde{\tau} } 
        &= \tilde{\omega} \underbrace{\left[- {D}'_N(\mathring{A}^N)
            \frac{{C}''_N + {C}''_{S}}
            {{C}''_{S}} \right]}_{ \mathring{\tau}^N}
            + (1-\tilde{\omega}) \underbrace{\left[- {D}'_S(\mathring{A}^S)
            \frac{{C}''_S + {C}''_{N}} 
            {{C}''_{N}}\right]}_{ \mathring{\tau}^S}.
    \end{aligned}
    \end{gather*}

Using the assumption of constant marginal damages, it is easy to see that $\tilde{\omega} = \frac{C_S''}{C_N'' + C_S''}$. Using the assumption that $C_i'' = \frac{\kappa}{W_i}$ for some constant $\kappa > 0$ yields $\tilde{\omega} = \frac{W_N}{\sum_j W_j}$ and $(1-\tilde{\omega})= \frac{W_S}{\sum_j W_j}$. Plugging in $\tilde{\omega}$ yields $\tilde{\tau} = \sum_i \frac{W_i}{\sum_j W_j} \, \mathring{\tau}^i$.

\vspace{6pt}

\myuline{Part (ii)}:

    Similarly, we write the utilitarian uniform carbon price as a weighted average of regions' preferred uniform carbon prices:
\begin{gather*}
\begin{aligned}
    &\underbrace{- \sum_i {u}'(\check{x}_i) {D}'_i(\check{A})
        \frac{{C}''_S + {C}''_N}
        {{u}'(\check{x}_N) {C}''_S + {u}'(\check{x}_S) {C}''_N }}_{\check{\tau} } \\
    &=  \check{\omega} \underbrace{\left[- {D}'_N(\mathring{A}^N)
        \frac{{C}''_N + {C}''_{S}}
        {{C}''_{S}} \right]}_{\mathring{\tau}^N}
        + (1-\check{\omega}) \underbrace{\left[- {D}'_S(\mathring{A}^S)
        \frac{{C}''_S + {C}''_{N}} 
        {{C}''_{N}}\right]}_{\mathring{\tau}^S}.
\end{aligned}
\end{gather*}
Using the assumption of constant marginal damages, we can see that $\check{\omega} = \frac{u'(\check{x}_N)C_S''}{u'(\check{x}_S) C_N'' + u'(\check{x}_N) C_S''}$. Using the assumption that  $C_i'' = \frac{\kappa}{W_i}$ for some constant $\kappa > 0$, and logarithmic utility, we have 
$\check{\omega} = \frac{\check{x}_S L_N w_N}{\check{x}_N L_S w_S + \check{x}_S L_N w_N}$.
Finally, using the approximation $\frac{\check{x}_S}{\check{x}_N} \approx \frac{w_S}{w_N}$ yields 
$\check{\omega} \approx \frac{L_N}{\sum_j L_j}$ and $(1 - \check{\omega}) \approx \frac{L_S}{\sum_j L_j}$.
Plugging in $\check{\omega}$ yields $\check{\tau} = \sum_i \frac{L_i}{\sum_j L_j} \, \mathring{\tau}^i$.

\end{proof}

\subsection{Generalizations of Proposition 4} \label{Assec: Generalizations of Proposition 4}

This section considers generalizations of Proposition \ref{Prop: voting}, relaxing one assumption at a time, while maintaining all other assumptions. The derivations are analogous to the ones in Appendix \ref{Assec: Proof of Proposition 4}.

\vspace{8pt}
\noindent
{\textbf{Abatament cost convexity}}

\noindent
If $C_i''=\kappa_i/W_i$ with\textit{ region-specific} $\kappa_i$, the weighted-average representations are adjusted for abatement responsiveness: (i) $\tilde{\tau} = \sum_i \frac{W_i/ \kappa_i}{\sum_j W_j / \kappa_j} \, \mathring{\tau}^i$, and (ii) $\check{\tau} \approx \sum_i \frac{L_i/\kappa_i}{\sum_j L_j \kappa_j} \, \mathring{\tau}^i$.

A steeper marginal abatement cost function (higher $k_i$) reduces the weight on the preferred uniform carbon price of region $i$. The intuition is that a relatively higher $k_i$ raises the preferred carbon price of region $i$. This effect is adjusted for in determining the optimal uniform carbon prices under welfare weights that are strictly between the ``edge weights'', which give full weight to one region.

\vspace{8pt}
\noindent
{\textbf{Isoelastic utility}}

\noindent
If $u(x_i) = \frac{x_i^{1-\eta}}{1-\eta}$ (for $\eta=1$,  $u(x_i) = \log(x_i)$), the weighted-average representation of the utilitarian uniform carbon price is adjusted for the concavity of the utility function\footnote{The weighted-average representation of the Negishi-weighted carbon price remains $\tilde{\tau} = \sum_i \frac{W_i}{\sum_j W_j} \, \mathring{\tau}^i$.}: 
\begin{equation}
\label{AEq: utilitarian voting weights isoelastic utility}
    \check{\tau} = \sum_i \frac{L_i w_i \check{x}_i^{-\eta}}{\sum_j L_j w_j \check{x}_j^{-\eta}} \, \mathring{\tau}^i.
\end{equation}
Using the approximation $\frac{\check{x}_S}{\check{x}_N} \approx \frac{w_S}{w_N}$ yields $\check{\tau} \approx \sum_i \frac{L_i \check{x}_i^{1 -\eta}}{\sum_j L_j \check{x}_j^{1 -\eta}} \, \mathring{\tau}^i$.

The weight on the preferred uniform carbon price of the South (which is assumed to be poorer), increases as $\eta$ increases (reflecting greater inequality aversion).  If we further assume that the approximation holds with equality, $\frac{\check{x}_S}{\check{x}_N} = \frac{w_S}{w_N}$, we obtain the following crisp intuitions: the weight on the preferred uniform carbon price of the South, $(1-\check{\omega})$, exceeds the South's population share, $\frac{L_S}{\sum_j L_j}$, for $\eta > 1$, equals it for $\eta = 1$, and is below it for  $\eta < 1$. Finally, for linear utility ($\eta = 0$), the Negishi-weighted and utilitarian carbon prices coincide, and the weights collapse to endowment share: $\check{\tau} = \tilde{\tau} = \sum_i \frac{W_i}{\sum_j W_j} \, \mathring{\tau}^i$.

\vspace{8pt}
\noindent
{\textbf{Endowment and consumption ratios}}

\noindent
Using Equation \eqref{AEq: utilitarian voting weights isoelastic utility}, and log utility, we can make the following observations about the weighted-average representation of the utilitarian uniform carbon price\footnote{As before, the weighted-average representation of the Negishi-weighted carbon price remains unaffected.}. The weight on the preferred uniform carbon price of the South, $(1-\check{\omega})$, exceeds the South's population share, $\frac{L_S}{\sum_j L_j}$, for $\frac{\check{x}_S}{\check{x}_N} < \frac{w_S}{w_N}$, equals it for $\frac{\check{x}_S}{\check{x}_N} = \frac{w_S}{w_N}$, and is below it for  $\frac{\check{x}_S}{\check{x}_N} > \frac{w_S}{w_N}$.

\subsection{Proof of Proposition 5} \label{Assec: Proof of Proposition 5}

 \begin{proof}

 \myuline{Proof of forward implication}: $\tilde{\tau} < \check{\tau}$  $\implies$ 
    $
    \frac{\check{u}'_{N2} \check{D}'_{N2} 
        + \check{u}'_{S2} \check{D}'_{S2}}
        {\check{D}'_{N2} + \check{D}'_{S2}}
        > 
        \nu 
        \frac{{\check{u}'_{N1} C''_{S1} + \check{u}'_{S1} C''_{N1}} }{C''_{S1} + C''_{N1}}$.
    
    We ask under which conditions $\tilde{\tau} < \check{\tau}$, or equivalently, $\tilde{C}'_{1} < \check{C}'_{1}$. 
  For strictly convex abatement cost functions, $\tilde{C}'_{1} < \check{C}'_{1}$ implies $\tilde{C}_{i1} < \check{C}_{i1}$ and $\tilde{A}_{i1} < \check{A}_{i1}$ for all $i$, and thus $\tilde{A}_1 < \check{A}_1$.
    This implies $\tilde{D}_{i2} > \check{D}_{i2}$ and $\tilde{D}'_{i2} < \check{D}'_{i2}$ for all $i$ (note that marginal damages of abatement are negative).
    Using the budget constraints, this implies $\check{X}_{i1} < \tilde{X}_{i1}$ and $\check{X}_{i2} > \tilde{X}_{i2}$ for all $i$.
    For an exogenous population, we thus have $\check{x}_{i1} < \tilde{x}_{i1}$ and $\check{x}_{i2} > \tilde{x}_{i2}$ for all $i$.
    Therefore, $\check{u}'_{i1} > \tilde{u}'_{i1}$ and $\check{u}'_{i2} < \tilde{x}'_{i2}$ for all $i$.

    We have $\tilde{C}'_{1} < \check{C}'_{1}$ if and only if
    \begin{gather*}
        - \beta \nu \left(\tilde{D}'_{N2} + \tilde{D}'_{S2}\right) 
        < 
        - \beta \left( \check{u}'_{N2} \check{D}'_{N2} 
        + \check{u}'_{S2} \check{D}'_{S2}\right)
        \frac{C''_{S1} + C''_{N1}}
        {\check{u}'_{N1} C''_{S1} + \check{u}'_{S1} C''_{N1}} .
    \end{gather*}

    Cancelling $\beta$, multiplying both sides by the denominator on the RHS (which is positive), and rearranging, we get
    \begin{gather*}
        \left( \check{u}'_{N2} \check{D}'_{N2} 
        + \check{u}'_{S2} \check{D}'_{S2}\right)
        \left(C''_{S1} + C''_{N1}\right)
        < 
        \nu \left(\tilde{D}'_{N2} + \tilde{D}'_{S2}\right) 
        \left( {\check{u}'_{N1} C''_{S1} + \check{u}'_{S1} C''_{N1}} \right)
        .
    \end{gather*}

Since $\tilde{D}'_{i2} < \check{D}'_{i2}$ for all $i$, the previous inequality implies
\begin{gather*}
        \left( \check{u}'_{N2} \check{D}'_{N2} 
        + \check{u}'_{S2} \check{D}'_{S2}\right)
        \left(C''_{S1} + C''_{N1}\right)
        < 
        \nu \left(\check{D}'_{N2} + \check{D}'_{S2}\right) 
        \left( {\check{u}'_{N1} C''_{S1} + \check{u}'_{S1} C''_{N1}} \right)
        .
    \end{gather*}

Noting that $(\check{D}'_{N2} + \check{D}'_{S2})$ is negative so the sign flips upon division, rearranging yields
\begin{gather*}
        \frac{\check{u}'_{N2} \check{D}'_{N2} 
        + \check{u}'_{S2} \check{D}'_{S2}}
        {\check{D}'_{N2} + \check{D}'_{S2}}
        > 
        \nu 
        \frac{{\check{u}'_{N1} C''_{S1} + \check{u}'_{S1} C''_{N1}} }{C''_{S1} + C''_{N1}}.
    \end{gather*}

 \vspace{6pt}

  \myuline{Proof of backward implication}:   $    
        \frac{\check{u}'_{N2} \check{D}'_{N2} 
        + \check{u}'_{S2} \check{D}'_{S2}}
        {\check{D}'_{N2} + \check{D}'_{S2}}
        > 
        \nu 
        \frac{{\check{u}'_{N1} C''_{S1} + \check{u}'_{S1} C''_{N1}} }{C''_{S1} + C''_{N1}}
        \implies
        \tilde{\tau} < \check{\tau}$.

We will prove the contrapositive of the stated result. That is, we will prove 
\begin{gather*}
        \tilde{\tau} \geq \check{\tau}
        \implies
    \frac{\check{u}'_{N2} \check{D}'_{N2} 
        + \check{u}'_{S2} \check{D}'_{S2}}
        {\check{D}'_{N2} + \check{D}'_{S2}}
        \leq 
        \nu 
        \frac{{\check{u}'_{N1} C''_{S1} + \check{u}'_{S1} C''_{N1}} }{C''_{S1} + C''_{N1}}.
\end{gather*}

    We start by establishing the implications of $\tilde{\tau} \geq \check{\tau}$, or equivalently, $\tilde{C}'_1 \geq \check{C}'_1$.
    
 For strictly convex abatement cost functions, $\tilde{C}'_{1} \geq \check{C}'_{1}$ implies $\tilde{C}_{i1} \geq \check{C}_{i1}$ and $\tilde{A}_{i1} \geq \check{A}_{i1}$ for all $i$, and thus $\tilde{A}_1 \geq \check{A}_1$.
   This implies $\tilde{D}_{i2} \leq \check{D}_{i2}$ and $\tilde{D}'_{i2} \geq \check{D}'_{i2}$ for all $i$ (note that marginal damages of abatement are negative).
    Using the budget constraints, this implies $\check{X}_{i1} \geq \tilde{X}_{i1}$ and $\check{X}_{i2} \leq \tilde{X}_{i2}$ for all $i$.
    For an exogenous population, we thus have $\check{x}_{i1} \geq \tilde{x}_{i1}$ and $\check{x}_{i2} \leq \tilde{x}_{i2}$ for all $i$.
    Therefore, $\check{u}'_{i1} \leq \tilde{u}'_{i1}$ and $\check{u}'_{i2} \geq \tilde{x}'_{i2}$ for all $i$.

    We have $\tilde{C}'_{1} \geq \check{C}'_{1}$ if and only if
    \begin{gather*}
        - \beta \nu \left(\tilde{D}'_{N2} + \tilde{D}'_{S2}\right) 
        \geq 
        - \beta \left( \check{u}'_{N2} \check{D}'_{N2} 
        + \check{u}'_{S2} \check{D}'_{S2}\right)
        \frac{C''_{S1} + C''_{N1}}
        {\check{u}'_{N1} C''_{S1} + \check{u}'_{S1} C''_{N1}} .
    \end{gather*}

    Cancelling $\beta$, multiplying both sides by the denominator on the RHS (which is positive), and rearranging, we get
    \begin{gather*}
        \left( \check{u}'_{N2} \check{D}'_{N2} 
        + \check{u}'_{S2} \check{D}'_{S2}\right)
        \left(C''_{S1} + C''_{N1}\right)
        \geq 
        \nu \left(\tilde{D}'_{N2} + \tilde{D}'_{S2}\right) 
        \left( {\check{u}'_{N1} C''_{S1} + \check{u}'_{S1} C''_{N1}} \right)
        .
    \end{gather*}
    
Since $\tilde{D}'_{i2} \geq \check{D}'_{i2}$ for all $i$, the previous inequality implies
\begin{gather*}
        \left( \check{u}'_{N2} \check{D}'_{N2} 
        + \check{u}'_{S2} \check{D}'_{S2}\right)
        \left(C''_{S1} + C''_{N1}\right)
        \geq 
        \nu \left(\check{D}'_{N2} + \check{D}'_{S2}\right) 
        \left( {\check{u}'_{N1} C''_{S1} + \check{u}'_{S1} C''_{N1}} \right)
        .
    \end{gather*}

Noting that $(\check{D}'_{N2} + \check{D}'_{S2})$ is negative so the sign flips upon division, rearranging yields
\begin{gather*}
        \frac{\check{u}'_{N2} \check{D}'_{N2} 
        + \check{u}'_{S2} \check{D}'_{S2}}
        {\check{D}'_{N2} + \check{D}'_{S2}}
        \leq 
        \nu 
        \frac{{\check{u}'_{N1} C''_{S1} + \check{u}'_{S1} C''_{N1}} }{C''_{S1} + C''_{N1}}.
    \end{gather*}

We have thus shown that
\begin{gather*}
    \tilde{\tau} \geq \check{\tau}
    \implies
        \frac{\check{u}'_{N2} \check{D}'_{N2} 
        + \check{u}'_{S2} \check{D}'_{S2}}
        {\check{D}'_{N2} + \check{D}'_{S2}}
        \leq 
        \nu 
        \frac{{\check{u}'_{N1} C''_{S1} + \check{u}'_{S1} C''_{N1}} }{C''_{S1} + C''_{N1}}.
\end{gather*}

By contraposition, we therefore have
\begin{gather*}
        \frac{\check{u}'_{N2} \check{D}'_{N2} 
        + \check{u}'_{S2} \check{D}'_{S2}}
        {\check{D}'_{N2} + \check{D}'_{S2}}
        > 
        \nu 
        \frac{{\check{u}'_{N1} C''_{S1} + \check{u}'_{S1} C''_{N1}} }           {C''_{S1} + C''_{N1}}
        \implies
        \tilde{\tau} < \check{\tau}.
\end{gather*}

Together, the proofs of the forward and backward implications yield the equivalence
   \begin{gather*}
    \tilde{\tau} < \check{\tau}
    \iff
        \frac{\check{u}'_{N2} \check{D}'_{N2} 
        + \check{u}'_{S2} \check{D}'_{S2}}
        {\check{D}'_{N2} + \check{D}'_{S2}}
        > 
        \nu 
        \frac{{\check{u}'_{N1} C''_{S1} + \check{u}'_{S1} C''_{N1}} }{C''_{S1} + C''_{N1}}.
\end{gather*}

\end{proof}

\subsection{Calculation of welfare-equivalent consumption changes} \label{Assec: Calculation of welfare-equivalent consumption changes}

    The aim is to calculate the consumption changes in the initial period (2005), $\Delta X_{i0}$ (where $t=0$ corresponds to the year 2005), that would yield a welfare change (in utility terms) that is equivalent to the intertemporal welfare difference between each of the utilitarian solutions and the Negishi solution. I begin by computing the net present value (NPV) of the utilitarian welfare changes across two solutions for each region (the numerator in Equation~\eqref{E:WECC_step1})
    \footnote{
    I use this approach, rather than calculating the NPV by discounting the consumption changes with fixed discount rates, to account for the fact that the social discount rates (SDR) are different across regions and change over time due to different economic growth rates. To see this, note that the SDR is approximated by the Ramsey Rule, $SDR \approx \rho + \eta g$, where $g$ is the growth rate in per capita consumption, which differs across regions and over time.
    }, 
    and divide that by the population size in 2005 to obtain the required per capita welfare change in 2005. I then set the NPV of the per capita welfare change equal to a counterfactual per capita welfare change in the initial period:
\begin{equation}\label{E:WECC_step1}
        \frac{\sum_t{L_{it}\beta^t u(x_{it}^{Util})} - \sum_t{L_{it}\beta^t u(x_{it}^{Neg})}}{L_{i0}} = u(x_{i0}^{cf}) - u(x_{i0}^{Neg}),
    \end{equation}
    
    where  $\beta^t$ is the utility discount factor ($\beta^t = (1+\rho)^{-t}$, where $\rho$ is the utility discount rate), and the superscripts on $x_{it}$ indicate whether this is the per capita consumption of one of the two utilitarian solutions ($Util$), the Negishi solution ($Neg$), or a counterfactual ($cf$) consumption which we compute. The remaining notation is the same as in the main text.
    
    Using the isoelastic specification of the utility function in the RICE model, $u(x_{it}) = \frac{x_{it}^{1-\eta}}{1-\eta}$ (for $\eta=1$,  $u(x_{it}) = \log(x_{it})$, where $\eta$ is the elasticity of the marginal utility of consumption), we can solve for the counterfactual per capita consumption in the initial period:
    \begin{equation*}
        x_{i0}^{cf} = \left[ (1-\eta)\frac{\sum_t{L_{it}\beta^t u(x_{it}^{Util})} - \sum_t{L_{it} \beta^t u(x_{it}^{Neg})}}{L_{i0}} + (x_{i0}^{})^{1-\eta} \right]^{\frac{1}{1-\eta}}.
    \end{equation*}
    
    Finally, the aggregate welfare-equivalent consumption change is calculated as
    \begin{equation*}
        \Delta X_{i0} = L_{i0} \left( x_{i0}^{cf} - x_{i0}^{Neg} \right).
    \end{equation*}


    In addition, I express the utilitarian welfare changes as the welfare-equivalent consumption change in 2005 if consumption were distributed equally (specifically, the ``Global'' values in Figures \ref{fig:C_pctchange_NPV} and \ref{fig:C_pctchange_NPV_rho01} and all values in Figures \ref{fig:C_bar_change_NPV} and \ref{fig:C_bar_change_NPV_rho01}). Let $\bar{x}_0$ be the world average per capita consumption in 2005; that is
    $\bar{x}_0 = \frac{\sum _i L_{i0} x_{i0}}{\sum_i L_{i0}}$.

    I then proceed as above to calculate the counterfactual per capita consumption in the initial period for the world average consumer:
    \begin{equation*}
        \bar{x}_{i0}^{cf} = \left[ (1-\eta)\frac{\Delta PV(U)}{\sum_i L_{i0}} + (\bar{x}_{i0}^{})^{1-\eta} \right]^{\frac{1}{1-\eta}},
    \end{equation*}
    where $\Delta PV(U) = \sum_t{L_{it}\beta^t u(x_{it}^{Util})} - \sum_t{L_{it}\beta^t u(x_{it}^{Neg})}$ for the regional values and $\Delta PV(U) = \sum_t \sum_i{L_{it}\beta^t u(x_{it}^{Util})} - \sum_t \sum_i{L_{it}\beta^t u(x_{it}^{Neg})}$ for the global values.

    Finally, the aggregate welfare-equivalent consumption change in 2005, if consumption were distributed equally, is calculated as follows: 
    \begin{equation*}
            \Delta \bar{X}_{0} = \sum_i L_{i0} \left( \bar{x}_{0}^{cf} - \bar{x}_{0}^{Neg} \right).
        \end{equation*}

\clearpage

\stepcounter{section}
\section*{Appendix \thesection: Additional Tables and Figures} \label{sec:appendixb}
\addcontentsline{toc}{section}{Appendix \thesection: Additional Tables and Figures} 

\subsection{Additional tables}

\renewcommand{\arraystretch}{1.2}
\begin{table}[!htb]
\centering
\caption{Cumulative global industrial CO$_2$ emissions (GtCO$_2$) depending on the optimization problem and the utility discount rate ($\rho$).}
\begin{tabular}{l c c c}
\hline \hline
\multirow{2}{*}{} & \multicolumn{3}{c}{Optimization problem} \\ 
\cline{2-4}
\noalign{\vskip 4pt}
\parbox{2.5cm}{Utility \\ discount rate} & \parbox{2.5cm}{\centering Negishi SWF} & \parbox{4.79cm}{\centering Utilitarian SWF:\\Uniform carbon price}
 & \parbox{5cm}{\centering Utilitarian SWF:\\Differentiated carbon price} \\
\noalign{\vskip 3pt}
\hline
$\rho =1.5\%$ & 3,815 & 3,629 & 3,032 \\
$\rho =0.1\%$ & 1,373 & 1,199 & 1,005 \\
\hline \hline
\end{tabular}
\label{T: Cumulative global industrial CO2 emissions}
\end{table}

\subsection{Additional figures}

\todo[inline]{Make figures nicer and consistent. Remove headings. Change figure numbering to B.x.}

\begin{figure}[!htb]
    \centering
    \includegraphics[width=1\linewidth]{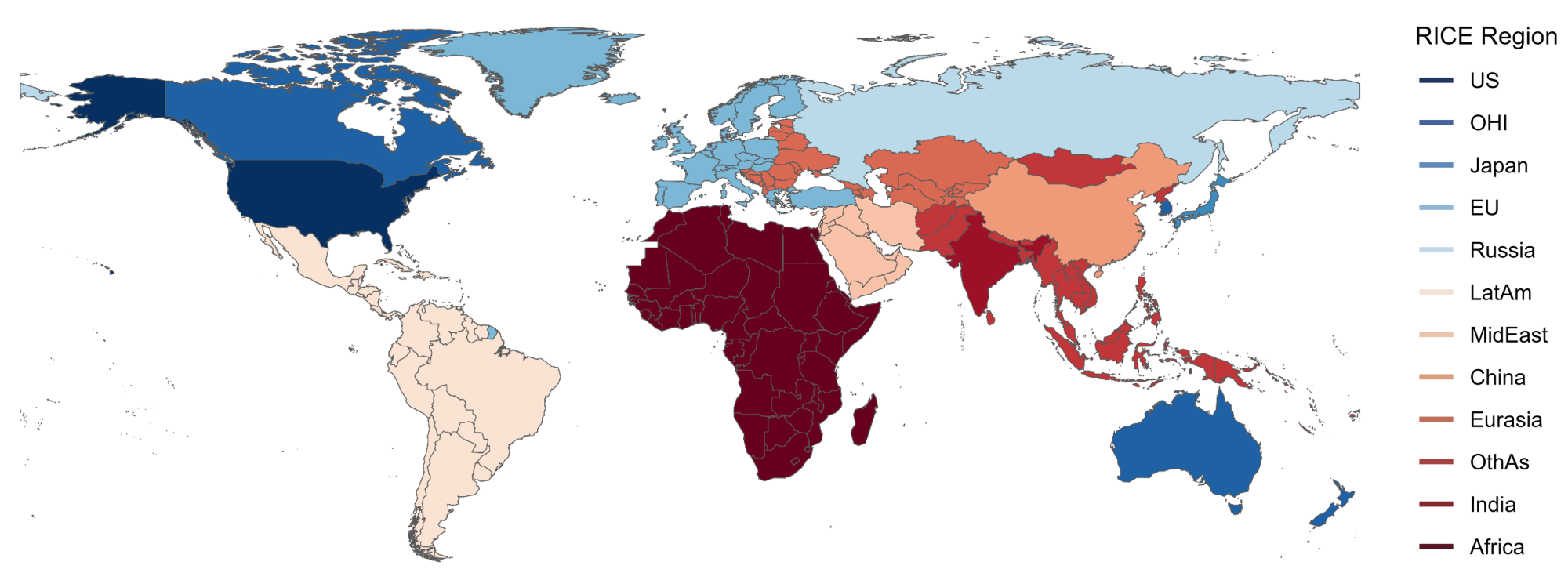}
    \caption{Regions of the RICE model.}
    \vspace{1mm}  
    \begin{minipage}{1\textwidth}
        \small \textit{Notes}:
        Countries of the same color belong to the same region (OHI = Other High Income countries, OthAs = Other Asia). Regions are arranged on the color scale from rich (blue) to poor (red).
    \end{minipage}
    \label{fig:RICE regions}
\end{figure}

\begin{figure}[!htb]
\centering \includegraphics[width=1\textwidth]{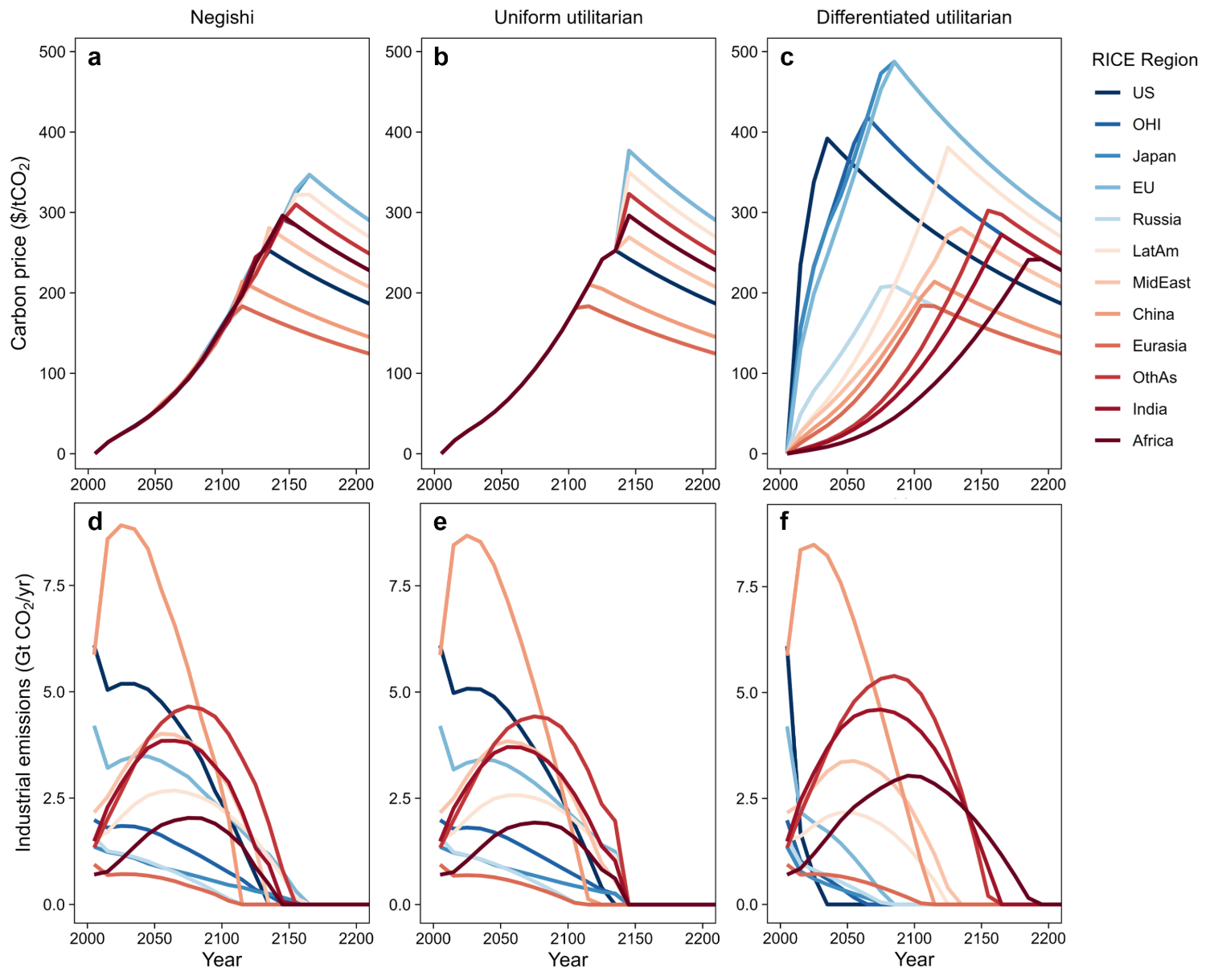}
\caption{Optimal trajectories for carbon prices and industrial emissions conditional on the optimization problem.}
\vspace{1mm}  
    \begin{minipage}{1\textwidth}
        \small \textit{Notes}:
        Results are for the utility discount rate of 1.5\%. Panels (a)-(c) show the optimal carbon price trajectories under the Negishi solution (a) and the utilitarian solution with (b) and without (c) the additional constraint of equalized carbon prices. Panels (d)-(f) show the corresponding industrial emission trajectories. Note that the carbon price decreases once it reaches the region-specific backstop price. Also note that Mimi-RICE-plus only yields an approximately equalized carbon price for the Negishi solution.
    \end{minipage}
\label{F: Optimal trajectories for the carbon price and industrial emissions}
\end{figure}

\begin{figure}[!htb]
    \centering \includegraphics[width=1\textwidth]{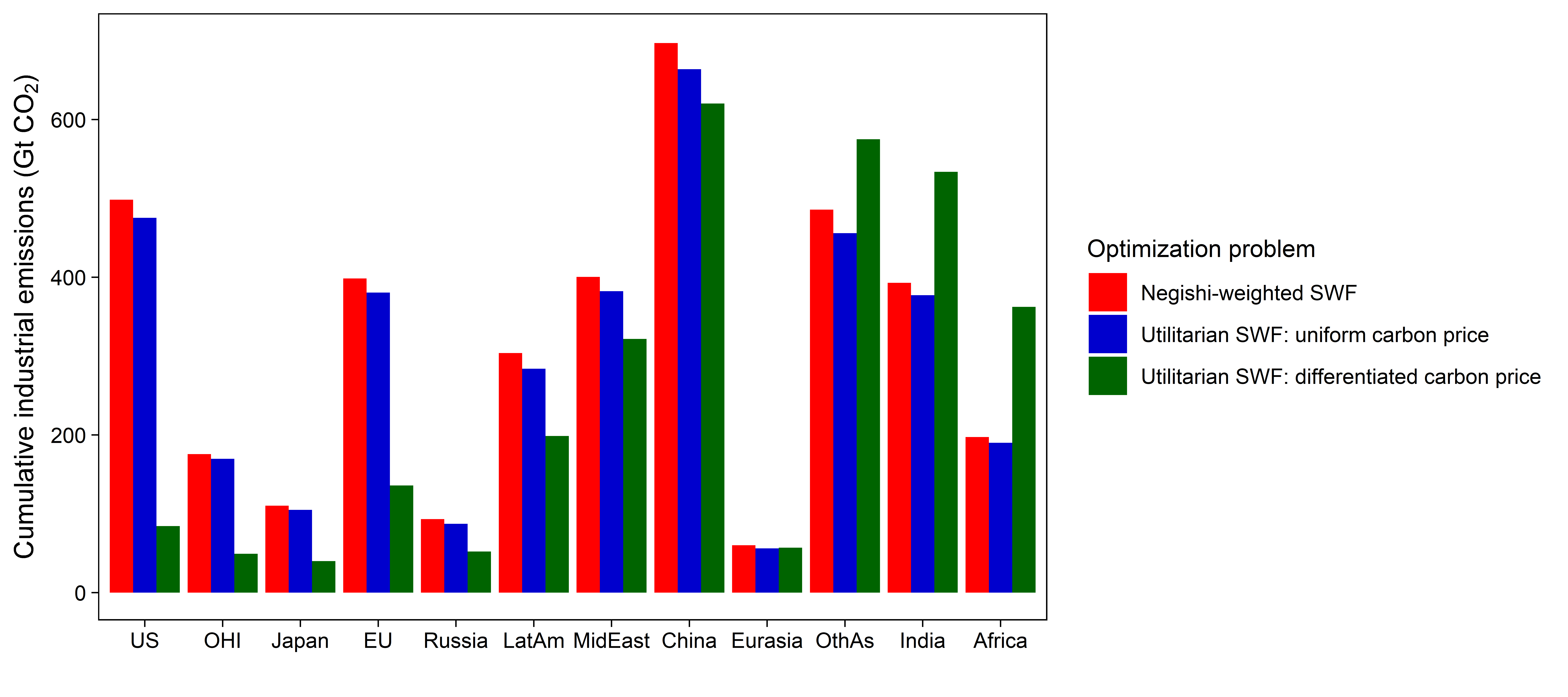}
    \caption{Optimal cumulative industrial emissions depending on the optimization problem.}
    \vspace{1mm}  
    \begin{minipage}{1\textwidth}
        \small \textit{Notes}:
        The figure shows the results for the utility discount rate of 1.5\%.
    \end{minipage}
    \label{EIND_cum_ro_NWvsUtil}
\end{figure}

\begin{figure}[!htb]
    \centering
    \includegraphics[width=1\linewidth]{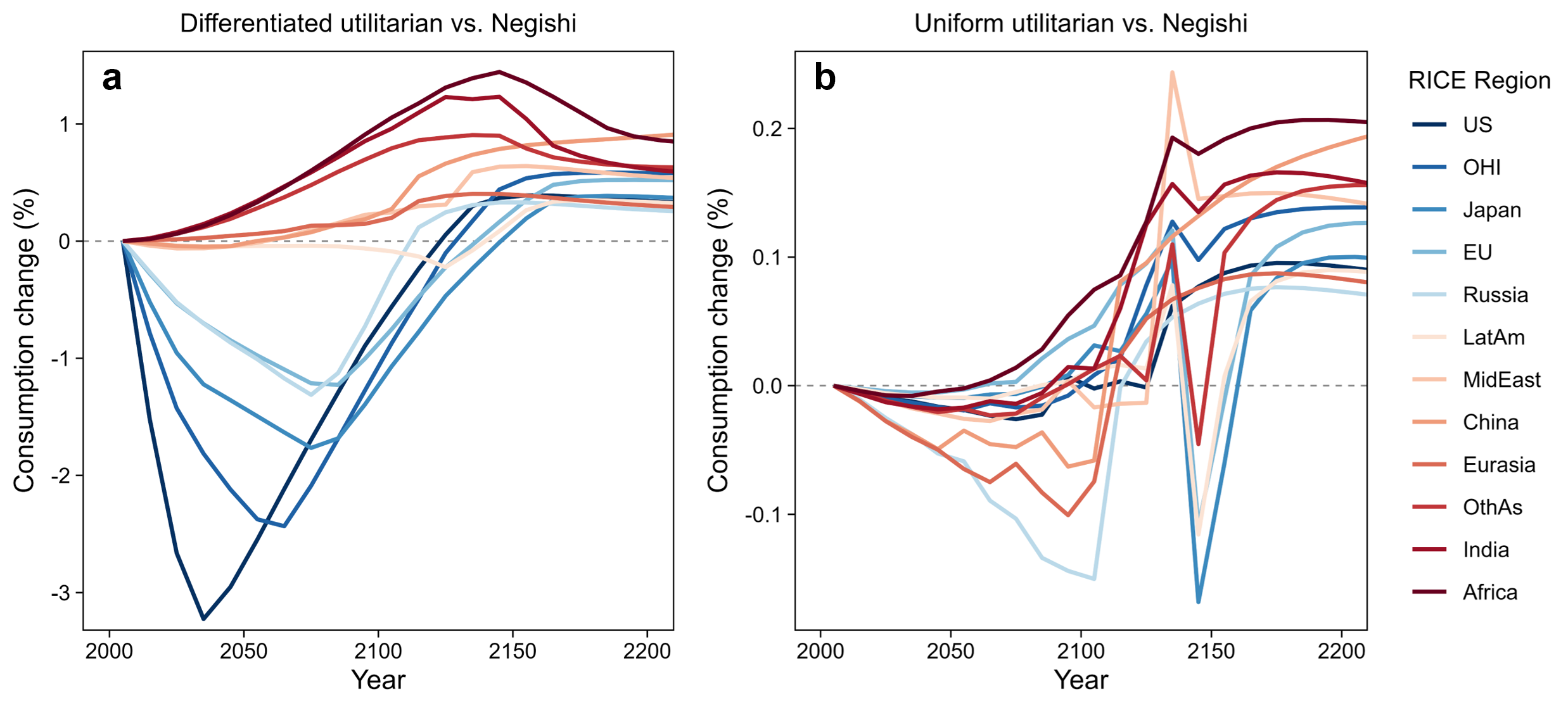}
    \caption{Relative regional consumption changes between the Negishi solution and the utilitarian solutions.}
    \vspace{1mm}  
    \begin{minipage}{1\textwidth}
        \small \textit{Notes}:
        Consumption changes are percentage changes relative to the consumption level in the Negishi solution. Positive values indicate a higher consumption level in the utilitarian solutions. The figure shows the results for the utility discount rate of 1.5\%.
    \end{minipage}
    \label{fig:C_pctchange}
\end{figure}

\begin{figure}[!htb]
    \centering
    \includegraphics[width=1\linewidth]{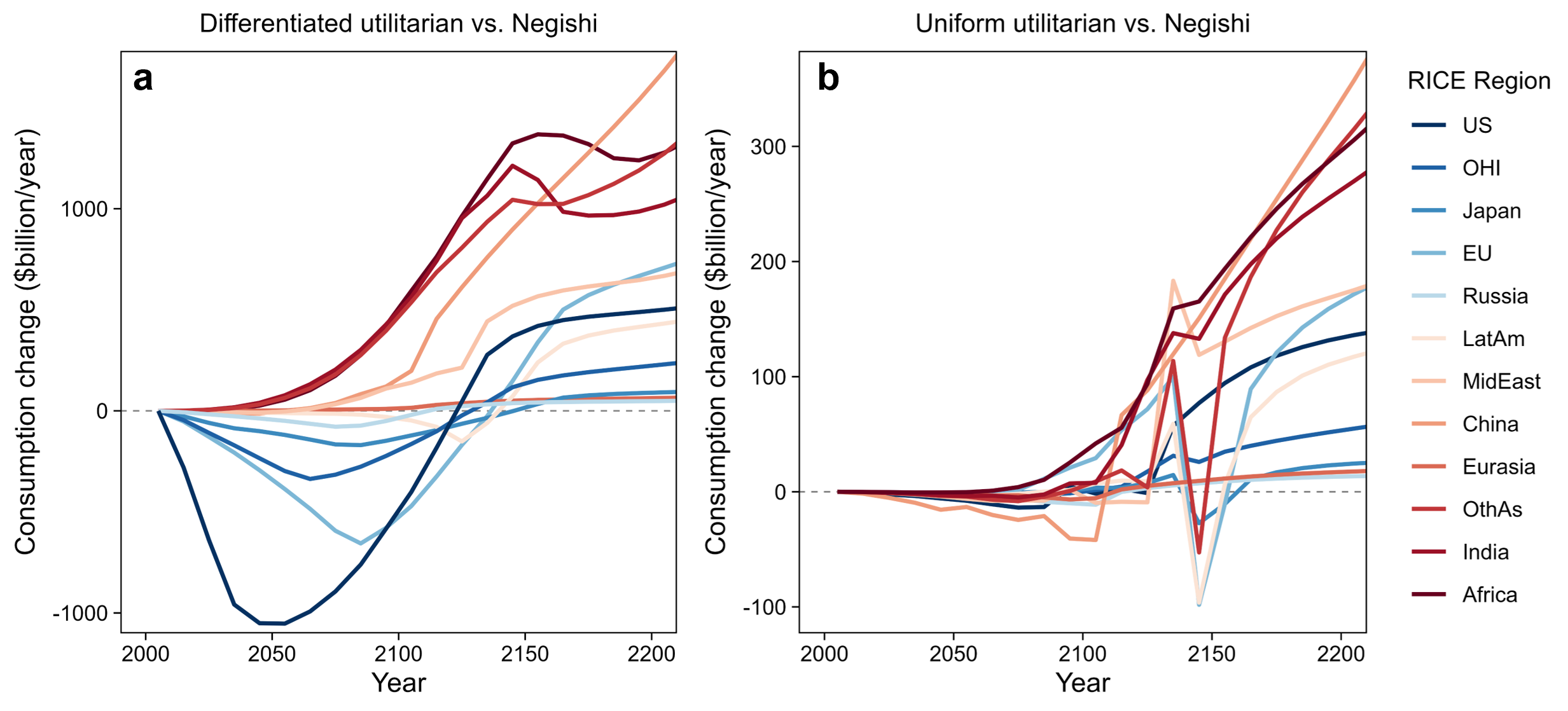}
    \caption{Regional consumption changes between the Negishi solution and the utilitarian solutions.}
    \vspace{1mm}  
    \begin{minipage}{1\textwidth}
        \small \textit{Notes}:
        Positive values indicate a higher consumption level in the utilitarian solutions. The figure shows the results for the utility discount rate of 1.5\%.
    \end{minipage}
    \label{fig:C_change}
\end{figure}

\begin{figure}[!htb]
    \centering
    \includegraphics[width=1\linewidth]{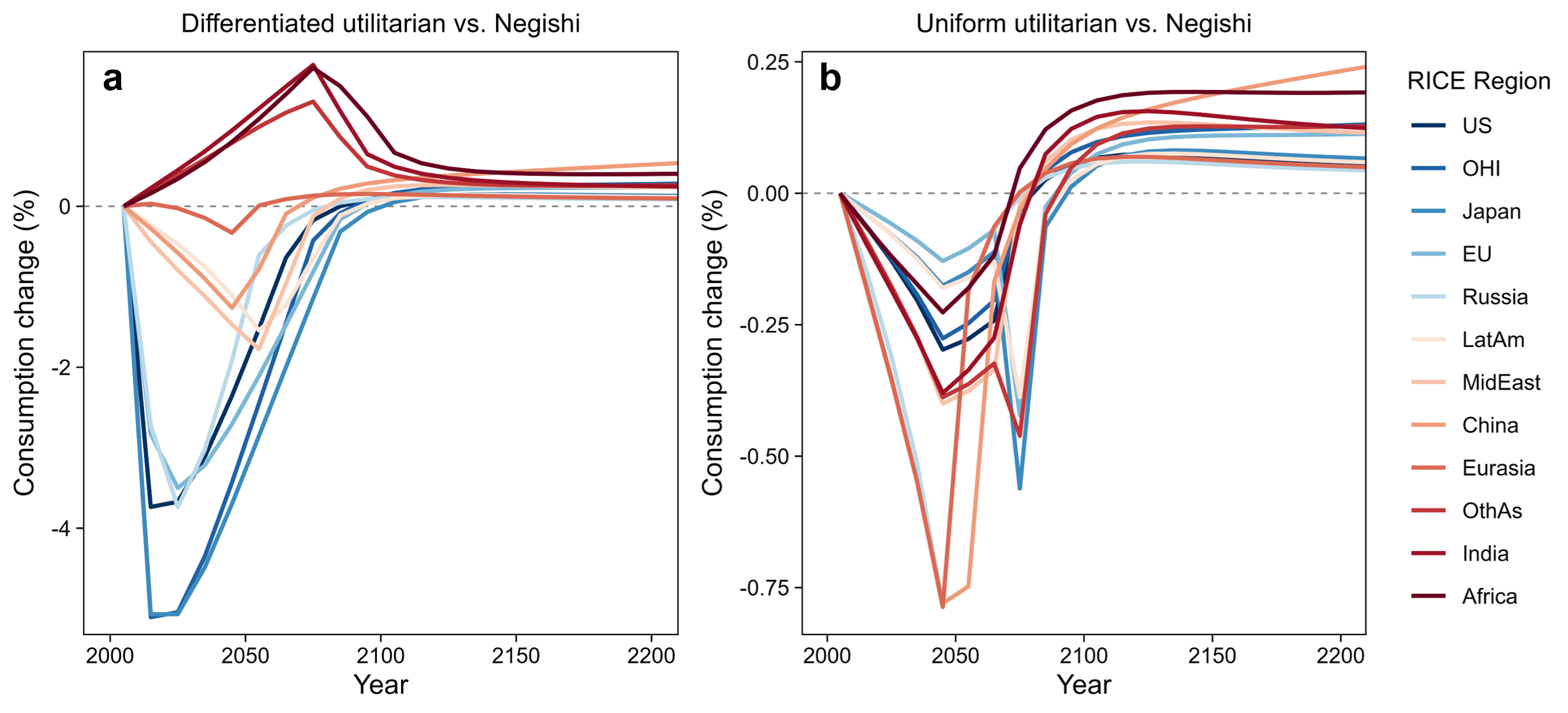}
    \caption{Relative regional consumption changes between the Negishi solution and the utilitarian solutions.}
    \vspace{1mm}  
    \begin{minipage}{1\textwidth}
        \small \textit{Notes}:
        Consumption changes are percentage changes relative to the consumption level in the Negishi solution. Positive values indicate a higher consumption level in the utilitarian solutions. The figure shows the results for the utility discount rate of 0.1\%.
    \end{minipage}
    \label{fig:C_pctchange}
\end{figure}

\begin{figure}[!htb]
    \centering
    \includegraphics[width=1\linewidth]{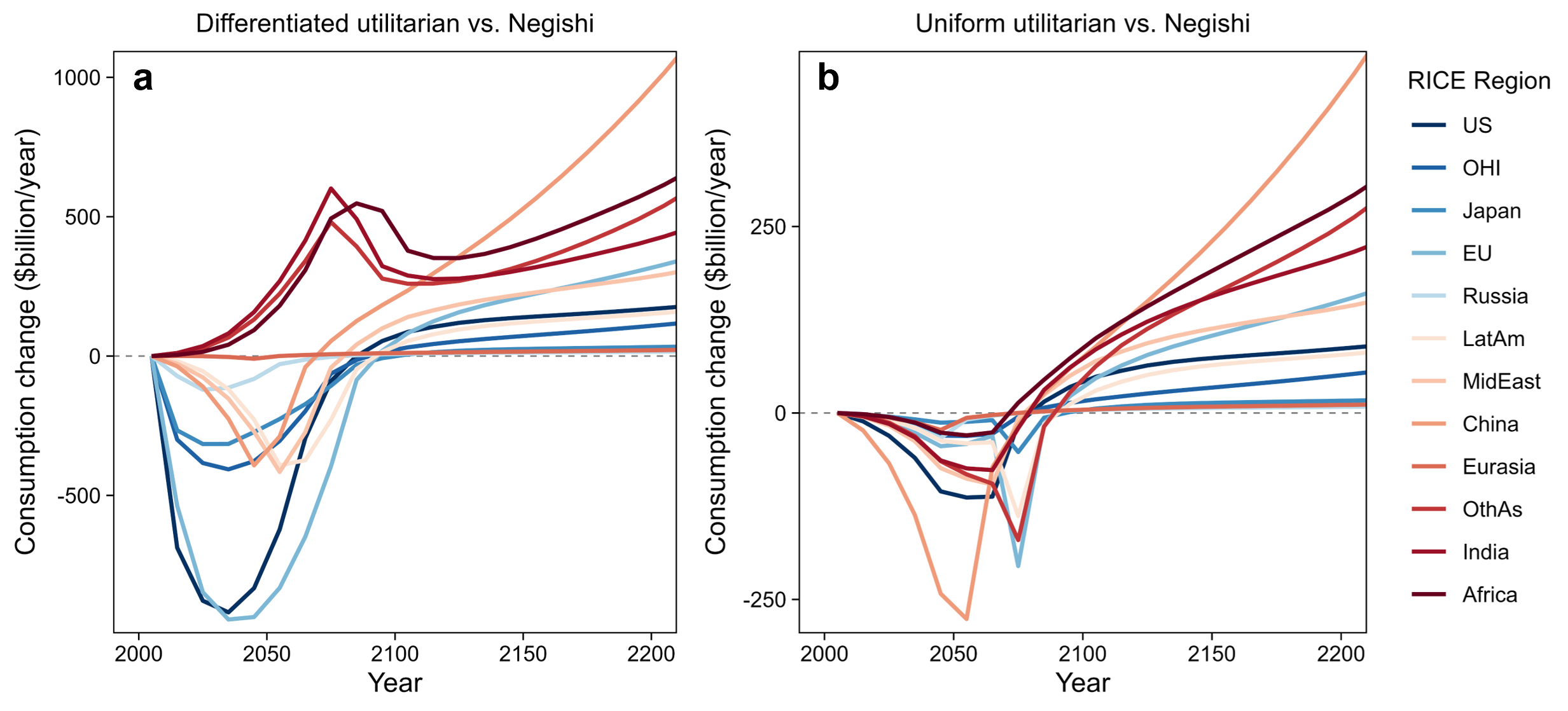}
    \caption{Regional consumption changes between the Negishi solution and the utilitarian solutions.}
    \vspace{1mm}  
    \begin{minipage}{1\textwidth}
        \small \textit{Notes}:
        Positive values indicate a higher consumption level in the utilitarian solutions. The figure shows the results for the utility discount rate of 0.1\%.
    \end{minipage}
    \label{fig:C_change}
\end{figure}

\begin{figure}[!htb]
    \centering
    \includegraphics[width=0.75\linewidth]{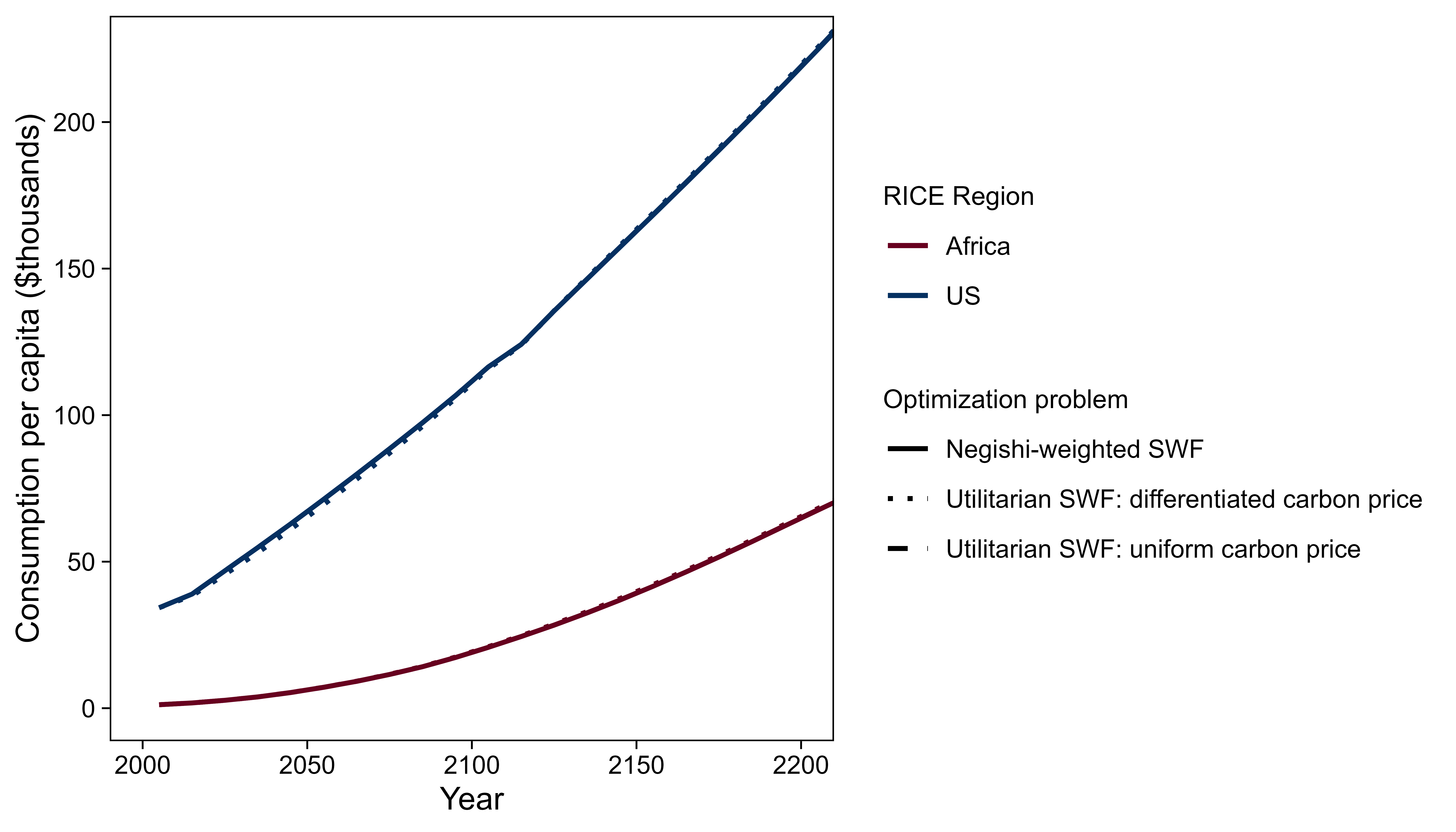}
    \caption{Consumption per capita trajectories for Africa and the US.}
    \vspace{1mm}  
    \begin{minipage}{1\textwidth}
        \small \textit{Notes}:
        The figure shows the results for the utility discount rate of 1.5\%.
    \end{minipage}
    \label{fig: Consumption per capita trajectories}
\end{figure}

\begin{figure}[!htb]
    \centering
    \includegraphics[width=0.9\linewidth]{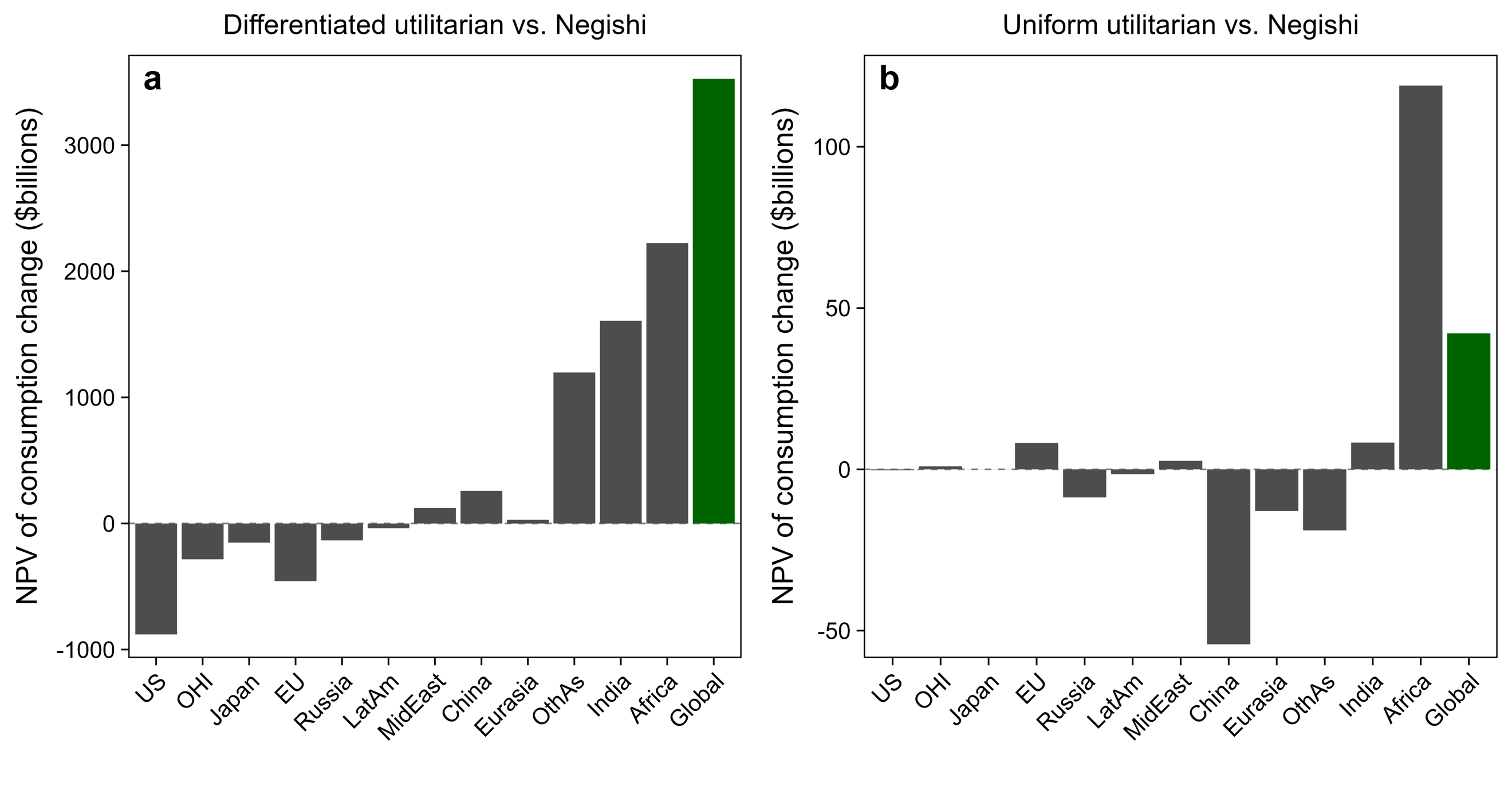}
    \caption{Utilitarian welfare changes.}
    \vspace{1mm}  
    \begin{minipage}{1\textwidth}
        \small \textit{Notes}:
        The values express the regional and global utilitarian welfare change in the welfare-equivalent consumption change in 2005 if consumption were distributed equally (for details, see Appendix \ref{Assec: Calculation of welfare-equivalent consumption changes}). The figure shows the results for the utility discount rate of 1.5\%.
    \end{minipage}
    \label{fig:C_bar_change_NPV}
\end{figure}
\begin{figure}[!htb]
    \centering
    \includegraphics[width=0.9\linewidth]{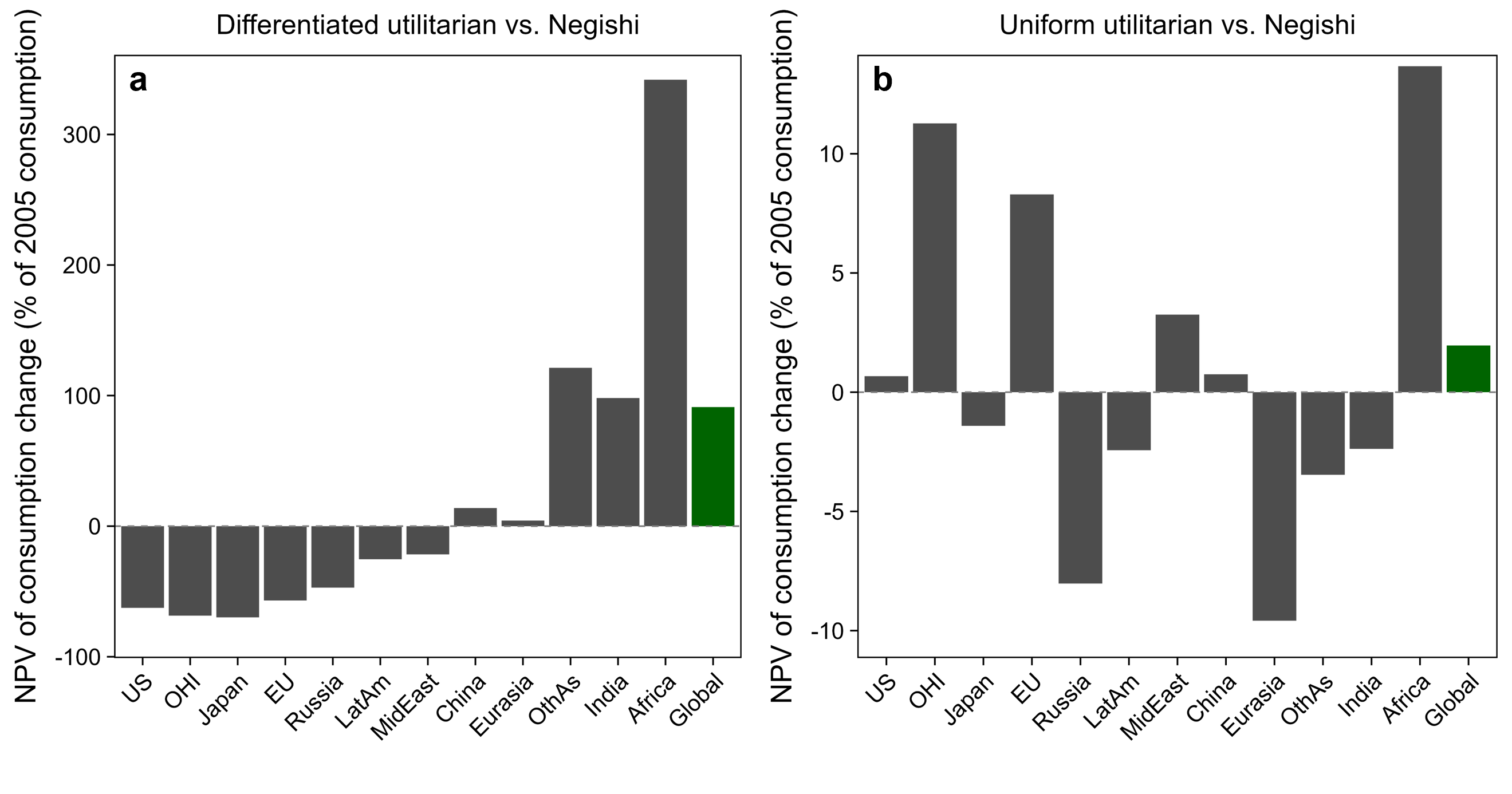}
    \caption{Net present value of consumption changes.}
    \vspace{1mm}  
    \begin{minipage}{1\textwidth}
        \small \textit{Notes}:
        The values show the welfare-equivalent consumption change in 2005, as a percentage of the consumption in 2005. The ``Global'' value expresses the global utilitarian welfare change in the welfare-equivalent consumption change in 2005 if consumption were distributed equally (for details, see Appendix \ref{Assec: Calculation of welfare-equivalent consumption changes}). The figure shows the results for the utility discount rate of 0.1\%.
    \end{minipage}
    \label{fig:C_pctchange_NPV_rho01}
\end{figure}

\begin{figure}[!htb]
    \centering
    \includegraphics[width=0.9\linewidth]{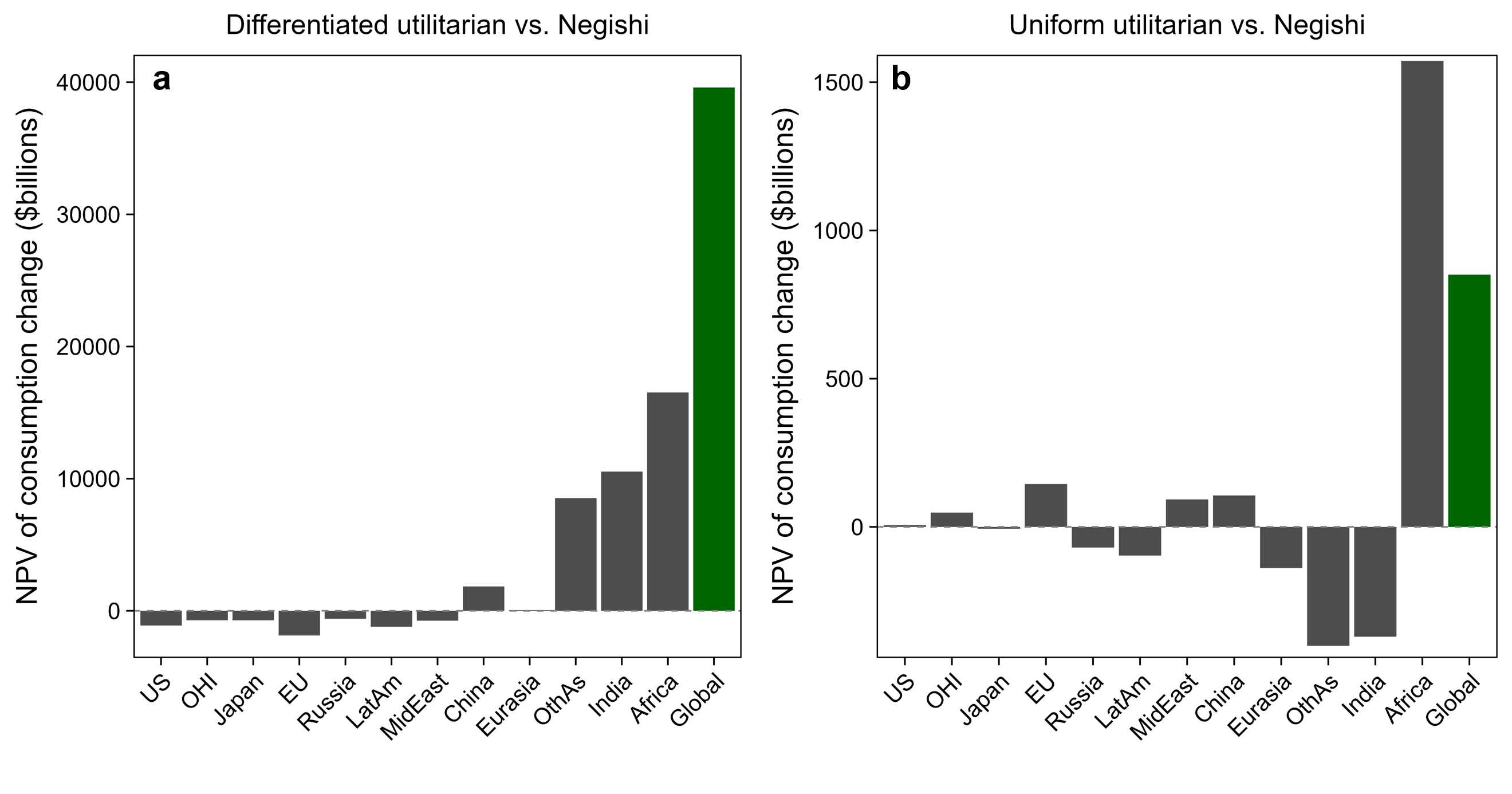}
    \caption{Utilitarian welfare changes.}
    \vspace{1mm}  
    \begin{minipage}{1\textwidth}
        \small \textit{Notes}:
        The values express the regional and global utilitarian welfare change in the welfare-equivalent consumption change in 2005 if consumption were distributed equally (for details, see Appendix \ref{Assec: Calculation of welfare-equivalent consumption changes}). The figure shows the results for the utility discount rate of 0.1\%.
    \end{minipage}
    \label{fig:C_bar_change_NPV_rho01}
\end{figure}

\begin{figure}[!htb]
    \centering
    \includegraphics[width=0.95\linewidth]{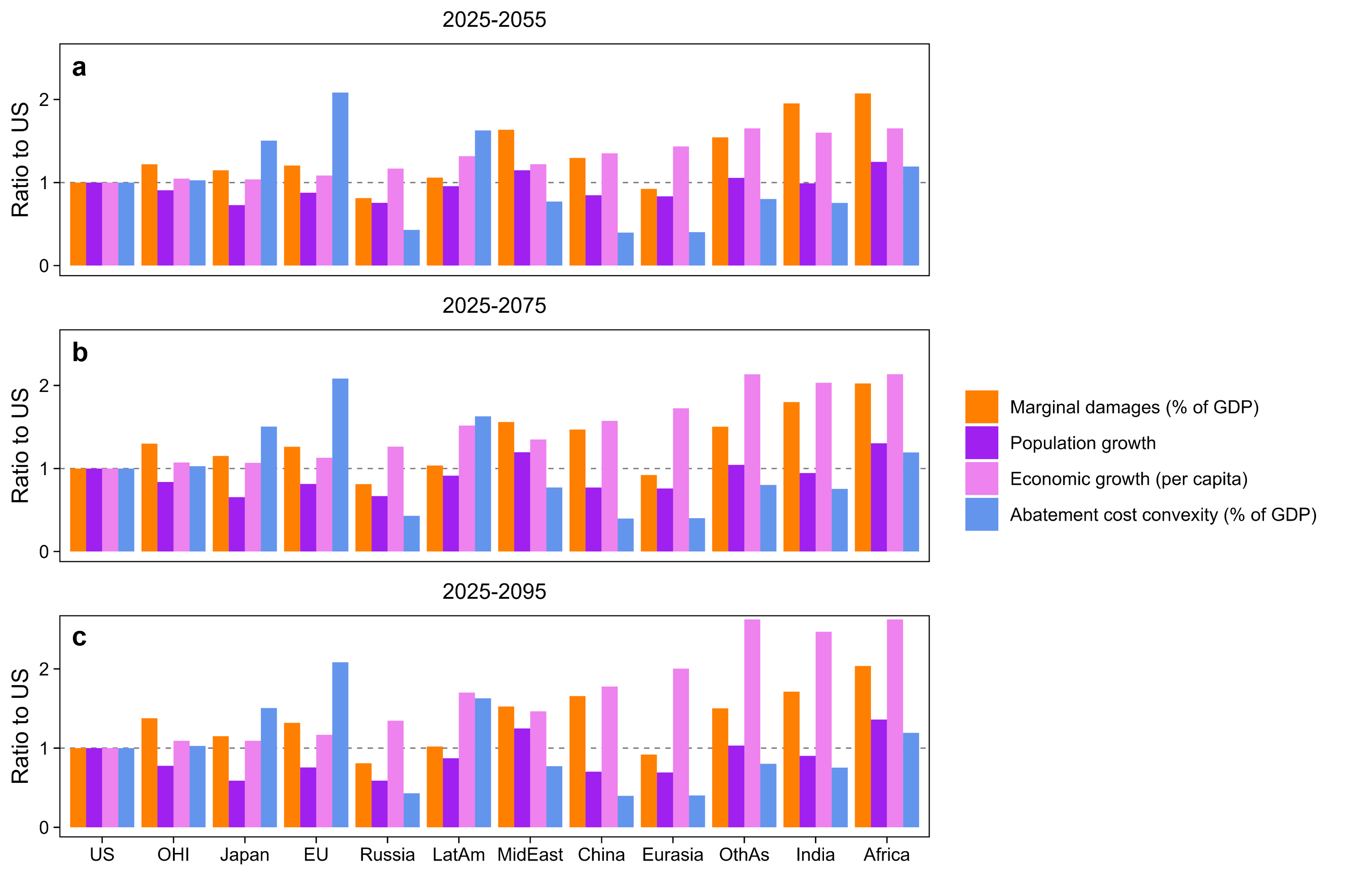} 
    \caption{Relative regional marginal damages, abatement cost convexities, population and economic growth.}
    \vspace{1mm}  
    \begin{minipage}{1\textwidth}
        \small \textit{Notes}:
        The ratio of the convexities in the abatement cost functions is $c''_{i,2025}/c''_{US,2025}$ (evaluated at uniform carbon prices).
        The ratio of the marginal damages as a percentage of GDP is $d'_{i,t}/d'_{US,t}$.
        The ratios of population growth and economic growth are given by the relative growth factors $\frac{L_{i,t}/L_{i,2025}}{L_{US,t}/L_{US,2025}}$ and $\frac{y_{i,t}/y_{i,2025}}{y_{US,t}/y_{US,2025}}$. where $y$ is the GDP per capita. The year $t$ corresponds to either 2055, 2075, or 2095 in panels (a), (b), and (c), respectively.
        The figure shows the results for the utility discount rate of 1.5\%.
    \end{minipage}
    \label{fig:ratios_growth_2025}
\end{figure}

\begin{figure}[!htb]
\centering \includegraphics[width=1\textwidth]{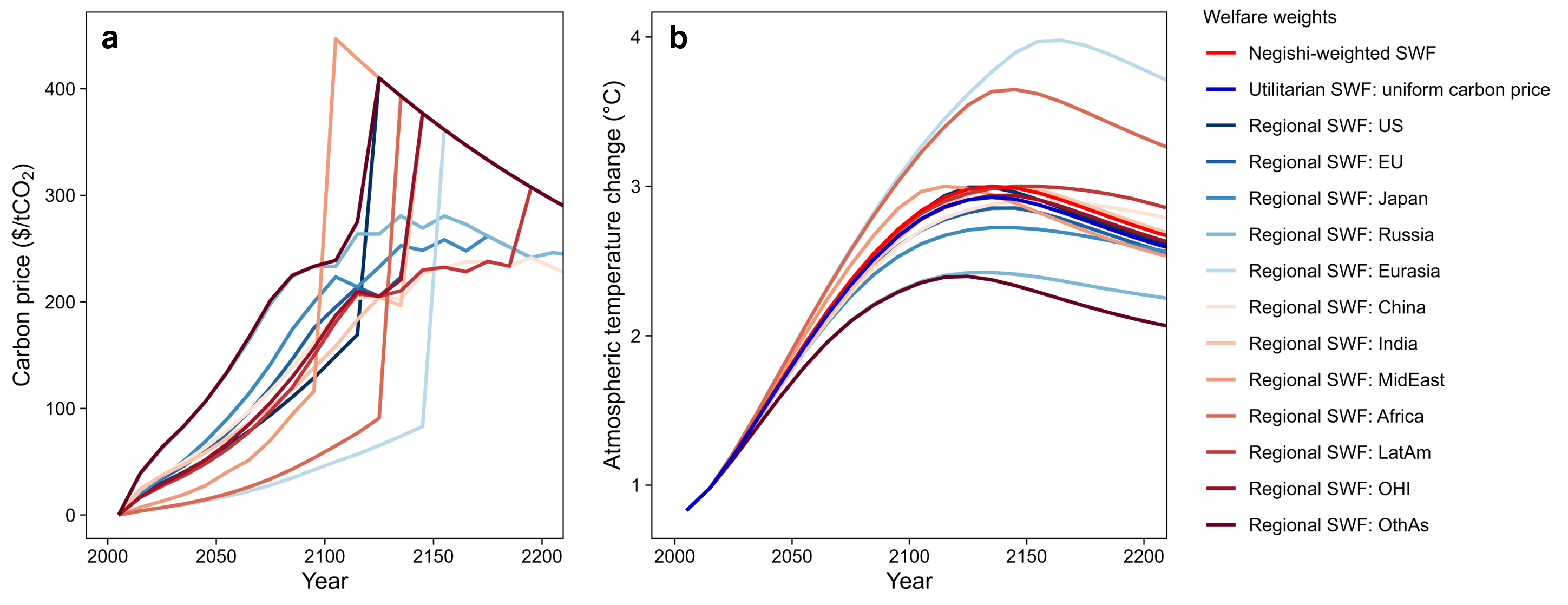}
\caption{Regions' preferred uniform carbon prices and corresponding temperature trajectories.}
\vspace{1mm}  
    \begin{minipage}{1\textwidth}
        \small \textit{Notes}:
        The figure shows the results for the utility discount rate of 1.5\%. 
        Temperature changes are relative to 1900.
        Note that Russia, Eurasia and China have the lowest backstop technology prices, causing the large carbon price increases once it is beneficial for these regions to increase the globally uniform carbon price above the level of their respective backstop prices.
    \end{minipage}
\label{fig: preferred uniform carbon prices}
\end{figure}


\clearpage

\stepcounter{section}
\section*{Appendix \thesection: Supplementary Information} \label{sec:appendixc}
\addcontentsline{toc}{section}{Appendix \thesection: Supplementary Information} 



\subsection{Optimal carbon prices for arbitrary welfare weights} \label{Assec: Optimal carbon prices for arbitrary welfare weights}

\begin{definition}
    The \textbf{optimal differentiated carbon price} (for arbitrary welfare weights) for region $i$ is implicitly defined by
    \begin{gather}\label{Def: optimal differentiated carbon price}
        \tau_i^{diff} = C'_i(A_i^*) = - \frac{1}{\alpha_i u'(x_i^*)}
        \sum_{j \in \mathcal{I}} \alpha_j u' (x_j^*) D'_j(A^*).
    \end{gather}
\end{definition}

In words, the optimal differentiated carbon price is equal to the sum of the avoided weighted marginal welfare damages divided by the weighted marginal utility. Thus, the optimal differentiated carbon price is inversely proportional to the \textit{weighted} marginal utility, $\alpha_i u_i'$. Consequently, the optimal differentiated carbon price is lower in the region with the higher weighted marginal utility. This result has first been established by \textcite{eyckmans_efficiency_1993} and \textcite{chichilnisky_who_1994}. Note that, if the weighted marginal utilities are equal across regions (i.e., $\alpha_S u'_S = \alpha_N u'_N$), we obtain the knife-edge result that the optimal ``differentiated'' carbon price is in fact uniform. This is the case if the weights are the Negishi weights. I return to this below. 

It is insightful to rearrange Equation~\eqref{Def: optimal differentiated carbon price} to
\begin{gather*}
    \alpha_i u'(x_i^*) C'^*_i = -  \sum_{j \in \mathcal{I}} \alpha_j u' (x_j^*) D'_j(A^*).
\end{gather*}

Since the right-hand side is the same for all regions, we know that $\alpha_N u'(x_N^*) C'^*_N = \alpha_S u'(x_S^*) C'^*_S$. That is, the weighted marginal welfare cost of abatement (rather than the marginal abatement cost in monetary terms) is equalized across regions.

\begin{definition}
    The \textbf{optimal uniform carbon price} (for arbitrary welfare weights) is implicitly defined by
    \begin{gather}
    \label{Def: optimal uniform carbon price}
        \tau^{uni} = C_i'(A_i^*) = - \sum_i \alpha_i u'(x_i^*) D'_i(A^*) 
        \frac{C''_S + C''_N}
        {\alpha_N u'(x_N^*) C''_S + \alpha_S u'(x_S^*) C''_N} .
    \end{gather}
\end{definition}

The optimal uniform carbon price again depends on the sum of the avoided weighted marginal welfare damages. However, it also depends on a second factor which contains the second derivatives of the abatement cost functions. To gain some intuition, we can note that the expression collapses to the expression for the optimal differentiated carbon price if one of the regions has a linear abatement cost function\footnote{
Note that I am here, for a moment, relaxing the assumption of strictly convex abatement cost functions.
}; that is, $C''_i = 0$ for one $i$.
Specifically, if the Global North has a linear abatement cost function, then the expression collapses to the differentiated carbon price expression for the Global North\footnote{
However, note that while the algebraic expression is the same as for the optimal differentiated carbon price, the values of the arguments, and thus the optimal carbon prices, are not. This is because the aggregate abatement would be different from the differentiated carbon price optimum since the optimal carbon price in both regions is given by this expression under the uniform carbon price solution.
}; and vice-versa for the Global South. The intuition is that if one region has a linear abatement cost function, and thus constant marginal abatement costs, then the only way to equalize marginal abatement costs across regions is to adjust the marginal abatement cost of the other region. Unsurprisingly, this provides the intuition that the optimal uniform carbon price lies in between the two optimal differentiated prices. Moreover, whether the uniform carbon price is closer to one or the other differentiated carbon prices depends on the relative convexities of the abatement cost functions, the welfare weights, and the relative marginal utilities at the optimal solution.

\subsection{Abatement cost and damage functions of the RICE model} \label{Assec: Abatement cost and damage functions of the RICE model}
In the RICE model, regional climate damages are given by
\begin{gather*}
    \label{E: Temp damage function RICE}
     D_{it} = Y_{it} d_{it},
\end{gather*}
where $Y_{it}$ is the GDP gross of damages and abatement costs, and $d_{it}$ denotes the climate damage as a fraction of GDP, which is composed of damages from atmospheric temperature changes and sea level rise (which are ultimately functions of emissions/abatement).

Regional abatement costs are given by
\begin{gather*}\label{E: AC function RICE}
    C_{it} = Y_{it}
    \underbrace{\frac{b_{it} \sigma_{it}}{\theta} \left(\frac{A_{it}}{\sigma_{it} Y_{it}}\right)^{\theta}}_{c_{it}},
\end{gather*}
where $b_{it}$ is the price of a backstop technology (i.e., the marginal abatement cost at which emissions can be abated completely), $\sigma_{it}$ is the baseline emissions intensity (emissions per GDP) of the economy in the absence of abatement, $\theta > 1$ is a parameter that governs the convexity of the abatement cost function (in RICE, $\theta=2.8$). Note that the abatement costs per GDP, $c_{it}$, are a function of $\frac{A_{it}}{Y_{it}}$.

The damage function from atmospheric temperature changes is shown in Figure \ref{F: Regional damage functions for atmospheric temperature changes in the RICE model.}. The trajectories of the regional baseline carbon intensities and backstop technology prices are shown in Figure \ref{F: carbon intensity and backstop prices}.

\begin{figure}[!ht]
    \centering
    \includegraphics[width=0.5\linewidth]{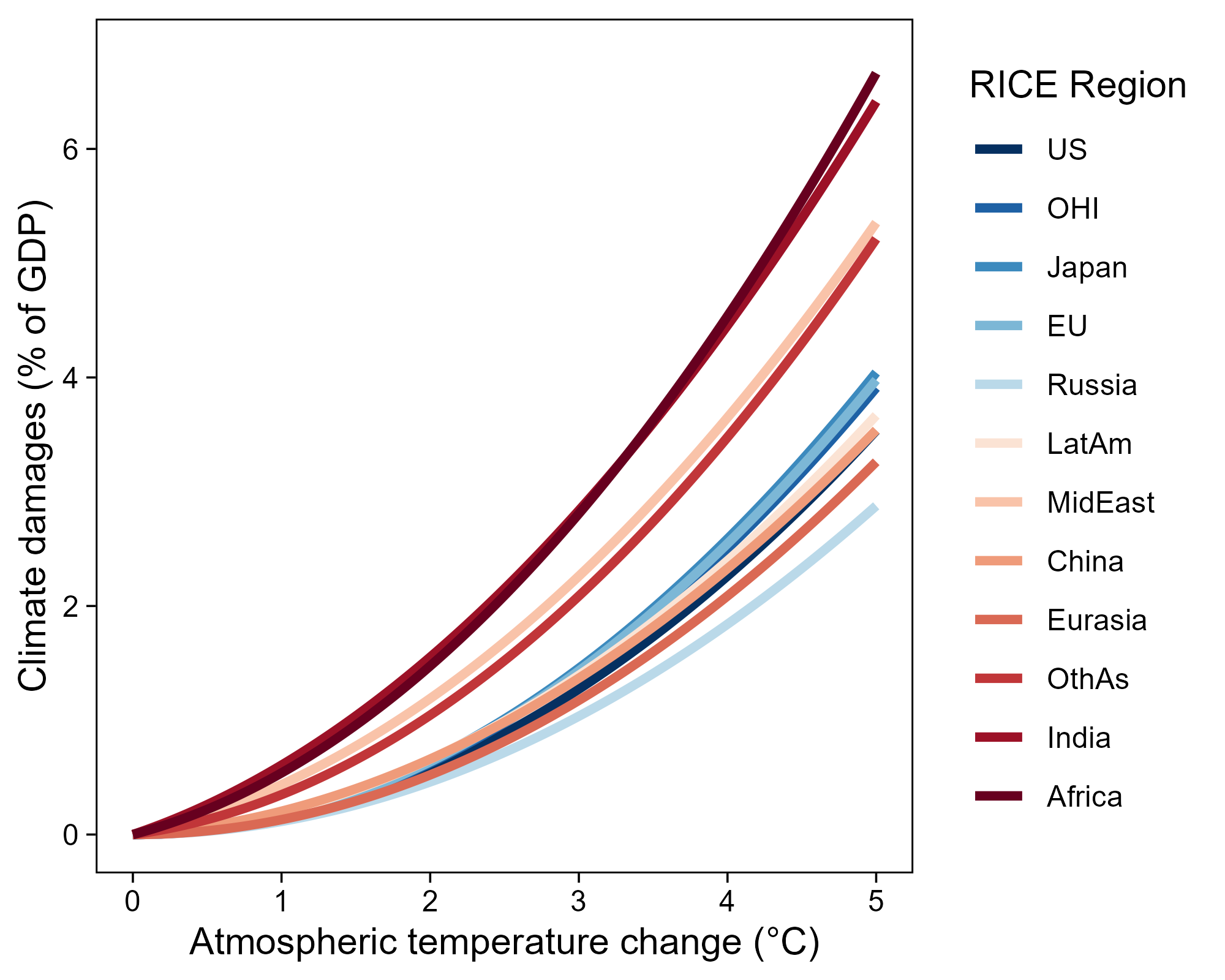}
    \caption{Regional damage functions for atmospheric temperature changes in the RICE model.}
    \vspace{1mm}  
    \begin{minipage}{1\textwidth}
        \small \textit{Notes}:
        Temperature changes are relative to temperatures in 1900.
    \end{minipage}
    \label{F: Regional damage functions for atmospheric temperature changes in the RICE model.}
\end{figure}

\begin{figure}[!ht]
    \centering
    \includegraphics[width=0.9\linewidth]{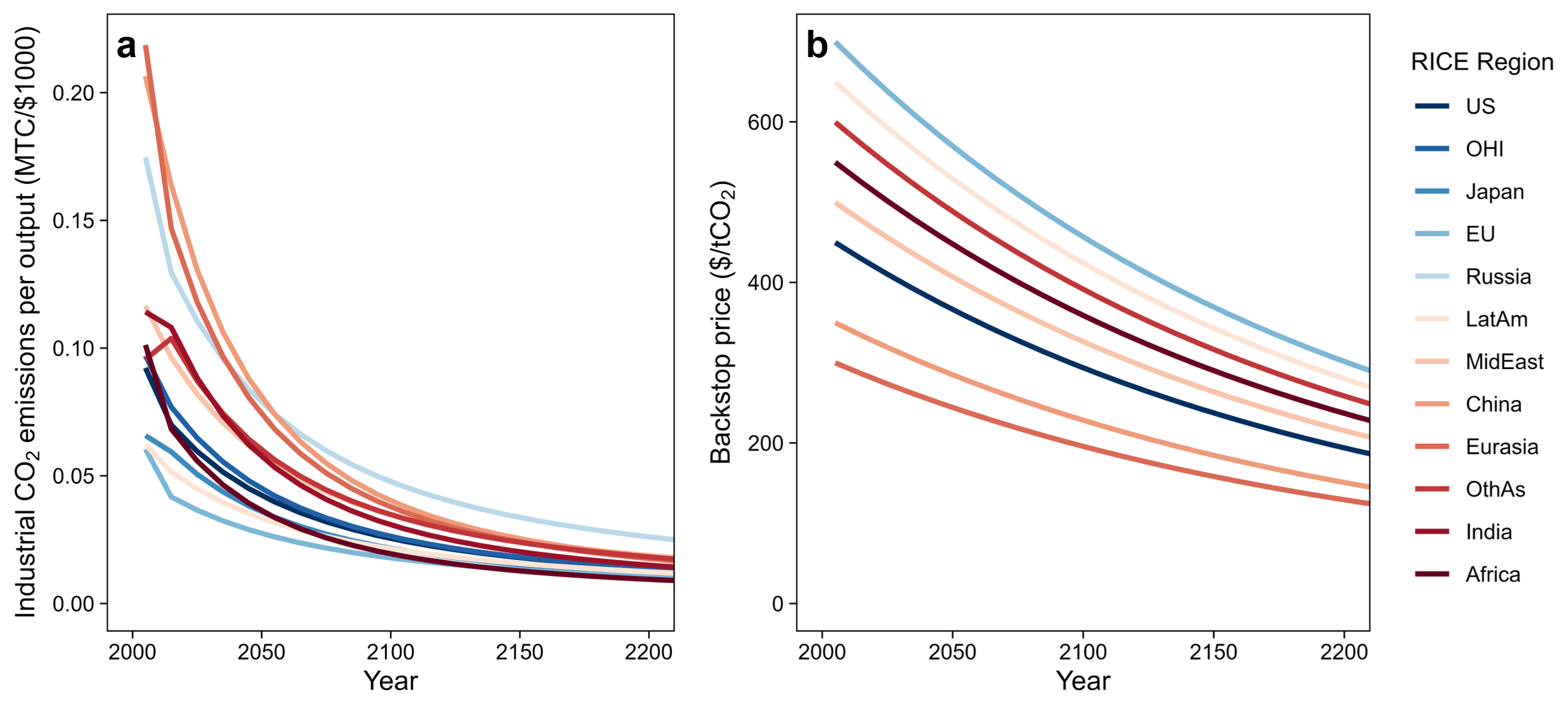}
    \caption{Regional baseline carbon intensities (a) and backstop technology prices (b) in the RICE model.}
    \vspace{1mm}  
    \begin{minipage}{1\textwidth}
        \small \textit{Notes}:
        The carbon intensity is given by the industrial CO$_2$ emissions per economic output. The backstop technology price corresponds to the marginal abatement cost at which all emissions are abated. The following regions have identical backstop prices: (1) Russia and Eurasia, (2) Other High Income (OHI) countries, Africa, and India, (3) Japan and the EU.
    \end{minipage}
    \label{F: carbon intensity and backstop prices}
\end{figure}

\subsection{Time-variant Negishi weights} \label{Assec: Time-variant Negishi weights}

The time-variant Negishi welfare weights are given by
\begin{equation*}
    \alpha_{it} = \frac{1}{u^{\prime}\left(x_{it}\right)} v_t, 
\end{equation*}
where $v_t$ is the wealth-based component of the social discount factor.
In the RICE-2010 model \parencite{nordhaus_economic_2010}, it is defined as the capital-weighted average of the regional wealth-based discount factors:
    \begin{gather*}
    \begin{aligned}
        v_t
        &= \frac{u'_{US,t}}{u'_{US, 0}}
        \sqrt{
        \frac{\sum_{i \in \mathscr{I}} \left(\frac{{\frac{u'_{US,0}}{u'_{i0}}}}{\frac{u'_{US,t}}{u'_{it}}} \frac{K_{it}}{\sum_{j \in \mathscr{I}} K_{jt}}\right)}
        {\sum_{i \in \mathscr{I}}\left( \frac{\frac{u'_{US,t}}{u'_{it}}}{\frac{u'_{US,0}}{u'_{i0}}}  \frac{K_{it}}{\sum_{j \in \mathscr{I}} K_{jt}}\right)}
        } ,
        \label{E: welfare weight}
    \end{aligned}
    \end{gather*}
where $K_{it}$ is the capital stock.

Note that $\frac{1}{u'(x_{it})} v_t$ equalizes the weighted marginal utility across regions. To obtain equalized weighted marginal utilities in each period, the discount factor needs to be equal across regions. Thus, $v_t$ is not region-specific and it pins down the wealth-based component of the world discount factor \parencite{nordhaus_warming_2000}.

\subsection{Discussion of the differentiated carbon price optimum}\label{Assec: Discussion of DCPO}

\todo[inline]{Potentially delete from Appendix (and include more concisely in the main text)}

The welfare maximizing policy that allows for differentiated carbon prices requires much higher carbon prices in rich regions than in poor regions (see Figure~\ref{F: Optimal trajectories for the carbon price and industrial emissions} and Table~\ref{T: Optimal carbon price in 2025}). This result warrants a discussion of several issues.

First, the differentiated carbon price optimum may be opposed by rich nations as it results in an implicit transfer from rich to poor regions. It should be noted, however, that the uniform carbon price optimum is welfare inferior to the differentiated carbon price optimum, as it imposes an additional constraint \parencite{budolfson_optimal_2020}. Importantly, the differentiated carbon price optimum is also in accordance with the principle of “common but differentiated responsibilities and respective capabilities” of the United Nations Framework Convention on Climate Change \parencite{unfccc_united_1992}. As such, \textcite{budolfson_optimal_2020} argue that the differentiated carbon price optimum is a natural focal point for international climate policy and for evaluating the adequacy of the nationally determined contributions (NDCs), which are at the heart of the Paris Agreement. A more recent study by \textcite{budolfson_utilitarian_2021} provides this comparison of the NDCs to implied carbon budgets under the differentiated carbon price optimum.

Second, since differentiated carbon prices are not cost-effective, it should be emphasized that a further welfare improvement over the differentiated price optimum could be achieved by establishing an international emissions trading scheme. This would allow regions with higher carbon prices to buy emission permits from poorer regions where the carbon price is lower, implying a transfer from the rich to the poor. Due to the differential carbon prices, mutual gains can be achieved by such a trading scheme \parencite{budolfson_optimal_2020}. If the permit market is fully competitive, this would result in a globally harmonized carbon price. However, as \textcite{budolfson_optimal_2020} point out, this outcome would be different from the uniform carbon price optimum discussed above, where an a priori constraint of equalized carbon prices was imposed; total emissions will be reduced and the poorest countries will bear a lower burden under the harmonized carbon price attained by the emissions trading scheme. \textcite{chichilnisky_who_1994} thus propose that the efficient allocation of emission permits is established by the differentiated carbon price optimum, and once the optimal allocation of permits is found, these permits are then traded internationally to achieve further welfare gains. The emission budgets shown in Figure~\ref{F: Optimal trajectories for the carbon price and industrial emissions} can thus be understood as providing the first step of this process.

Third, a potential problem with differentiated carbon prices is carbon leakage – an increase in carbon emissions in a country with comparatively laxer climate policies as a result of stricter climate policies in another country (e.g., due to a relocation of carbon-intensive industries to countries with laxer climate policies). The problem of carbon leakage, if it is not addressed, may thus undermine the policy. \textcite{budolfson_utilitarian_2021} provide a brief discussion of the issue of carbon leakage and how it may be addressed. They note that there are two channels for carbon leakage: (1) competitiveness differences resulting from carbon price differences, and (2) lower fossil fuel prices due to decreased global demand. \textcite{budolfson_utilitarian_2021} argue that the competitiveness channel can be addressed with border tax adjustments, such as those proposed by \textcite{flannery_framework_2018}. The second channel is shut down if countries commit to a global emissions cap \parencite{budolfson_utilitarian_2021}. Of course, there is also no carbon leakage if each region commits to its own regional carbon budget.


\clearpage

\listoftodos

\end{document}